\documentclass[sigconf, screen, nonacm]{acmart}

\newcommand\vldbdoi{XX.XX/XXX.XX}
\newcommand\vldbpages{XXX-XXX}
\newcommand\vldbvolume{19}
\newcommand\vldbissue{1}
\newcommand\vldbyear{2026}
\newcommand\vldbauthors{\authors}
\newcommand\vldbtitle{\method{}: Enabling \XOR{} in Decimal Space for Streaming Lossless Compression of Floating-point Data} 
\newcommand\vldbavailabilityurl{https://github.com/SuDIS-ZJU/DeXOR}
\newcommand\vldbpagestyle{plain}

\usepackage{graphicx}	
\usepackage{algorithm}
\usepackage{algorithmicx}
\usepackage{multirow}
\usepackage{subcaption}
\usepackage{booktabs}
\usepackage{paralist}
\usepackage{balance}
\usepackage{siunitx}
\usepackage[normalem]{ulem}
\useunder{\uline}{\ul}{}
\usepackage[noend]{algpseudocode}
\usepackage{amsmath}
\usepackage{setspace}
\usepackage{adjustbox}
\usepackage{caption}
\usepackage{xcolor}
\usepackage{mathtools}
\usepackage{pifont}
\usepackage{colortbl}
\usepackage[normalem]{ulem}
\usepackage{todonotes}

\usepackage[shortlabels]{enumitem}
\setenumerate[1]{itemsep=0pt,partopsep=0pt,parsep=\parskip,topsep=0.5pt}
\setitemize[1]{itemsep=0pt,partopsep=0pt,parsep=\parskip,topsep=0.5pt}
\setdescription{itemsep=0pt,partopsep=0pt,parsep=\parskip,topsep=0.5pt}

\usepackage[most]{tcolorbox}
\newtcolorbox{myquote}[1][]{%
    colback=black!5,
    colframe=black!5,
    notitle,
    sharp corners,
    borderline west={1pt}{0pt}{gray!80!black},
    enhanced,
    breakable,
    left=1pt,    
    right=1pt,   
    top=1pt,     
    bottom=1pt,  
    #1
}

\newtheoremstyle{custom}
  {3pt}   
  {3pt}   
  {\itshape}       
  {}       
  {\bfseries} 
  {.}      
  { }      
  {}       

\theoremstyle{custom}
\newtheorem{lemma}{\bf Lemma}
\newtheorem{definition}{\bf Definition}
\newtheorem{example}{\bf Example}
\newtheorem*{theorem*}{\bf Problem Statement}

\newcommand{\Sprintz}{\textsf{Sprintz}}
\newcommand{\Gorilla}{\textsf{Gorilla}}
\newcommand{\Chimp}{\textsf{Chimp}}
\newcommand{\Chimpp}{$\textsf{Chimp}_{128}$}
\newcommand{\Elf}{\textsf{Elf}}
\newcommand{\Elfp}{\textsf{Elf}+}
\newcommand{\ElfStar}{\textsf{Elf}$\ast$}
\newcommand{\SElfStar}{\textsf{SElf}$\ast$}
\newcommand{\PDE}{\textsf{PDE}}
\newcommand{\ALP}{\textsf{ALP}}
\newcommand{\Camel}{\textsf{Camel}}
\newcommand{\method}{\textsf{DeXOR}}

\newcommand{\XOR}{\textsc{xor}}
\newcommand{\DXOR}{\textsc{decimal xor}}

\newcommand{\problem}{{\textsc{SLC}}}

\newcommand{\zerobit}{\(\fboxsep=1pt\fbox{\texttt{0}}\)}
\newcommand{\onebit}{\(\fboxsep=1pt\fbox{\texttt{1}}\)}

\definecolor{emphbgcolor1}{HTML}{F79F9F}
\definecolor{emphbgcolor2}{HTML}{CAF4FF}
\definecolor{emphbgcolor3}{HTML}{FFD09B}

\definecolor{emphbgcolor4}{HTML}{FFC6C2}
\definecolor{emphbgcolor5}{HTML}{E0F7FA}
\definecolor{emphbgcolor6}{HTML}{FFA07A}

\definecolor{textcolorgreen}{HTML}{11AA33}
\definecolor{textcolorgray}{HTML}{993399}

\definecolor{lowdpcolor}{HTML}{DDDDFF}
\definecolor{highdpcolor}{HTML}{FFCCCC}

\newcommand{\zl}[1]{{#1}} 
\newcommand{\cy}[1]{{#1}}
\newcommand{\cyr}[1]{}
\newcommand{\lcy}[1]{{#1}}
\newcommand{\lcyr}[1]{}
\newcommand{\cyblue}[1]{{#1}}
\AtBeginDocument{%
  \providecommand\BibTeX{{%
    \normalfont B\kern-0.5em{\scshape i\kern-0.25em b}\kern-0.8em\TeX}}}

\begin{document}
\title{\method{}: Enabling \XOR{} in Decimal Space for Streaming Lossless Compression of Floating-point Data}

\author{Chuanyi Lv}
\affiliation{
  \institution{Zhejiang University, China}
}
\email{chuanyi.lv@zju.edu.cn}

\author{Huan Li}
\affiliation{
  \institution{Zhejiang University, China}
}
\email{lihuan.cs@zju.edu.cn}

\author{Dingyu Yang}
\affiliation{
  \institution{Zhejiang University, China}
}
\email{yangdingyu@zju.edu.cn}

\author{Zhongle Xie}
\affiliation{
  \institution{Zhejiang University, China}
}
\email{xiezl@zju.edu.cn}

\author{Lu Chen}
\affiliation{
  \institution{Zhejiang University, China}
}
\email{luchen@zju.edu.cn}

\author{Christian S. Jensen}
\affiliation{
  \institution{Aalborg University, Denmark}
}
\email{csj@cs.aau.dk}

\begin{abstract}
With streaming floating-point numbers being increasingly prevalent, effective and efficient compression of such data is critical. Compression schemes must be able to exploit the similarity, or smoothness, of consecutive numbers and must be able to contend with extreme conditions, such as high-precision values or the absence of smoothness.
We present \method{}, a novel framework that enables \DXOR{} procedure to encode decimal-space longest common prefixes and suffixes, achieving optimal prefix reuse and effective redundancy elimination. To ensure accurate and low-cost decompression even with binary-decimal conversion errors, \method{} incorporates 1) scaled truncation with error-tolerant rounding and 2) different bit management strategies optimized for \DXOR{}. 
Additionally, a robust exception handler enhances stability by managing floating-point exponents, maintaining high compression ratios under extreme conditions.
In evaluations across 22 datasets, \method{} surpasses state-of-the-art schemes, achieving a 15\% higher compression ratio and a 20\% faster decompression speed while maintaining a competitive compression speed. \method{} also offers scalability under varying conditions and exhibits robustness in extreme scenarios where other schemes fail.
\end{abstract}

\settopmatter{printacmref=false} 
\pagestyle{plain}               

\maketitle

\section{Introduction}
\label{sec:intro}


\begin{figure}[t]
    \centering
    \includegraphics[width=0.96\columnwidth]{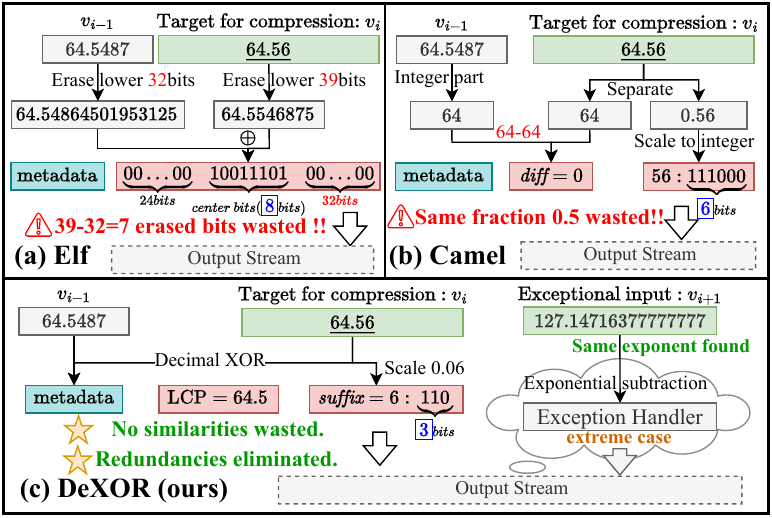}
    \caption{
    Compressing the target $\uline{64.56}$ with the previous value $64.5487$:
    (a) \Elf{}~\cite{li_elf_2023,Elf_Plus_2023} 
    removes trailing zeros in the binary representation but wastes bits after \XOR{}.
    (b) \Camel{}~\cite{camel} targets the integer part of a decimal value, leaving the fractional part unoptimized.
    (c) \method{} extracts the longest common decimal prefix ($64.5$) to exploit smoothness and scales the suffix (from $0.06$ to $6$) to remove redundancy. It also includes an exception handler based on exponential subtraction for extreme-case target values ($\uline{127.14716377\cdots}$).
    }
    \label{fig:combination}
    \vspace{-2mm}
\end{figure}

The rapidly growing Internet of Things~\cite{li2022spatial} emits massive volumes of streaming data, typically in the form of time series. For example, aviation systems emit terabytes per second~\cite{Jensen_Pedersen_Thomsen_2017}.
Across many domains, including industrial automation, finance, and scientific research, high-precision floating-point data streams require accurate decompression, as even minor inaccuracies can have critical operational consequences.
As a result, lossless compression with real-time efficiency is a fundamental and pressing challenge.

To address this challenge, \emph{streaming lossless compression} (\problem{}) has emerged as a key research focus.
The present study targets the common setting where only the previous value is required to compress the current one (Figure~\ref{fig:combination}), addressing scenarios with strict memory constraints or streaming applications where retaining earlier values is infeasible.
\zl{
Existing \problem{} approaches bifurcate into two categories:
\textit{smoothness-based} and \textit{redundancy-based} schemes.
}
Smoothness-based schemes, such as \Gorilla{}~\cite{pelkonen_gorilla_2015} and \Chimp{}~\cite{liakos_chimp_2022}, leverage the similarity among consecutive values, called \emph{smoothness}, to compress using bitwise operations (e.g., \XOR{}).
While effective for sequential patterns, these schemes struggle with the complex storage formats used for floating-point numbers.
In contrast, redundancy-based schemes target representational redundancy in floating-point structures, introducing leading or trailing zeros to minimize bit storage.
Prominent examples include scaling-to-integer schemes (\PDE{}~\cite{kuschewski_btrblocks_2023} and \ALP{}~\cite{afroozeh_alp_2023}). 
However, such schemes often do not exploit inherent smoothness in the data, causing inferior compression ratios.


\zl{Recent} proposals, including the \Elf{} series~\cite{li_elf_2023,Elf_Plus_2023,li2025adaptive,li2025serf} and \Camel{}~\cite{camel}, attempt to combine smoothness- and redundancy-based schemes (see Figure~\ref{fig:combination}(a) and (b)). 
However, this integration causes interference, where redundancy elimination involves format conversion, disrupting smoothness.
For example, in the \Elf{} series, mantissa erasure disrupts smoothness by misaligning binary representations with differing precision. In Figure~\ref{fig:combination}(a), $\uline{64.56}$ erases $39$ trailing zeros, while $64.5487$ erases $32$, resulting in $\uline{64.56}$ losing $7$ additional trailing zeros after the \XOR{}, leading to redundancy in the compressed output.
Next, \Camel{}~\cite{camel}, covered in Figure~\ref{fig:combination}(b), separates integer and fractional components, partially exploiting similarities in the integer parts ($64-64$), while reducing redundancies in the fractional part by scaling ($0.56$ to $56$ for an integer-based representation). However, this scheme does not exploit smoothness in the fractional parts, limiting its compression ratios.

Unlike schemes constrained by binary systems, 
\zl{
\Camel{} pioneers smoothness exploitation in the decimal domain, inspiring our design of \method{}, a unified framework that synergizes smoothness- and redundancy-based techniques via \emph{decimal-space conversion}.
}
Its core is the \DXOR{}, which transforms binary floating-point values into decimal-aligned ones, \zl{isolating the largest common prefix for smoothness while leaving residuals for redundancy elimination.}
%
Consider the target $\uline{64.56}$ in Figure~\ref{fig:combination}(c). 
Using \DXOR{}, the \textbf{longest common prefix (LCP)} with the previous value --- e.g., \emph{64.5} is shared between $\uline{64.56}$ and $64.5487$ --- is identified and removed.
The residual value (0.06) is then scaled into an integer suffix (6), achieving optimal redundancy elimination.

While \DXOR{} is conceptually simple, its 
\zl{processing w.r.t \problem{} workflows faces} three key challenges:
\begin{itemize}[leftmargin=*]
     \item 
     \zl{
     \textbf{C1 (Lossless Reconstruction)}: Binary-to-decimal conversion risks precision loss (see Section~\ref{ssec:problem}), compromising precise suffix extraction and violating the guarantee of \emph{lossless} reconstruction.
     }
    
    \item 
     \zl{
     \textbf{C2 (Computational Overhead)}: Prefix extraction for smoothness, as well as suffix processing for redundancy elimination, imports additional computational and metadata costs, increasing time complexity and storage demands.
     }
    
    \item 
     \zl{
     \textbf{C3 (Robustness in Extreme Cases)}: For high-precision, non-smooth datasets with little redundancy, \DXOR{} benefits may decline, just incurring increased time or space costs.
     }
    
\end{itemize}

It is essential to address these challenges, \method{} offers a suite of \zl{sophisticated}
optimizations.

%

To address Challenge C1, \method{} ensures accurate prefix and suffix identification through a \textbf{\emph{scaled truncation approach with error-tolerant rounding}} (Section~\ref{sssec:preconditions}).
Unlike traditional schemes that convert from binary to decimal formats for string matching, \method{} computes prefixes and suffixes directly from \zl{vanilla} binary format.
Moreover, an error-tolerant rounding mechanism is introduced to mitigate edge-case risks, ensuring the \DXOR{} operation is robust and lossless.
This eliminates conversion-induced prefix and suffix identification errors while additionally reducing computational overhead.

To address Challenge C2, \method{} offers several \textbf{\emph{\DXOR{}-specific bit management strategies}}:
(i) \method{} minimizes prefix storage complexity by encoding the LCP and suffix using minimal coordinate data (i.e., metadata in Figure~\ref{fig:combination}(c)), and a lightweight 2-bit case encoding further enables efficient reuse of these coordinate results (Section~\ref{sssec:postprocessing}).
(ii) By identifying the sign consistency between the suffix and prefix; and by assessing the inefficiency of variable-length suffix storage, \method{} adopts an optimal fixed-bit unsigned suffix encoding design (Section~\ref{ssec:compressor}).

To address Challenge C3, \method{} includes an \textbf{\emph{exception handler}} operating in parallel with the \DXOR{}-driven main pipeline, designed expressly to handle high-precision or non-smooth streams.
The handler isolates the floating-point exponent and encodes the exponential subtraction (see Figure~\ref{fig:combination}(c)), leveraging smoothness in the data while bypassing data components with minimal benefits (Section~\ref{ssec:ES}).
Given the bounded range, low volatility, and low-bit density of exponential subtraction, \method{} introduces compact adaptive-length encoding to efficiently store essential information (Section~\ref{ssec:elastic}).  The handler can also function as a standalone compressor for datasets known to exhibit exceptional properties.


We evaluate \method{} using 22 publicly available datasets commonly used in state-of-the-art research~\cite{pelkonen_gorilla_2015,liakos_chimp_2022,li_elf_2023,Elf_Plus_2023,camel}. 
Across all metrics, \method{} establishes itself as the new leader, achieving 15\% higher compression ratios and 20\% faster decompression speeds compared to the best-performing competitors, while maintaining comparable compression speeds.
On low-precision datasets, \method{} delivers an impressive 30\% improvement in decompression speed and 20+\% faster compression speed. It achieves the highest compression ratios across nearly all datasets and exhibits robust performance, maintaining competitive compression even in extreme cases where other schemes fail.
\cy{Consistent 15\%--20\% compression gains are \cyr{R4.W3 \\ \&D3} impactful as they offer benefits across storage cost~\cite{pelkonen_gorilla_2015}, network bandwidth~\cite{google2018}, query performance~\cite{abadi_columnstores,clickhouse}, and system sustainability~\cite{maneas2022ssd,energy-efficient}, thereby delivering substantial benefits.}

The main contributions are summarized as follows.
\begin{itemize}[leftmargin=*]
\item We revisit existing \problem{} schemes and provide insights into leveraging smoothness and redundancy with decimal-space operations. Building on these insights, we propose \method{}, a comprehensive \problem{} scheme centered around \DXOR{} (Sections~\ref{sec:preliminaries} and~\ref{sec:overall}).


\item We equip \method{}'s main pipeline with \DXOR{}-specific optimizations, including scaled truncation with error-tolerant rounding for efficient and precise prefix-suffix extraction, as well as different bit management strategies to ensure compact storage of data and metadata (Section~\ref{sec:main}).


\item We present an exception handler to complement \DXOR{} in extreme settings involving high-precision and non-smooth streams.
The handler stabilizes exponents and employs adaptive encoding, maintaining high compression ratios even where other schemes fall short (Section~\ref{sec:Exception}).

\item We report on extensive experiments involving pertinent datasets and baselines, offering evidence that \method{} is capable of superior compression ratios and decompression speeds, offers scalability under varying conditions, and achieves robustness in extreme scenarios (Section~\ref{sec:exp}).

\end{itemize}
\section{Preliminaries}
\label{sec:preliminaries}

\begin{table*}[ht]
\caption{Examples of approaches in converters (assuming the previous value $v_1 = 88.1537$ and the current value $v_2 = 88.1479$): \colorbox{emphbgcolor2}{leading zeros}, \colorbox{emphbgcolor1}{center bits}, and \colorbox{emphbgcolor3}{trailing zeros} are highlighted. CBL includes additional cost, e.g., \cy{the decimal separation method requires a mapping operation, resulting in an index cost of 4 bits, i.e., ($\text{index} = 0010$)} (see row 3 in Table~\ref{tab:converter-methods}).}
\label{tab:converter-methods}
\centering
\begin{adjustbox}{width=\textwidth}
\begin{tabular}{@{}clllc@{}}
\toprule
\textbf{Type} & \textbf{Approaches} & \textbf{Conversion} & \textbf{Binary Representation (after conversion)} & \textbf{CBL} \\ 
\midrule
- & Original IEEE754 & $\omega(88.1479)$ & \colorbox{emphbgcolor2}{0}\colorbox{emphbgcolor1}{100000001010110000010010111011100110001100011111100010100000101} 
& 63 \\ \midrule 
\multirow{2}{*}{\ding{172}} & \XOR{} (\Gorilla{}~\cite{pelkonen_gorilla_2015} and \Chimp{}~\cite{liakos_chimp_2022}) & $\omega(88.1479)\textcolor{textcolorgray}{~\oplus~\omega(88.1537)}$ & \colorbox{emphbgcolor2}{000000000000000000000000}\colorbox{emphbgcolor1}{101000010000100100001001100111000100111}\colorbox{emphbgcolor3}{0} 
& 39 \\ 
 & Decimal separation (\Camel{}~\cite{camel}) & $(88.1479 \textcolor{textcolorgray}{-88)~}\textcolor{textcolorgreen}{\%2^{-4}} \textcolor{textcolorgray}{\times 10^4}$ & \colorbox{emphbgcolor2}{0000000000000000000000000000000000000000000000000000} \colorbox{emphbgcolor1}{11100101 $\implies$ \textcolor{textcolorgreen}{(index=0010)}} & 12  \\ \midrule
\multirow{2}{*}{\ding{173}} & Erasure (\Elf{}~\cite{li_elf_2023}) & $88.147\textcolor{textcolorgray}{88818359375}$ & \colorbox{emphbgcolor2}{0}\colorbox{emphbgcolor1}{1000000010101100000100101110111}\colorbox{emphbgcolor3}{00000000000000000000000000000000} 
& 31 \\ 
 & Scaling to integers (\ALP{}~\cite{afroozeh_alp_2023}) & $88.1479 \textcolor{textcolorgray}{~\times 10^6 \times 10^{-2} = 881479}$ & \colorbox{emphbgcolor2}{00000000000000000000000000000000000000000000}\colorbox{emphbgcolor1}{11011110111111100111} & 20 \\ 
 \midrule
\ding{174} & Decimal \XOR{} (\textbf{ours}) & $(\uline {88.1479} \textcolor{textcolorgray}{~\diamond~ 88.1537) = 479}$ & \colorbox{emphbgcolor2}{0000000000000000000000000000000000000000000000000000000}\colorbox{emphbgcolor1}{111011111} & 9  \\ 
\bottomrule
\end{tabular}
\end{adjustbox}
\end{table*}

\subsection{Numerical Data Representation}
\label{ssec:numerical_data_format}

In practice, numerical values are most commonly represented as floating-point numbers rather than integers, with the \textbf{IEEE754 standard}~\cite{IEEE754_2019} being the \emph{de facto} format. 
Moreover, computer systems must store each numerical value $v_i$ using a binary representation $\omega_i$, most often in 64-bit double-precision form:
\begin{equation}\label{equ:ieee754}
f_{\text{IEEE754}} : v_i \rightarrow \omega_i =  
\underbrace{\mathit{sign}}_{1~bit}~||~
\underbrace{\mathit{exp}}_{11~bits}~||~
\underbrace{\mathit{fraction}}_{52~bits}.
\end{equation}
The original decimal value $v_i$ can be reconstructed as:
\begin{equation}\label{equ:ieee754_reconstruct}
\cy{v_i \approx \mathit{sign} \times 2^{(\mathit{exp} - \mathit{bias})} \times (1 + \mathit{fraction}),}
\end{equation}
where $\mathit{sign}$ indicates polarity; \cyr{R1.D6 \\ (1)}\cy{
$\mathit{exp}$ is the stored exponent; $\mathit{bias}$ is a predefined offset that maps negative exponents to non-negative values for compact storage (the actual exponent is $\mathit{exp}-\mathit{bias}$);} and $\mathit{fraction}$ represents the significand as a fraction, which, when added to 1, forms the complete significand value.
This format ensures compact and efficient representation across platforms.
\lcyr{M1.(1)}\lcy{
This study adopts double‐precision (with $\mathit{bias}=1023$ in Equation~\ref{equ:ieee754_reconstruct}) by default due to its higher accuracy. Nevertheless, the proposed compression techniques can be applied readily to single‐precision values, with a similar format ($8$-bit $\mathit{exp}$ and $23$-bit $\mathit{fraction}$) but $\mathit{bias}=127$.}

\begin{myquote}
\noindent\textbf{Rounding Errors in Binary-decimal Conversions}.
Limited bits in IEEE754 introduce rounding errors during binary-decimal conversions in Equation~\ref{equ:ieee754_reconstruct}. For example, when storing a decimal value $v_i = 88.1479$, it may be reconstructed as $v'_i = 88.14790000000000702\cdots$, causing a negligible error (e.g., the rounding error is below $2^{(\mathit{exp}-\mathit{bias})-52}$ for 64-bit precision). While typically insignificant, rounding errors can accumulate in typical compression schemes involving binary-decimal conversions~\cite{kuschewski_btrblocks_2023,afroozeh_alp_2023,camel}; thus, care is needed to maintain accuracy.
We discuss how to combat this issue in Section~\ref{ssec:converter}.
\end{myquote}

\subsection{Problem Formulation}
\label{ssec:problem}

Following prior studies~\cite{pelkonen_gorilla_2015,liakos_chimp_2022,li_elf_2023,Elf_Plus_2023,camel}, we assume that the compression system processes a stream of incoming numerical values, denoted as $v_1, v_2, \cdots$. \cyr{R1.W1 \\ \&D1 (1)}\cy{Such a stream may originate from a subscription to an existing time series, or it may encompass random data emitted by multiple sources.}
Timestamps, if present, are excluded from consideration in our framework, as they are typically handled separately using specialized compression techniques, such as \textit{Delta-of-delta} compression~\cite{pelkonen_gorilla_2015}. Instead, we focus exclusively on the compression of numerical values. As discussed, such input data is generally represented as bit streams, defined as follows.

\begin{definition}[Bit Stream]
\label{def:ts}
An input bit stream is denoted as $\omega_1\omega_2\cdots$, where each $\omega_i$ is the IEEE754 binary representation of a numerical value $v_i$ from the decimal space. 
\end{definition}


We define our problem as follows.

\begin{theorem*}[Streaming Lossless Compression, \problem{}]
\label{problem:tsc}
\cyr{R1.W1 \\ \&D1 (1)}\cyr{R4.W1 \\ \&D1 (1)}\cy{Given an input bit stream, without assuming temporal smoothness, streaming compression is a function $\mathcal{C}: \omega_{i-N} \cdots \uline{\omega_{i}} \rightarrow \omega'_{i}$ that transforms a target $\uline{\omega_i}$ into a compressed binary representation $\omega'_i$ using a buffer of $N$ preceding values.}
A streaming lossless compression ensures a reverse mapping $\mathcal{D}: \omega'_{i-N} \cdots \omega'_{i} \rightarrow \uline{\omega_i}$, enabling exact recovery of $\uline{\omega_i}$ from the compressed bit stream ($\omega'_i$ and its $N$ preceding ones).
\end{theorem*}


\cy{Our problem statement adopts the univariate streaming setting, thus aligning with prevailing state-of-the-art methods~\cite{pelkonen_gorilla_2015,liakos_chimp_2022,li_elf_2023,afroozeh_alp_2023}. For high-dimensional datasets (e.g., vector data\cyr{R1.W1 \\ \&D1 (2)}), each dimension can be compressed independently and assigned to a separate computational thread\cyr{R1.W3 \\ \&D3}, enabling straightforward parallelization.}

This study addresses \problem{} for $N = 1$, where compression and decompression involve only the preceding value. This minimal context renders it more challenging to achieve effective compression.     
When $N>1$, the larger buffer enables the use of sliding windows (compressing each current using its previous $N$) or truncated windows (compressing $N$ data in a mini-batch). As shown in Section~\ref{ssec:larger_buffer}, our proposal outperforms schemes using larger buffers in terms of compression ratio and/or efficiency.

\subsection{Rethinking Existing \problem{} Schemes}
\label{ssec:rethinking}

In numerical data compression, \textbf{variable-length encoding}~\cite{pelkonen_gorilla_2015} reduces storage by removing unnecessary bits, such as \textbf{leading} or \textbf{trailing zeros}, resulting from fixed assignments.
Referring to the example in Table~\ref{tab:converter-methods}, $v_2$'s IEEE754 representation $\omega(88.1479)$ has one leading zero, which can be omitted. The remaining necessary bits, the \textbf{center bits}, start and end with a \onebit{} bit.

To reconstruct the original value, auxiliary bits (i.e., metadata) are needed to record the length and (start) position of the center bits, with little opportunity to compress these bits.
Thus, reducing the \textbf{center bit length (CBL)} often becomes a bottleneck when trying to achieve higher compression ratios. 
To address this, prior approaches~\cite{blalock_sprintz_2018,liakos_chimp_2022,li_elf_2023,afroozeh_alp_2023,camel} introduce preprocessing steps, leveraging properties such as \emph{data smoothness}~\cite{blalock_sprintz_2018, liakos_chimp_2022,li_elf_2023,camel} or \emph{numerical precision redundancy}~\cite{li_elf_2023,afroozeh_alp_2023,camel}, to shorten the CBL. We refer to these preprocessing steps as \textbf{converters}, with Table~\ref{tab:converter-methods} summarizing common types.

\smallskip
\noindent\textbf{Type-\ding{172} \emph{Smoothness-based Converters}} leverage the smoothness of streaming data, where adjacent values often are similar---a property particularly prevalent in time series.
\begin{example}
    Methods like \Gorilla{}~\cite{pelkonen_gorilla_2015} and \Chimp{}~\cite{liakos_chimp_2022} use \XOR{} operations (e.g., $\omega_2 \oplus \omega_1$) to highlight similarities, yielding many leading zeros that reduce the CBL from 63 to 39 (see row~2 in Table~\ref{tab:converter-methods}).

    \Camel{}~\cite{camel} separates decimal values into integer and fractional parts. For $v_2 = 88.1479$, the integer part $88$ is subtracted from the previous integer part and stored (e.g., $88 - 88 = 0$), while the fractional part $0.1479$ is mapped into a smaller range using modulus $2^{-4}$ (corresponding to the 4 decimal places of $v_2$) and scaled by $10^4$. This reduces the center bits to 8 bits but adds a 4-bit overhead for indexing (row~3 in Table~\ref{tab:converter-methods}).
\end{example}

\smallskip
\noindent\textbf{Type-\ding{173} \emph{Redundancy-based Converters}} target IEEE754 specifically and reduce representational redundancy by utilizing knowledge of data precision.
\begin{example}
    \Elf{}~\cite{li_elf_2023} erases unnecessary mantissas, introducing more trailing zeros while preserving the original decimal precision (stored separately) that can round decimal values back ($88.14788\cdots$ in row 4 of Table~\ref{tab:converter-methods} is rounded up to $88.1479$ if 4 decimal places are known).
 
    \ALP{}~\cite{afroozeh_alp_2023} scales floating-point numbers into integers via $v' = v \times 10^e \times 10^{-f}$, where $\langle e, f \rangle$ are chosen carefully to avoid rounding errors. This scaling adds leading zeros, optimizing the CBL to 20 for $88.1479$ (row~5 in Table~\ref{tab:converter-methods}).
\end{example}

\smallskip
\noindent\textbf{Joint Utilization of Smoothness and Redundancy}.
Most \problem{} converters exploit either smoothness or redundancy, but combining their approaches introduces challenges:
\begin{itemize}[leftmargin=*]
\item \textbf{\emph{Smoothness Disruption}}: Precision-based transformations in type-\ding{173} can disrupt smoothness in type-\ding{172}. 
For instance, the \Elf{} series schemes subtract minimal values based on precision, altering binary representations and reducing similarity to prior values. Similarly, \ALP{} scales values independently, distorting coordinate systems (e.g., $66.101 \rightarrow 66101$ and $66.1 \rightarrow 661$).
\item \textbf{\emph{Partial Solutions}}: \Camel{} preserves smoothness in integer components but ignores fractional similarities (e.g., $88.1479$ and $88.1537$ share $88.1$, but only $88$ is retained). \Elf{} attempts to combine erasure (type-\ding{173}) with \XOR{} (type-\ding{172}), but inconsistencies in CBL after erasure cause inefficiencies in \XOR{}.
\lcyr{M2.(1)}\lcy{Here, smoothness is inter-value continuity. In \Elf{}, variable-length bit erasure is applied to each value based on its precision; the subsequent \XOR{} then aligns the lengths to the longer, causing smoothness loss.}

\end{itemize}

In this study, we propose \textbf{the first type-\ding{174} converter}, which combines the strengths of smoothness- and redundancy-based converters.
\lcyr{M2.(1)}
Our converter is enabled by a decimal-space \XOR{}-like operation that maximizes the exploitation of smoothness (\emph{88.1} shared between $88.1479$ and $88.1537$) while preserving the remaining portion ($88.1479-88.1 = 0.0479$) for precision-based transformation.
As shown in row 6 of Table~\ref{tab:converter-methods}, the converter, named \DXOR{}, achieves a shorter CBL than do existing \problem{} converters.
\begin{figure}[t]
    \centering
    \includegraphics[width=\columnwidth]{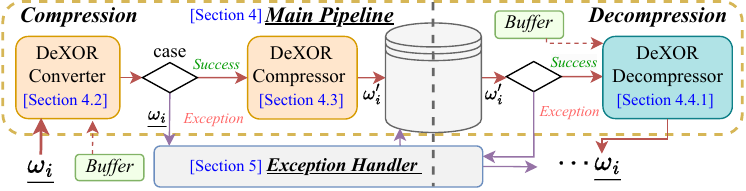}
    \caption{\cyblue{The ``main pipeline'' features a \DXOR{}-based converter and a specialized compressor/decompressor, accompanied by an ``exception handler'' for extreme cases.}}
    \label{fig:our_frame}
\end{figure}

\section{Overall Framework of \method{}}
\label{sec:overall}

The proposed framework, illustrated in \lcyr{M1.(2)}\lcy{Figure~\ref{fig:our_frame}}, compresses each floating-point number $\omega_i$ (in IEEE754 format) through a \emph{main pipeline} (Section~\ref{sec:main}) built around the \DXOR{} technique, supplemented by an \emph{exception handler} (Section~\ref{sec:Exception}) for extreme cases.

\textbf{Main Pipeline}:
The input is first processed by the \method{} converter (Section~\ref{ssec:converter}) to eliminate precision redundancy while exploiting smoothness. 
This step generates a 2-bit case code for execution status during decompression. A special case code, \onebit{}\onebit{}, triggers a switch to the exception process, while other cases proceed with metadata storage and suffix encoding. The suffix, with its significantly reduced information capacity, is then passed to the \method{} compressor (Section~\ref{ssec:compressor}) for specialized variable-length encoding.

\textbf{Exception Handler}:
\cy{Data handled poorly by \method{} (e.g., 16+ decimals) is redirected to the exception handler, in its original IEEE754 format via an \emph{in-line conditional branch}.}\cyr{R2.W2 \\ (2) 2a}
This module identifies the most stable components (e.g., exponent with high smoothness) for subtraction and applies adaptive storage of the subtraction result to reduce the storage overhead. The exception handling process is consistent with the global input structure, enabling it to perform compression independently when necessary.

The framework integrates a \DXOR{}-centered pipeline with a plug-in exception-handling module, to achieve robust performance across diverse datasets. It dynamically determines whether to process data through the main pipeline or the exception handler, optimizing compression for all cases. Thus, \cyr{R2.W1 \\ (2) 1b}\cyr{R2.W2 \\ (3) 2b}\cy{prior knowledge of a data stream can allow direct entry into the exception handler (see Section~\ref{ssec:exception_usage}), reducing the conversion overhead and further improving the compression ratio for high-precision data}.

\section{Main Pipeline in \method{}}
\label{sec:main}

\subsection{Definition of \DXOR{}}
\label{ssec:pilot_exp}

As discussed in Section~\ref{ssec:rethinking}, leveraging both smoothness in binary representations and precision redundancy in decimal values presents challenges. This inspires us to consider from another perspective: floating-point numbers \zl{also} exhibit smoothness in decimal space. This suggests that it is possible to apply \XOR{}-like operations first, to strip away the common parts in decimal space, and then handle the redundancy in the residual. Hence, we propose \textbf{\DXOR{}}, a floating-point data conversion method within a global-aligned decimal coordinate system.
It separates the handling of common decimal prefixes and the rest, addressing smoothness and redundancy independently.

A decimal value $v_{X}$ can be represented by a finite segment on its coordinate range, i.e., $v_{X} = x_{p} x_{p-1} \cdots x_{q+1} x_{q}$, where $p \in \mathbb{Z}$ is the \textbf{head coordinate} (highest decimal place) and $q \in \mathbb{Z}$ is the \textbf{tail coordinate} (lowest decimal place).
We have $v_{X} = \sum\nolimits_{j = q}^{p} x_j \times 10^j$.
The \textbf{data precision} $dp$ is defined as the length of meaningful decimal digits, i.e., $dp = p - q + 1$.
Referring to Figure~\ref{fig:DXOR}, for the target value $v = 88.1479$, we have $p = 1$, $q = -4$, and $dp = 1 - (-4) + 1 = 6$.


\begin{figure}
    \centering
    \includegraphics[width=\columnwidth]{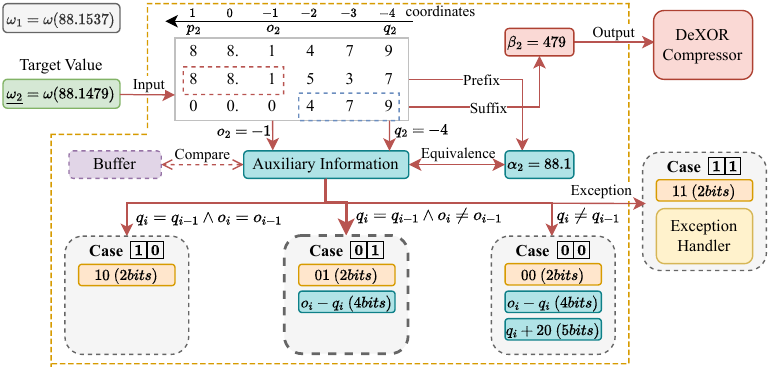}
    \caption{\cyblue{The \method{} converter keeps coordinates as auxiliary information to replace the extracted prefix, while routing the suffix to the compressor for variable-length encoding.}}
    \label{fig:DXOR}
\end{figure}

Given $\uline{v_{X}}$ the target and $v_{Y}$, represented by their decimal digits $x_j$ and $y_j$ ($j \in \mathbb{Z}$), we define the asymmetric $\DXOR{}$ as:
\begin{equation}
\label{equ:dxor}
\uline{v_{X}} \diamond v_{Y} = \Big(\sum\nolimits_{j = -\infty}^\infty \big(\uline{x_j} \diamond y_j\big) \cdot 10^j \Big) \cdot 10^{(-q)},
\end{equation}
where $x_j \diamond y_j$ denotes the \DXOR{} operation applied to the individual digits $x_j$ and $y_j$, and is defined as:
\begin{equation}
\uline{x_j} \diamond y_j = 
\begin{cases} 
0 & \text{if } \forall k \in [j,\infty) (x_k=y_k) \\
\uline{x_j} & \text{otherwise}
\end{cases}
\end{equation}

To put it simply, we have $\uline{88.1479} \diamond 88.1537 = 0.0479 \times 10^{4} = 479$. This leads to a CBL as small as 9 (see type-\ding{174} in Table~\ref{tab:converter-methods}).
To further verify its efficacy, we conduct a pilot experiment comparing different converters using six representative benchmark datasets, including both time-series and non-time-series ones (see dataset descriptions in Section~\ref{ssec:settings}). \cy{Figure~\ref{fig:pilot} reports the average CBLs after conversion. \cyr{R1.W4 \\ \&D4 (1)}By design the CBL captures the \textbf{intrinsic compression potential} of each converter and excludes any metadata overhead.}


\begin{figure}
    \centering
    \includegraphics[width=\columnwidth]{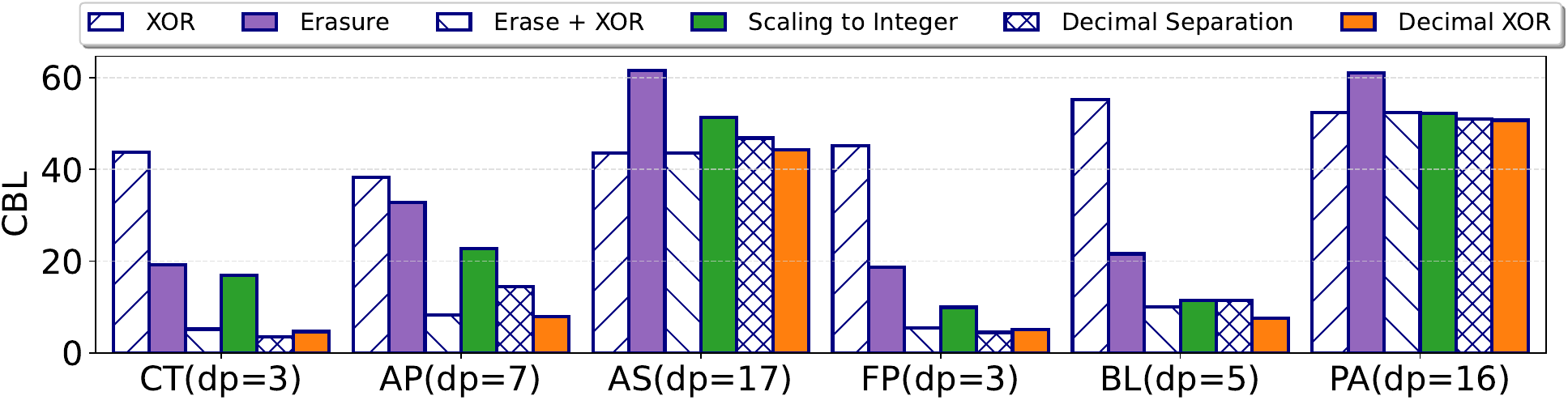}
    \caption{CBL (metadata overhead excluded) reported for \emph{time-series datasets}: \textsf{City-temp} (\textsf{CT}), \textsf{Air-pressure} (\textsf{AP}), and \textsf{Air-sensor} (\textsf{AS}), and \emph{non-time-series datasets}: \textsf{Food-Price} (\textsf{FP}), \textsf{Blockchain-tr} (\textsf{BL}), and \textsf{POI-lat} (\textsf{PA}). Datasets are sorted by data precision ($dp$) per category.}
    \label{fig:pilot}
\end{figure}

We draw the following key observations: 
1) While \XOR{} (type-\ding{172}) is insensitive to increases in $dp$, the resulting CBL remains consistently poor (CBL $\geq 38$). 
2) Methods involving precision redundancy (type-\ding{172} and type-\ding{173}) degrade in performance as $dp$ increases, with the erasure method nearly failing when $dp$ exceeds 15 bits (CBL $> 60$); the scaling-to-integers approach performs worse on time-series data compared to non-time-series data as it fails to leverage the smoothness.
3) \DXOR{} achieves the best CBL reduction in most cases; this is attributed to its successful integration of the strengths of type-\ding{172} and type-\ding{173} methods.
Interestingly, \DXOR{} performs similarly to \XOR{} in high-$dp$ scenarios (\textsf{AS} and \textsf{PA}); these exceptional cases will be analyzed and addressed in Section~\ref{sec:Exception}.

\subsection{The \method{} Converter}
\label{ssec:converter} 

While \DXOR{} offers significant advantages, implementing it in \problem{} scenarios is non-trivial. The compression system operates on binary IEEE754 data, requiring costly \emph{binary-decimal conversions} to identify similar prefixes in the decimal space. 
Moreover, conversion errors can compromise accurate prefix and suffix detection.
To resolve these issues, we propose a novel approach in Section~\ref{sssec:preconditions}, called scaled truncation with error-tolerant rounding, which enables lightweight and precise prefix-suffix identification. Furthermore, Section~\ref{sssec:postprocessing} introduces an efficient method for storing metadata associated with identified prefixes and suffixes.

\subsubsection{Prefix and Suffix Identification}
\label{sssec:preconditions} 

Let $\alpha$ be the \textbf{longest common prefix (LCP)} shared between $\uline{v_{X}}$ and $v_{Y}$.
We have $\uline{v_{X}} \diamond v_{Y} = v_{X} - \alpha = \beta'$, where $\beta'$ is a residual value.
Let $\beta$ be the suffix after \DXOR{} and we have $\beta = \beta' \times 10^{-q}$ where $q$ is the tail coordinate of $v_{X}$.
We use $o \in [p, q]$ to denote the LCP coordinate associated with $\alpha$ in $v_{X}$. 
Given the current value $\uline{v_{X}} = 88.1479$ and the previous value $v_{Y} = 88.1537$ in Figure~\ref{fig:DXOR}, we have $\alpha = 88.1$, $o = -1$, $\beta' = 0.0479$, and $\beta = 479$. 

To determine the prefix $\alpha$ and the suffix $\beta$, two coordinates are required: (1) the LCP coordinate $o$, as $\alpha = \sum\nolimits_{j = o}^{\infty} x_j \times 10^o$; and (2) the tail coordinate $q$, as $\beta = \beta' \times 10^{-q} = (v_{X} - \sum\nolimits_{j = o}^{\infty} x_j \times 10^o) \times 10^{-q}$.

However, $v_{X}$ and its decimal space coordinates are unknown as compression systems take in only binary values.
A naive method is to reconstruct a binary value $\omega_X$ into its approximate decimal representation $v'_X$ using Equation~\ref{equ:ieee754_reconstruct}, convert $v'_X$ to a \texttt{String} for prefix-suffix matching with $\texttt{String}(v'_Y)$, and then revert the \texttt{String} results back to numerical values for further processing.
This is inefficient due to the overhead of multiple type conversions, and can fail to accurately identify prefixes and suffixes due to IEEE754 conversion errors ($v'_X \approx v_X$), as discussed later in this section.

To address this issue, we propose a lightweight approach to determine $o$ and $q$ using simple scaling and truncation operators. 
The approach is formalized through the following two lemmas. 
\begin{lemma}[LCP Coordinate Judgment]
\label{theo:CP}
    \cyr{R1.D6 \\ (2)}
    \cy{
    Any common prefix $\alpha'$ between $v_{X}$ and $v_{Y}$ satisfies $\alpha' = \sum\nolimits_{j = l}^{\infty} x_j \times 10^j =$ 
    \begin{equation} \label{equ:common_prefix}
        \cy{\operatorname{trunc}}(v_{X} \times 10^{-l})  \times 10^l = \cy{\operatorname{trunc}}(v_{Y} \times 10^{-l}) \times 10^l,
    \end{equation}
    where $l \geq q$ is the ending coordinate that indicates $\alpha'$ and $\operatorname{trunc}(\cdot)$ denotes the \emph{truncation} function.
    The LCP coordinate $o$ is the smallest $l$ that satisfies this condition: 
    $\nexists~o' < o (\cy{\operatorname{trunc}}(v_{X} \times 10^{-o'} )\times 10^{o'} = \cy{\operatorname{trunc}}(v_{Y} \times 10^{-o'}) \times10^{o'})$. 
    %
    }
\end{lemma}

We omit the proof for its simplicity.
\begin{example}
    For LCP coordinate $o = -1$ of $88.1479$, $\cy{\operatorname{trunc}}(88.1479 \times 10^1) \times 10^{-1} = 88.1 = \cy{\operatorname{trunc}}(88.1537 \times 10^1) \times 10^{-1}$. However, for $o' = -2$, $\cy{\operatorname{trunc}}(88.1479 \times 10^2) \times 10^{-2} = 88.14 \neq 88.15 = \cy{\operatorname{trunc}}(88.1537 \times 10^2) \times 10^{-2}$. For $o' = 0$, the common prefix is $88$, but not the longest.
\end{example}

Lemma~\ref{theo:CP} suggests that we can iterate from the tail coordinate $q$ to find the smallest coordinate satisfying Equation~\ref{equ:common_prefix}. However, since $q$ is unknown, it must be determined first.
\begin{lemma}[Tail Coordinate Judgment]
\label{theo:end}
    A coordinate $q$ is the tail coordinate of $v_{X} (v_{X} \neq 0)$ if and only if:
    \begin{equation} \label{equ:q}
    \begin{split}
    &\bigl(\cy{\operatorname{trunc}}(v_{X} \times 10^{-q})
    = v_{X} \times 10^{-q}\bigr) \\
    &\quad \land
    \bigl(\cy{\operatorname{trunc}}(v_{X} \times 10^{-(q+1)})
    \neq v_{X} \times 10^{-(q+1)}\bigr)
\end{split}
\end{equation}
\end{lemma}
\begin{proof}
\textbf{\emph{Sufficiency}}: if $q$ is the tail coordinate, the scaled value $v_{X} \times 10^{-q}$ must be an integer, i.e., $\cy{\operatorname{trunc}}(v_{X} \times 10^{-q}) = v_{X} \times 10^{-q}$. Also, for $q+1$, the scaled value $v_{X} \times 10^{-(q+1)}$ cannot be an integer, as a smaller scaling introduces a fractional part: $\cy{\operatorname{trunc}}(v_{X} \times 10^{-(q+1)}) \neq v_{X} \times 10^{-(q+1)}$.
\textbf{\emph{Necessity}}: if $q$ satisfies the condition in Equation~\ref{equ:q}, $q$ is the rightmost coordinate such that $v_{X} \times 10^{-q}$ is an integer. Thus, $q$ must be the tail coordinate.
\end{proof}

\begin{example}
    For $v_2=88.1479$ with tail $q = -4$:
    $\cy{\operatorname{trunc}}(v_2\times 10^{-(-4)}) = 881479 = v_2 \times 10^{-(-4)}$ whereas $\cy{\operatorname{trunc}}(v_2 \times 10^{-(-3)}) = 88147 \neq 88147.9 = v_2\times 10^{-(-3)}$.
\end{example}

Lemma~\ref{theo:end} gives the condition for determining $q$, but its possible range is too large for an iterative search. 
However, regularities in input data streams can simplify this. For instance, in the \textsf{AP} dataset, adjacent values often share the same tail coordinate ($q_i = q_{i-1}$) 89\% of the time. This observation allows us to reuse the previous tail coordinate $q_{i-1}$ as the initial point for the tail coordinate search.

\smallskip
So far, we have not addressed \textbf{IEEE754 conversion errors}.
Instead of obtaining the exact values $v_{X}$ and $v_{Y}$, we work with their approximations $v'_{X}$ and $v'_{Y}$.
Moreover, scaling operations like $v'_{X} \times 10^{-j}$ (see Equations~\ref{equ:common_prefix} and~\ref{equ:q}) introduce additional errors as $10^{-j}$ is also represented as a floating-point number.
As a result, conditions like $\cy{\operatorname{trunc}}(v'_X \times 10^{-q}) = v'_X \times 10^{-q}$ in Lemma~\ref{theo:end} \cyr{R2.W3 \\ (2) 3b}\cy{may fail due to minor rounding errors.
To address this, we introduce a small tolerance $\Delta$, typically set to $10^{-6}$. If $|\cy{\operatorname{trunc}}(v'_X \times 10^{-q}) - v'_X \times 10^{-q}| < \Delta$, we treat the condition as satisfied. The same tolerance is applied to the equality condition in Lemma~\ref{theo:CP}.
}

\begin{example}
    Assume $v_2 = 88.1479$ with its reconstructed value $v'_2=88.147900000000007\cdots$.
    Scaling introduces small rounding errors, causing exact equality to fail: $v'_2 \times 10^4 = 881479.0000000001\cdots \neq 881479 = \cy{\operatorname{trunc}}(v'_2 \times 10^4)$. As a result, $v'_2.q$ may differ from the actual $v_2.q = -4$.
    However, the difference $|v'_2 \times 10^4 - \cy{\operatorname{trunc}}(v'_2 \times 10^4)| = 0.0000000001\cdots < \Delta$. By allowing this tolerance, we ensure that $v'_2.q = v_2.q = -4$.
\end{example}


The scaled truncation with error‑tolerant rounding is described in Algorithm~\ref{alg:precondition}.
Here, $q_i$ is determined via a locality‑aware heuristic starting from $q_{i-1}$, converging within two steps.
If $\lvert v'_i \times 10^{-q_{i-1}} - \cy{\operatorname{trunc}}(v'_i \times 10^{-q_{i-1}}) \rvert < \Delta$ then $q_i \geq q_{i-1}$, and a forward search is performed until Lemma~\ref{theo:end} is met. Otherwise, a backward search finds the correct $q_i$.
\lcyr{M1.(5)}
And the search for $o_i$ starts from $q_i$ and continues until the first common prefix is identified\lcyr{M1.(5)}.The longest common prefix $\alpha_i$ and the suffix $\beta_i$ are extracted for processing (lines 4--5).


\begin{algorithm}[t]
\caption{\texttt{get\_prefix\_and\_suffix}~(current value $\omega_i$, previous tail coordinate $q_{i-1}$)}
\label{alg:precondition}
\footnotesize
\begin{algorithmic}[1] 

\State $v'_i \gets$ \texttt{binary\_to\_decimal}($\omega_i$) \Comment{Equation~\ref{equ:ieee754_reconstruct}}
\State Tail coordinate $q_i \gets$ \texttt{get\_tail}($v'_i$, $q_{i-1}$) 
\State LCP coordinate $o_i \gets$ \texttt{get\_LCP}($v'_i$, $q_i$) 
\State Prefix $\alpha_i \gets [v'_i \times 10^{-o_i}] \times 10^{o_i}$
\State Suffix $\beta_i \gets \texttt{round}((v'_i - \alpha_i) \times 10^{-q_i})$
\State \Return $q_i$, $o_i$, $\alpha_i$, $\beta_i$

\end{algorithmic}
\end{algorithm}

\subsubsection{Efficient Metadata Storage}
\label{sssec:postprocessing}

Figure~\ref{fig:DXOR} depicts the procedure after identifying the LCP $\alpha_i$ and suffix $\beta_i$. The prefix $\alpha_i$ is excluded from storage, as it can be derived using the LCP coordinate $o_i$ and the previous value $v'_{i-1}$. Storing $o_i$ is far more space-efficient than storing $\alpha_i$. Conversely, $\beta_i$ is passed to the \method{} compressor for variable-length encoding (detailed in Section~\ref{ssec:compressor}).  
Note that $q_i$ must be stored to scale $\beta_i$ back to its decimal value in decompression, e.g., $479 \times 10^{-4} = 0.0479$. Thus, both $o_i$ and $q_i$ are stored as metadata using auxiliary bits.

Assuming the decimal values fall within a coordinate range: $-20 \leq q \leq p \leq 11$~\footnote{This range is sufficient in practical real-world applications such as geographical information systems and scientific computing~\cite{Jensen_Pedersen_Thomsen_2017}.}, $q_i$ and $o_i$ each require 5 bits to cover the interval of $31$. To minimize the sign overhead, an offset of $20$ is added to make $q_i$ and $o_i$ non-negative. Furthermore, instead of storing $o_i$, we store $(o_i - q_i)$, which typically falls within $[0, 15]$ and only requires 4 bits. If $(o_i - q_i) > 15$, the value $v_i$ is treated as high precision and directed to the exception handler in Section~\ref{sec:Exception}.

Importantly, $o_i$ and $q_i$ (and $(o_i - q_i)$) can often be reused without storage, as they frequently match the values of the previous element. A two-bit control code is used to indicate the following reuse cases:
\begin{itemize}[leftmargin = *]
\item \textbf{Case \onebit{}\zerobit{} ($q_i = q_{i-1} \land o_i = o_{i-1}$)}: Only 2 bits are needed to indicate reuse of both $q_{i-1}$ and $o_{i-1}$.
     
\item \textbf{Case \zerobit{}\onebit{} ($q_i = q_{i-1} \land o_i \neq o_{i-1}$)}: In addition to the control code, 4 bits are used to store $(o_i - q_i)$, while 5 bits for $q_i$ are saved.
        
\item \textbf{Case \zerobit{}\zerobit{} ($q_i \neq q_{i-1}$)}: 2 bits are used for the control code, 5 bits for $q_i$, and 4 bits for $(o_i - q_i)$.
\end{itemize}

We do not distinguish whether $o_i = o_{i-1}$ when $q_i \neq q_{i-1}$, as $(o_{i-1} - q_{i-1})$ is rarely reusable in such cases.
\lcyr{M1.(5)}
Instead, the remaining \textbf{Case \onebit{}\onebit{}} serves as the entry point for the exception handler.

\begin{example}
    Referring to Figure~\ref{fig:DXOR}, the binary stream after the \method{} converter is \zerobit{}\onebit{}\zerobit{}\zerobit{}\onebit{}\onebit{}. The first two bits are the control code, and \zerobit{}\zerobit{}\onebit{}\onebit{} indicates $o_i - q_i = -1 - (-4) = 3$. Five bits for $(q_i + 20)$ are saved due to the reuse of $(q_{i-1} + 20)$.
\end{example}

\subsection{The \method{} Compressor}
\label{ssec:compressor}

The \method{} converter extracts the suffix $\beta_i$, the most informative part of the value. Since $\beta_i$ is an integer, its binary representation often contains leading zeros. For simplicity, we refer to the binary representation of $\beta_i$ without leading zeros as the \textbf{binary suffix}. For example, the binary suffix of $\beta_2 = 479$ is \colorbox{emphbgcolor1}{111011111} (see Table~\ref{tab:converter-methods}). The \method{} compressor then reduces the storage overhead of the binary suffix using variable-length encoding.

Figure~\ref{fig:compressor}(a) illustrates the vanilla variable-length encoding: 1 bit for the sign, 6 bits for the length of the unsigned binary suffix (up to 63 bits), and the required bits for the suffix itself. For example, \colorbox{emphbgcolor1}{111011111} uses $1 + 6 + 9 = 16$ bits. To optimize this process, the \method{} compressor employs the following two strategies.

\begin{figure}
    \centering
    \includegraphics[width=\columnwidth]{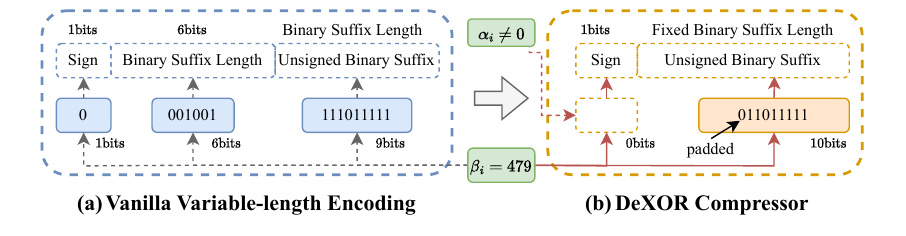}
    \caption{The optimization strategies in the \method{} compressor. The sign overhead can be waived except $\alpha_i=0$.}
    \label{fig:compressor}
\end{figure}

\subsubsection{Optimized Sign Management}

Using Lemma~\ref{theo:positive}, we reduce the overhead of sign storage:
\begin{lemma}[Sign Consistency]
\label{theo:positive}
    $\forall j, k \in \mathbb{Z}, \; x_j \cdot x_k \geq 0 \implies \alpha \cdot \beta' \geq 0 \implies \alpha \cdot \beta \geq 0$.
\end{lemma}

Lemma~\ref{theo:positive} implies that when $\alpha_i \neq 0$, the sign of $\beta_i$ matches that of $\alpha_i$. Since $\alpha_i$ is encoded via $o_i$, the sign of $\beta_i$ (and its binary suffix) can be omitted. However, in the rare case where $\alpha_i = 0$, the sign of $\beta_i$ must still be stored in 1 bit.

\subsubsection{Optimized Binary Suffix Length}
\label{sssec:binary_suffix_length}

Given $\beta_i$ as an integer, the length of its unsigned binary suffix is $\ell_i = \lceil \log_2(\operatorname{abs}(\beta_i) + 1) \rceil$ where $\operatorname{abs}(\cdot)$ denotes the absolute value.
E.g., $\beta_2 = 479$ needs $\ell_2 = 9$ bits to store the unsigned binary suffix.
Indeed, the recorded tail coordinate $q_i$ and the LCP coordinate $o_i$ can approximately determine the value $\beta_i$, such that $\operatorname{abs}(\beta_i) \in [10^{\delta-1}, 10^{\delta})$ where $\delta = o_i - q_i$.
E.g., if $\beta_2 = 479$, then $\beta_2 < 10^{(-1-(-4))} = 10^3$. 
This means the maximum required bits to represent $\beta_2$ is $\bar{\ell}_2 = \lceil \log_2(10^3) \rceil = 10$.
By using a fixed length $\bar{\ell}_i = \lceil \log_2(10^\delta) \rceil$, we avoid storing $\ell_i$ (which requires 6-bit storage as this length is up to 63 bits).
For cases where the unsigned binary suffix $\beta_i$ requires fewer bits than the fixed length, leading zeros are padded to match $\bar{\ell}_i$.

Overall, the optimized scheme incurs 4 bits to store $\delta = o_i - q_i$ (see Section~\ref{sssec:postprocessing}) and $\bar{\ell}_i$ bits for the suffix. In contrast, the original scheme requires 6 bits for $\ell_i$ and $\ell_i$ bits for the suffix. The following lemma formalizes the advantage of the fixed-length approach:

\begin{lemma}[Fixed Bit Allocation for Unsigned Binary Suffix] \label{theo:upper_bound}
    For any $\beta_i \in \mathbb{Z}$, fixed allocation of $\bar{\ell}_i$ bits achieves better compression than variable allocation of $\ell_i$, i.e., $\mathbb{E}[(4 + \bar{\ell}_i)] < \mathbb{E}[(6 + \ell_i)]$.
\end{lemma}


\lcy{
We provide a proof sketch for $\mathbb{E}[(\bar{\ell}_i - \ell_i - 2)] < 0$.
The derivation conditions on $[10^{\delta-1},10^{\delta})$.
With $2^{j-1}<10^{\delta-1}\leq2^{j}<2^{j+1}<2^{j+2}<10^{\delta}$, we have $\delta\in[(j-1)\log_{10}2+1,j\log_{10}2+1]$. 
We distinguish two cases: the probability $\mathbb{P}_1\{2^{j+3} > 10^\delta\} \approx 0.6781$, and $\mathbb{P}_2 \{2^{j+3}\leq10^{\delta}\} \approx0.3219$.\lcyr{M2.(7)}
We know that $\bar{\ell}_i$ is constant over $[10^{\delta-1},10^{\delta})$; $\ell_i$ is a piecewise function, the expectation $\mathbb{E}_1=\mathbb{E}(\bar{\ell}_i \mid 2^{j+3}>10^{\delta})- \mathbb{E}(\ell_i \mid 2^{j+3}>10^{\delta})-2 < (j+3) -(j+\frac{7}{6}) -2= -\frac{1}{6}$. Similarly, $\mathbb{E}_2  = j+2 - \mathbb{E}(\ell_i \mid 2^{j+3}\leq10^{\delta}) < -\frac{1}{7}$. 
Finally, $\mathbb{E}[(\bar{\ell}_i - \ell_i - 2)] = \mathbb{E}_1 \times \mathbb{P}_1 + \mathbb{E}_2 \times \mathbb{P}_2< -\frac{1}{6} \times \mathbb{P}_1 - \frac{1}{7} \times \mathbb{P}_2 \approx -0.159 < 0$.
}

\begin{example}
    Referring to Figure~\ref{fig:compressor}, the \method{} compressor reduces the storage for the suffix $\beta_2 = 479$ from 16 bits (vanilla design: 1 bit for sign + 6 bits for $\ell_2 = 9$ and $\ell_2 = 9$ bits for $\beta_2$) to 
    10 bits (0 bit for $\alpha_i \neq 0$ + $\bar{\ell}_2 = 10$ fixed bits allocated for $\beta_2$ in Figure~\ref{fig:compressor}(b)).
    Even after accounting for the 4 bits overhead for $\delta$ already recorded in the \method{} converter, the optimization yields a net improvement of 2 bits.
\end{example}


\subsection{Accelerated Decompression}
\label{ssec:decomp}

The decompression process extracts and parses segments from the compressed bit stream. For each new value, the \method{} decompressor reads a two-bit control code(see Section~\ref{ssec:converter}). If the exception handler is not triggered, $q_i$ and $o_i$ are reused or derived from the metadata bits: $q_i$ scales back the original suffix $\beta'_i$, while $o_i$ identifies the LCP $\alpha_i$. The final value is reconstructed by combining $\alpha_i$ and $\beta'_i$.
\lcyr{M2.(2)}\lcy{
When the control code is \onebit{}\zerobit{} and $q_i = q_{i-1} \land o_i = o_{i-1}$, we simply reuse the previous LCP ($\alpha_i = \alpha_{i-1}$), avoiding additional computation and accelerating the decompression. In Section~\ref{ssec:fundamental}, we examine the efficiency of this scheme.
}

\section{The \method{} Exception Handler}
\label{sec:Exception}

Referring back to Figure~\ref{fig:pilot}, all converters, including \DXOR{}, perform poorly on high-precision datasets (e.g., \textsf{AS} and \textsf{PA}), often matching or underperforming compared to smoothness-based \XOR{} converters. This is because high-precision data reduces precision redundancy, weakening the performance of redundancy-based techniques, including \DXOR{}.

\begin{example}
    High-precision datasets exhibit minimal smoothness and precision redundancy, leading to longer suffixes in \method{} main pipeline. E.g., with $\delta = o_i - q_i = 16$, the required suffix length $\bar{\ell}_i$ is 54 bits, exceeding the 52 bits used by the $\mathit{fraction}$ in the IEEE754 format (see Equation~\ref{equ:ieee754}). 
    \lcyr{M1.(3)}\lcy{Such data leaves precision-based methods such as \Elf{} with minimal opportunities for erasure and pushes \ALP{} to overflow.}
\end{example} 

While \XOR{} converters rely solely on smoothness and remain relatively robust, their performance is still suboptimal for high-precision data (CBL $\geq 38$) due to limited inherent smoothness. To address this, the proposed exception handler focuses on enhancing smoothness utilization specifically for high-precision cases.
As shown in Figure~\ref{fig:exception}(a), the exception handler targets the exponent (\textit{exp}) in the IEEE754 format. It performs subtraction on adjacent exponents, $\mathit{exp}_i - \mathit{exp}_{i-1}$ (red line), to enable data reuse. The \textit{sign} and \textit{fraction} fields are retained as-is (gray lines) due to their limited smoothness.  
To optimize storage further, an adaptive storage mechanism dynamically adjusts storage space for the exponential subtraction based on the streaming data. 
The subtraction and adaptive storage are detailed in Sections~\ref{ssec:ES} and~\ref{ssec:elastic}, respectively.

\begin{figure}
    \centering
    \includegraphics[width=\columnwidth]{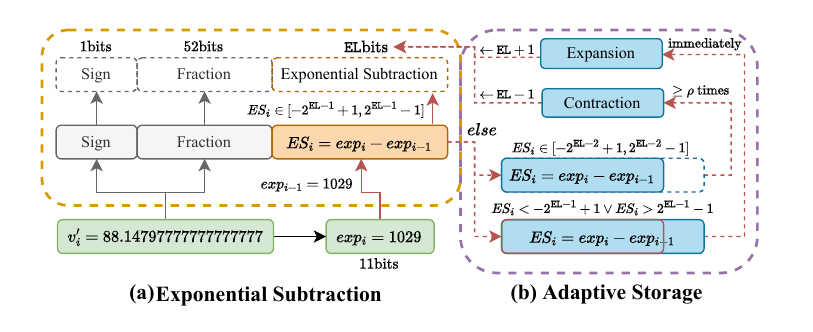}
    \caption{The data flow of the exception handler ($\mathit{exp}_i = 1029$ with the predefined $\mathit{bias} = 1023$); $\mathtt{EL}$ is an adjustable variable.}
    \label{fig:exception}
\end{figure}

\subsection{Exponential Subtraction}
\label{ssec:ES}

In the IEEE754 representation (see Table~\ref{tab:converter-methods}), the fraction part $\mathit{fraction}$ typically lacks smoothness, while the exponent $\mathit{exp}$ contains more stable and meaningful bits.

We summarize the key properties of the exponent stored in the IEEE754 representation:
(1) \emph{wide range of shared values}: A single exponent value represents all real numbers within a given range. For example, if $\mathit{exp} - \mathit{bias} = 6$ (see Equation~\ref{equ:ieee754_reconstruct}), it represents all values in $[2^6, 2^7) = [64, 128)$, regardless of precision.
(2) \emph{low volatility}: exponent values exhibit infrequent changes along the stream.
(3) \emph{dense lower bits}: the lower bits of an exponent's binary representation are compact because exponents are all integers.


Given these properties, the exponent is ideal for compression and forms the foundation for the adaptive storage strategy detailed in Section~\ref{ssec:elastic}. Subtraction is used instead of \XOR{} between consecutive exponents, as subtraction often yields shorter CBLs in high-precision datasets.
\lcyr{M1.(5)}
%
For high-precision data, the IEEE754 sign bit and fraction are stored in their original binary form, while compression focuses exclusively on the exponent, the most cost-effective component for capturing smoothness.

From these findings, we define \textbf{exponential subtraction} as $\mathit{ES}_i = \mathit{exp}_i - \mathit{exp}_{i-1}$. This result is then processed by the adaptive storage module. 
The required length to store $\mathit{ES}_i$ is defined as:
$\ell^e_i = 1 + \lceil log_2(\operatorname{abs}(\mathit{ES}_i) +1) \rceil$.
Unlike $\ell_i$ defined in Section~\ref{sssec:binary_suffix_length}, $\ell^e_i$ accounts for 1 extra bit to handle the sign, as $\mathit{ES}_i$ can be negative.

\begin{example}
    Consider the high-precision number $88.147977\cdots$ (11 trailing sevens) in Figure~\ref{fig:exception}, where $\mathit{exp}_i - \mathit{bias} = 6$.
    Typically, the preceding value $\mathit{exp}_{i-1}$ equals $\mathit{exp}_i$, resulting in $\mathit{ES}_i = 0$. This yields a required bit length of $\ell^e_i = 1 + \log_2(\lceil 0 + 1 \rceil) = 1$ for $\mathit{ES}_i = 0$.
    In contrast, encoding $\mathit{exp}_i$ with variable-length encoding requires 4 bits (1 bit for the sign and 3 bits for $(\mathit{exp}_i - \mathit{bias}) = 6$).
\end{example}

\subsection{Adaptive Storage of Subtraction Results}
\label{ssec:elastic}

Compressing only the exponent offers limited savings, as the exponent requires just 11 bits. Additionally, a 4-bit metadata overhead is needed to represent the length of $\mathit{ES}_i$, which ranges from 1 to 11 bits, further reducing compression ratios.

\begin{example}
\label{example:store_ES}
    Consider compressing three exponential subtraction results: $\mathit{ES}_1 = 3$, $\mathit{ES}_2 = 1$, and $\mathit{ES}_3 = 1$. They originate from four distinct exceptional values, which may not be contiguous in the original sequence. Their required lengths are $\ell^e_1 = 3$ and $\ell^e_2 = \ell^e_3 = 2$.
    With a fixed metadata approach, 4 bits are needed for each $\ell^e_i$, plus 53 bits for the residual sign and fraction, resulting in a fixed overhead of 57 bits per entry.
    The total storage cost is: $\ell^e_1 + \ell^e_2 + \ell^e_3 + 57 \times 3 = 178$ bits.
\end{example}

We propose an adaptive storage mechanism that dynamically adjusts the storage length $\mathtt{EL}$ based on historical data, fitting $\mathit{ES}_i$ \emph{elastically}. 
The storage length $\mathtt{EL}$ is maintained as a local variable in the main memory during compression and decompression.
$\mathtt{EL}$ varies from 1 to 12 bits, sufficient to cover any $\mathit{ES}_i \in [-2047, 2047]$\footnote{This range is required for IEEE754 \texttt{double} format, where $\mathit{exp}_i \in [0, 2047]$~\cite{IEEE754_2019}.}.

Given the current $\mathtt{EL}$, the range $[0, 2^{\mathtt{EL}} - 1]$ is used to store integers (i.e., $\mathit{ES}_i$), with the maximum representable value $2^{\mathtt{EL}} - 1$ reserved to indicate overflow (described below), reducing the usable range by 1.
Since $\mathit{ES}_i$ can be negative, a bias of $2^{\mathtt{EL}-1} - 1$ is added to offset the sign of $\mathit{ES}_i$. 
The actual stored value becomes: $\mathit{ES}_i + 2^{\mathtt{EL}-1} - 1$, and the valid range of the original $\mathit{ES}_i$ is $[-(2^{\mathtt{EL}-1} - 1), 2^{\mathtt{EL}-1} - 1]$.

After storing the current $\mathit{ES}_i$, $\mathtt{EL}$ is dynamically adjusted for the next $\mathit{ES}_{i+1}$ using two operations, as shown in Figure~\ref{fig:exception}(b):
\begin{itemize}[leftmargin=*]
\item \textbf{Expansion}: If $\mathit{ES}_i$ is outside the valid range defined by $\mathtt{EL}$, the current $\mathtt{EL}$-length space is filled with \onebit{}s, representing the maximum representable value $2^{\mathtt{EL}} - 1$. 
This is an \textbf{overflow} case that triggers the storage of the 64-bit IEEE754 representation. 
In response, $\mathtt{EL}$ is increased by 1 for the next entry $\mathit{ES}_{i+1}$ (hopefully addressing potential overflow).
    
\item \textbf{Contraction}: If $\mathit{ES}_i$ consistently fits within a smaller range ($\mathit{ES}_i \in [-(2^{\mathtt{EL}-2}-1), 2^{\mathtt{EL}-2}-1]$) for more than $\rho$ consecutive times (with interruptions resetting the count), $\mathtt{EL}$ is reduced by 1 (but never below 1), as a smaller $\mathtt{EL}$ is likely sufficient.
A small $\rho$ risks premature contraction, while a large $\rho$ delays adjustment. We use $\rho = 8$ by default, and Section~\ref{ssec:Parameter} analyzes its impact on high-precision datasets. Note that $\mathtt{EL}$ remains unchanged if the $\rho$-parameterized condition is not met.
\end{itemize}

\begin{example}\label{example:store_ES_adaptive}
    Continuing from Example~\ref{example:store_ES}, using adaptive storage, we attempt to store $\mathit{ES}_1 = 3$ with an initial length of $\mathtt{EL} = 1$.
    Since $3 > 2^{1}-1$, it triggers an overflow condition (all bits in $\mathtt{EL}$ are \onebit{}).
    The total storage required is 65 bits (1 for $\mathtt{EL}$ to indicate overflow + 64-bit original IEEE754 representation).
    
    After the overflow, $\mathtt{EL}$ is expanded to 2, which is sufficient to store $\mathit{ES}_2 \leq 2^{2-1}-1$. 
    This results in a storage cost of 55 bits ($\mathtt{EL} = 2$ to store $\mathit{ES}_2$ + 53 bits for residual sign and fraction).
    Similarly, $\mathit{ES}_3$ also requires 55 bits, as no further expansion is needed. The total storage for these three results is $65 + 55 + 55 = 175$ bits.  
\end{example}

Although this method requires a warm-up phase, $\mathtt{EL}$ quickly stabilizes at the optimal value ($\mathtt{EL} \rightarrow \ell^e_i$), allowing the scheme to approach the theoretical compression limit without storing any fixed overhead. Additionally, the gain of Example~\ref{example:store_ES_adaptive} over Example~\ref{example:store_ES} becomes more pronounced as the overhead is amortized across more $\mathit{ES}_i$.

The only drawback is that the maximum compression gain is capped at 11 bits. However, in most extreme cases with $\delta = 16$, the main pipeline requires at least 56 bits, increasing with precision. In contrast, the adaptive storage consistently achieves 56-bit storage (the pipeline's lower bound) regardless of precision, even for high-precision data.
Given this, our exception handling process is triggered when $\delta = o_i - q_i > 15$, with $q_i$ and $o_i$ provided by the converter (see Section~\ref{ssec:converter}), for maximum compression ratios.

\subsection{Extended Use of the Exception Handler}
\label{ssec:exception_usage}

The exception handler not only addresses inefficiencies in the main pipeline but also handles exceptions caused by rounding errors in binary-decimal conversions (see Section~\ref{ssec:problem}).
In Section~\ref{ssec:converter}, we introduce a tolerance for errors within $\Delta = 10^{-6}$, eliminating most rounding inaccuracies. However, two exceptions remain:
\begin{enumerate}[leftmargin=*]
    \item \textbf{Misclassification of Correct Results}: For example, a value $19.0000005$ with a very small residual may be incorrectly treated as $19$ when calculating the tail coordinate $q$, resulting in $q = 0$ instead of its true value $q = -7$. This causes lower coordinate bits to be discarded, and the decompressed value becomes $19$.

    \item \textbf{Rounding Errors During Decompression}: When reconstructing values using $v'_i = \alpha_i + \beta_i \times 10^{q_i}$, rounding errors, from the scaling operation, may occur even if $\alpha_i$, $\beta_i$, and $q_i$ are correct, leading to slight discrepancies from the original data.
\end{enumerate}

While these errors are extremely rare (less than 0.01\% in benchmark datasets), they still require special handling.
The exception handler operates entirely in binary, ensuring reliability and enabling full restoration of the original binary representation. Though it cannot match the high compression ratio of the main pipeline, it guarantees decompression accuracy with negligible impact on overall compression due to the rarity of anomalies.

The corresponding condition is given as $v'_i \neq \alpha_i+\beta_i \times 10^{q_i}$. 
This is combined with the condition $o_i - q_i > 15$ (Section~\ref{ssec:elastic}) to immediately identify exceptional cases after Algorithm~\ref{alg:precondition}.

Typically, dataset precision is known and consistent, making exception cases predictable. If the precision is known or consecutive exceptions exceed a threshold, we skip storing the case code and default to using the exception handler alone for compression.
\lcyr{M2.(4)}\lcy{
This trims a few percent of bits and nearly doubles the compression speed on high-precision, non-time-series datasets, yet the full \method{} still achieves the best compression on time-series sets.  
As this ``prior-knowledge'' mode yields an unfair advantage, Section~\ref{sec:exp} disables it and reports only on the precision-agnostic configuration.}
\begin{table*}[t]
\caption{Overall comparisons of ACB ($\downarrow$), compression speed (unit: MB/s, $\uparrow$), and decompression speed (unit: MB/s, $\uparrow$). Data precision $\mathit{dp} = p - q + 1$ reflects the average information per value. Datasets are divided into low-$\mathit{dp}$ ($\mathit{dp} \leq 7$) (marked in \textbf{\colorbox{lowdpcolor}{blue}}) and high-$\mathit{dp}$ (in \textbf{\colorbox{highdpcolor}{red}}). \Camel{}~\cite{camel} works losslessly only in low-$\mathit{dp}$ datasets.}
\label{tab:ACB}
\centering
\setlength{\tabcolsep}{1pt} 
\renewcommand{\arraystretch}{0.8} 
\begin{adjustbox}{width=\textwidth} 
\begin{tabular}{@{}cccccccccccccccccccccccccc@{}}
\hline
\multicolumn{2}{|c|}{} & \multicolumn{14}{c|}{\textbf{Time-Series Datasets with Ascending $\mathit{dp}$}} & \multicolumn{8}{c|}{\textbf{Non-Time-Series Datasets with Ascending $\mathit{dp}$}} & \multicolumn{2}{c|}{\textbf{GEOMEAN}} \\ \cline{3-26} 
\multicolumn{2}{|c|}{\multirow{-2}{*}{Datasets}} & \multicolumn{1}{c|}{\cellcolor[HTML]{DDDDFF}{\textsf{WS}}} & \multicolumn{1}{c|}{\cellcolor[HTML]{DDDDFF}{\textsf{PM}}} & \multicolumn{1}{c|}{\cellcolor[HTML]{DDDDFF}{\textsf{CT}}} & \multicolumn{1}{c|}{\cellcolor[HTML]{DDDDFF}{\textsf{IR}}} & \multicolumn{1}{c|}{\cellcolor[HTML]{DDDDFF}{\textsf{DPT}}} & \multicolumn{1}{c|}{\cellcolor[HTML]{DDDDFF}{\textsf{SUSA}}} & \multicolumn{1}{c|}{\cellcolor[HTML]{DDDDFF}{\textsf{SUK}}} & \multicolumn{1}{c|}{\cellcolor[HTML]{DDDDFF}{\textsf{SDE}}} & \multicolumn{1}{c|}{\cellcolor[HTML]{FFCCCC}{\textsf{AP}}} & \multicolumn{1}{c|}{\cellcolor[HTML]{FFCCCC}{\textsf{BM}}} & \multicolumn{1}{c|}{\cellcolor[HTML]{FFCCCC}{\textsf{BW}}} & \multicolumn{1}{c|}{\cellcolor[HTML]{FFCCCC}{\textsf{BT}}} & \multicolumn{1}{c|}{\cellcolor[HTML]{FFCCCC}{\textsf{BP}}} & \multicolumn{1}{c|}{\cellcolor[HTML]{FFCCCC}{\textsf{AS}}} & \multicolumn{1}{c|}{\cellcolor[HTML]{DDDDFF}{\textsf{FP}}} & \multicolumn{1}{c|}{\cellcolor[HTML]{DDDDFF}{\textsf{EVC}}} & \multicolumn{1}{c|}{\cellcolor[HTML]{DDDDFF}{\textsf{SSD}}} & \multicolumn{1}{c|}{\cellcolor[HTML]{DDDDFF}{\textsf{BL}}} & \multicolumn{1}{c|}{\cellcolor[HTML]{DDDDFF}{\textsf{CA}}} & \multicolumn{1}{c|}{\cellcolor[HTML]{DDDDFF}{\textsf{CO}}} & \multicolumn{1}{c|}{\cellcolor[HTML]{FFCCCC}{\textsf{PA}}} & \multicolumn{1}{c|}{\cellcolor[HTML]{FFCCCC}{\textsf{PO}}} & \multicolumn{1}{c|}{FULL} & \multicolumn{1}{c|}{low-$dp$} \\ \hline
\multicolumn{26}{c}{\vspace{-1em}} \\ \hline
\multicolumn{1}{|c|}{} & \multicolumn{1}{c|}{\Gorilla{}} & \multicolumn{1}{c|}{52.77} & \multicolumn{1}{c|}{32.12} & \multicolumn{1}{c|}{46.34} & \multicolumn{1}{c|}{39.73} & \multicolumn{1}{c|}{53.71} & \multicolumn{1}{c|}{43.85} & \multicolumn{1}{c|}{35.65} & \multicolumn{1}{c|}{46.18} & \multicolumn{1}{c|}{45.93} & \multicolumn{1}{c|}{49.28} & \multicolumn{1}{c|}{59.79} & \multicolumn{1}{c|}{58.71} & \multicolumn{1}{c|}{53.10} & \multicolumn{1}{c|}{{\ul 53.06}} & \multicolumn{1}{c|}{32.26} & \multicolumn{1}{c|}{59.33} & \multicolumn{1}{c|}{37.24} & \multicolumn{1}{c|}{44.82} & \multicolumn{1}{c|}{63.83} & \multicolumn{1}{c|}{68.02} & \multicolumn{1}{c|}{61.94} & \multicolumn{1}{c|}{68.31} & \multicolumn{1}{c|}{49.09} & \multicolumn{1}{c|}{46.82} \\
\multicolumn{1}{|c|}{} & \multicolumn{1}{c|}{\Chimp} & \multicolumn{1}{c|}{51.94} & \multicolumn{1}{c|}{33.31} & \multicolumn{1}{c|}{45.75} & \multicolumn{1}{c|}{43.36} & \multicolumn{1}{c|}{56.69} & \multicolumn{1}{c|}{49.52} & \multicolumn{1}{c|}{43.31} & \multicolumn{1}{c|}{52.18} & \multicolumn{1}{c|}{56.76} & \multicolumn{1}{c|}{54.58} & \multicolumn{1}{c|}{60.13} & \multicolumn{1}{c|}{59.96} & \multicolumn{1}{c|}{58.59} & \multicolumn{1}{c|}{58.33} & \multicolumn{1}{c|}{34.05} & \multicolumn{1}{c|}{57.99} & \multicolumn{1}{c|}{34.99} & \multicolumn{1}{c|}{44.57} & \multicolumn{1}{c|}{61.71} & \multicolumn{1}{c|}{64.20} & \multicolumn{1}{c|}{{\ul 60.75}} & \multicolumn{1}{c|}{{\ul 64.60}} & \multicolumn{1}{c|}{51.18} & \multicolumn{1}{c|}{48.31} \\
\multicolumn{1}{|c|}{} & \multicolumn{1}{c|}{\Elf{}} & \multicolumn{1}{c|}{16.44} & \multicolumn{1}{c|}{13.12} & \multicolumn{1}{c|}{20.78} & \multicolumn{1}{c|}{18.09} & \multicolumn{1}{c|}{27.38} & \multicolumn{1}{c|}{24.78} & \multicolumn{1}{c|}{24.25} & \multicolumn{1}{c|}{27.10} & \multicolumn{1}{c|}{36.81} & \multicolumn{1}{c|}{36.12} & \multicolumn{1}{c|}{39.33} & \multicolumn{1}{c|}{39.81} & \multicolumn{1}{c|}{43.92} & \multicolumn{1}{c|}{62.73} & \multicolumn{1}{c|}{17.82} & \multicolumn{1}{c|}{24.20} & \multicolumn{1}{c|}{18.24} & \multicolumn{1}{c|}{23.84} & \multicolumn{1}{c|}{36.53} & \multicolumn{1}{c|}{40.19} & \multicolumn{1}{c|}{64.39} & \multicolumn{1}{c|}{69.45} & \multicolumn{1}{c|}{29.72} & \multicolumn{1}{c|}{23.18} \\
\multicolumn{1}{|c|}{} & \multicolumn{1}{c|}{\Elfp} & \multicolumn{1}{c|}{14.27} & \multicolumn{1}{c|}{{\ul 9.70}} & \multicolumn{1}{c|}{18.35} & \multicolumn{1}{c|}{14.91} & \multicolumn{1}{c|}{24.32} & \multicolumn{1}{c|}{21.86} & \multicolumn{1}{c|}{22.82} & \multicolumn{1}{c|}{25.29} & \multicolumn{1}{c|}{{\ul 33.37}} & \multicolumn{1}{c|}{{\ul 34.61}} & \multicolumn{1}{c|}{{\ul 37.80}} & \multicolumn{1}{c|}{{\ul 37.21}} & \multicolumn{1}{c|}{{\ul 40.97}} & \multicolumn{1}{c|}{63.69} & \multicolumn{1}{c|}{{\ul 17.30}} & \multicolumn{1}{c|}{21.58} & \multicolumn{1}{c|}{{\ul 16.24}} & \multicolumn{1}{c|}{21.21} & \multicolumn{1}{c|}{34.56} & \multicolumn{1}{c|}{39.49} & \multicolumn{1}{c|}{65.34} & \multicolumn{1}{c|}{70.36} & \multicolumn{1}{c|}{{\ul 27.31}} & \multicolumn{1}{c|}{20.57} \\
\multicolumn{1}{|c|}{} & \multicolumn{1}{c|}{\Camel{}} & \multicolumn{1}{c|}{{\ul 10.36}} & \multicolumn{1}{c|}{10.79} & \multicolumn{1}{c|}{{\ul 11.47}} & \multicolumn{1}{c|}{{\ul 11.59}} & \multicolumn{1}{c|}{\textbf{12.09}} & \multicolumn{1}{c|}{{\ul 11.45}} & \multicolumn{1}{c|}{\textbf{10.18}} & \multicolumn{1}{c|}{{\ul 13.58}} & \multicolumn{1}{c|}{/} & \multicolumn{1}{c|}{/} & \multicolumn{1}{c|}{/} & \multicolumn{1}{c|}{/} & \multicolumn{1}{c|}{/} & \multicolumn{1}{c|}{/} & \multicolumn{1}{c|}{/} & \multicolumn{1}{c|}{\textbf{13.85}} & \multicolumn{1}{c|}{20.01} & \multicolumn{1}{c|}{{\ul 19.97}} & \multicolumn{1}{c|}{{\ul 31.94}} & \multicolumn{1}{c|}{{\ul 34.41}} & \multicolumn{1}{c|}{/} & \multicolumn{1}{c|}{/} & \multicolumn{1}{c|}{/} & \multicolumn{1}{c|}{{\ul 14.86}} \\
\multicolumn{1}{|c|}{\multirow{-6}{*}{\rotatebox{90}{ACB}}} & \multicolumn{1}{c|}{\cellcolor[HTML]{FFEEC4}\method{}} & \multicolumn{1}{c|}{\cellcolor[HTML]{FFEEC4}\textbf{10.35}} & \multicolumn{1}{c|}{\cellcolor[HTML]{FFEEC4}\textbf{7.12}} & \multicolumn{1}{c|}{\cellcolor[HTML]{FFEEC4}\textbf{11.33}} & \multicolumn{1}{c|}{\cellcolor[HTML]{FFEEC4}\textbf{8.01}} & \multicolumn{1}{c|}{\cellcolor[HTML]{FFEEC4}{\ul 13.22}} & \multicolumn{1}{c|}{\cellcolor[HTML]{FFEEC4}\textbf{9.71}} & \multicolumn{1}{c|}{\cellcolor[HTML]{FFEEC4}{\ul 11.59}} & \multicolumn{1}{c|}{\cellcolor[HTML]{FFEEC4}\textbf{12.22}} & \multicolumn{1}{c|}{\cellcolor[HTML]{FFEEC4}\textbf{14.87}} & \multicolumn{1}{c|}{\cellcolor[HTML]{FFEEC4}\textbf{19.47}} & \multicolumn{1}{c|}{\cellcolor[HTML]{FFEEC4}\textbf{30.67}} & \multicolumn{1}{c|}{\cellcolor[HTML]{FFEEC4}\textbf{29.25}} & \multicolumn{1}{c|}{\cellcolor[HTML]{FFEEC4}\textbf{25.89}} & \multicolumn{1}{c|}{\cellcolor[HTML]{FFEEC4}\textbf{52.28}} & \multicolumn{1}{c|}{\cellcolor[HTML]{FFEEC4}\textbf{12.38}} & \multicolumn{1}{c|}{\cellcolor[HTML]{FFEEC4}{\ul 14.13}} & \multicolumn{1}{c|}{\cellcolor[HTML]{FFEEC4}\textbf{13.27}} & \multicolumn{1}{c|}{\cellcolor[HTML]{FFEEC4}\textbf{15.00}} & \multicolumn{1}{c|}{\cellcolor[HTML]{FFEEC4}\textbf{24.76}} & \multicolumn{1}{c|}{\cellcolor[HTML]{FFEEC4}\textbf{26.53}} & \multicolumn{1}{c|}{\cellcolor[HTML]{FFEEC4}\textbf{57.86}} & \multicolumn{1}{c|}{\cellcolor[HTML]{FFEEC4}\textbf{58.70}} & \multicolumn{1}{c|}{\cellcolor[HTML]{FFEEC4}\textbf{17.82}} & \multicolumn{1}{c|}{\cellcolor[HTML]{FFEEC4}\textbf{12.70}} \\ 
\hline

\multicolumn{26}{c}{\vspace{-1em}} \\ \hline
\multicolumn{1}{|c|}{} & \multicolumn{1}{c|}{\Gorilla{}} & \multicolumn{1}{c|}{10.06} & \multicolumn{1}{c|}{34.37} & \multicolumn{1}{c|}{9.77} & \multicolumn{1}{c|}{34.42} & \multicolumn{1}{c|}{3.55} & \multicolumn{1}{c|}{18.19} & \multicolumn{1}{c|}{30.80} & \multicolumn{1}{c|}{17.69} & \multicolumn{1}{c|}{28.78} & \multicolumn{1}{c|}{2.97} & \multicolumn{1}{c|}{3.50} & \multicolumn{1}{c|}{{\ul 25.95}} & \multicolumn{1}{c|}{{\ul 21.16}} & \multicolumn{1}{c|}{{\ul 23.17}} & \multicolumn{1}{c|}{\textbf{22.68}} & \multicolumn{1}{c|}{{\ul 6.81}} & \multicolumn{1}{c|}{{\ul 11.66}} & \multicolumn{1}{c|}{30.46} & \multicolumn{1}{c|}{25.45} & \multicolumn{1}{c|}{7.41} & \multicolumn{1}{c|}{7.65} & \multicolumn{1}{c|}{\textbf{12.16}} & \multicolumn{1}{c|}{13.94} & \multicolumn{1}{c|}{14.95} \\
\multicolumn{1}{|c|}{} & \multicolumn{1}{c|}{\Chimp} & \multicolumn{1}{c|}{13.20} & \multicolumn{1}{c|}{{\ul 35.01}} & \multicolumn{1}{c|}{10.58} & \multicolumn{1}{c|}{15.13} & \multicolumn{1}{c|}{16.25} & \multicolumn{1}{c|}{35.63} & \multicolumn{1}{c|}{{\ul 34.60}} & \multicolumn{1}{c|}{\textbf{26.57}} & \multicolumn{1}{c|}{31.80} & \multicolumn{1}{c|}{2.25} & \multicolumn{1}{c|}{{\ul 5.35}} & \multicolumn{1}{c|}{20.34} & \multicolumn{1}{c|}{\textbf{28.19}} & \multicolumn{1}{c|}{16.57} & \multicolumn{1}{c|}{{\ul 21.51}} & \multicolumn{1}{c|}{\textbf{25.50}} & \multicolumn{1}{c|}{0.56} & \multicolumn{1}{c|}{31.92} & \multicolumn{1}{c|}{{\ul 26.26}} & \multicolumn{1}{c|}{{\ul 30.09}} & \multicolumn{1}{c|}{\textbf{23.01}} & \multicolumn{1}{c|}{{\ul 7.91}} & \multicolumn{1}{c|}{\textbf{15.69}} & \multicolumn{1}{c|}{17.47} \\
\multicolumn{1}{|c|}{} & \multicolumn{1}{c|}{\Elf{}} & \multicolumn{1}{c|}{10.94} & \multicolumn{1}{c|}{13.36} & \multicolumn{1}{c|}{\textbf{35.11}} & \multicolumn{1}{c|}{36.87} & \multicolumn{1}{c|}{\textbf{18.11}} & \multicolumn{1}{c|}{36.20} & \multicolumn{1}{c|}{18.45} & \multicolumn{1}{c|}{22.77} & \multicolumn{1}{c|}{32.79} & \multicolumn{1}{c|}{2.89} & \multicolumn{1}{c|}{{\textbf{6.67}}} & \multicolumn{1}{c|}{19.96} & \multicolumn{1}{c|}{1.58} & \multicolumn{1}{c|}{\textbf{26.90}} & \multicolumn{1}{c|}{4.14} & \multicolumn{1}{c|}{5.07} & \multicolumn{1}{c|}{1.12} & \multicolumn{1}{c|}{{\ul 34.45}} & \multicolumn{1}{c|}{4.98} & \multicolumn{1}{c|}{10.02} & \multicolumn{1}{c|}{9.90} & \multicolumn{1}{c|}{6.37} & \multicolumn{1}{c|}{10.93} & \multicolumn{1}{c|}{13.53} \\
\multicolumn{1}{|c|}{} & \multicolumn{1}{c|}{\Elfp} & \multicolumn{1}{c|}{{\ul 17.01}} & \multicolumn{1}{c|}{10.85} & \multicolumn{1}{c|}{11.16} & \multicolumn{1}{c|}{{\ul 39.70}} & \multicolumn{1}{c|}{{\ul 17.63}} & \multicolumn{1}{c|}{{\ul 39.74}} & \multicolumn{1}{c|}{31.35} & \multicolumn{1}{c|}{{\ul 26.23}} & \multicolumn{1}{c|}{\textbf{36.51}} & \multicolumn{1}{c|}{\textbf{4.12}} & \multicolumn{1}{c|}{4.77} & \multicolumn{1}{c|}{\textbf{33.40}} & \multicolumn{1}{c|}{20.86} & \multicolumn{1}{c|}{1.22} & \multicolumn{1}{c|}{21.01} & \multicolumn{1}{c|}{6.15} & \multicolumn{1}{c|}{7.82} & \multicolumn{1}{c|}{\textbf{38.34}} & \multicolumn{1}{c|}{{\ul 29.13}} & \multicolumn{1}{c|}{\textbf{30.79}} & \multicolumn{1}{c|}{{\ul 19.16}} & \multicolumn{1}{c|}{4.81} & \multicolumn{1}{c|}{{\ul 15.14}} & \multicolumn{1}{c|}{{\ul 19.92}} \\
\multicolumn{1}{|c|}{} & \multicolumn{1}{c|}{\Camel{}} & \multicolumn{1}{c|}{8.46} & \multicolumn{1}{c|}{0.47} & \multicolumn{1}{c|}{6.28} & \multicolumn{1}{c|}{15.25} & \multicolumn{1}{c|}{2.98} & \multicolumn{1}{c|}{9.61} & \multicolumn{1}{c|}{10.44} & \multicolumn{1}{c|}{4.21} & \multicolumn{1}{c|}{/} & \multicolumn{1}{c|}{/} & \multicolumn{1}{c|}{/} & \multicolumn{1}{c|}{/} & \multicolumn{1}{c|}{/} & \multicolumn{1}{c|}{/} & \multicolumn{1}{c|}{/} & \multicolumn{1}{c|}{2.75} & \multicolumn{1}{c|}{0.34} & \multicolumn{1}{c|}{3.66} & \multicolumn{1}{c|}{21.46} & \multicolumn{1}{c|}{3.97} & \multicolumn{1}{c|}{/} & \multicolumn{1}{c|}{/} & \multicolumn{1}{c|}{/} & \multicolumn{1}{c|}{4.23} \\
\multicolumn{1}{|c|}{\multirow{-6}{*}{\rotatebox{90}{Com. Speed}}} & \multicolumn{1}{c|}{\cellcolor[HTML]{FFEEC4}\method{}} & \multicolumn{1}{c|}{\cellcolor[HTML]{FFEEC4}\textbf{23.86}} & \multicolumn{1}{c|}{\cellcolor[HTML]{FFEEC4}\textbf{37.08}} & \multicolumn{1}{c|}{\cellcolor[HTML]{FFEEC4}{\ul 32.88}} & \multicolumn{1}{c|}{\cellcolor[HTML]{FFEEC4}\textbf{51.57}} & \multicolumn{1}{c|}{\cellcolor[HTML]{FFEEC4}16.74} & \multicolumn{1}{c|}{\cellcolor[HTML]{FFEEC4}\textbf{49.63}} & \multicolumn{1}{c|}{\cellcolor[HTML]{FFEEC4}\textbf{37.71}} & \multicolumn{1}{c|}{\cellcolor[HTML]{FFEEC4}19.60} & \multicolumn{1}{c|}{\cellcolor[HTML]{FFEEC4}{\ul 33.22}} & \multicolumn{1}{c|}{\cellcolor[HTML]{FFEEC4}{\ul 3.61}} & \multicolumn{1}{c|}{\cellcolor[HTML]{FFEEC4}4.65} & \multicolumn{1}{c|}{\cellcolor[HTML]{FFEEC4}21.44} & \multicolumn{1}{c|}{\cellcolor[HTML]{FFEEC4}13.49} & \multicolumn{1}{c|}{\cellcolor[HTML]{FFEEC4}7.45} & \multicolumn{1}{c|}{\cellcolor[HTML]{FFEEC4}15.89} & \multicolumn{1}{c|}{\cellcolor[HTML]{FFEEC4}6.77} & \multicolumn{1}{c|}{\cellcolor[HTML]{FFEEC4}\textbf{19.00}} & \multicolumn{1}{c|}{\cellcolor[HTML]{FFEEC4}21.82} & \multicolumn{1}{c|}{\cellcolor[HTML]{FFEEC4}19.93} & \multicolumn{1}{c|}{\cellcolor[HTML]{FFEEC4}17.49} & \multicolumn{1}{c|}{\cellcolor[HTML]{FFEEC4}2.14} & \multicolumn{1}{c|}{\cellcolor[HTML]{FFEEC4}1.86} & \multicolumn{1}{c|}{\cellcolor[HTML]{FFEEC4}14.94} & \multicolumn{1}{c|}{\cellcolor[HTML]{FFEEC4}\textbf{24.05}} \\ 
\hline

\multicolumn{26}{c}{\vspace{-1em}} \\ \hline
\multicolumn{1}{|c|}{} & \multicolumn{1}{c|}{\Gorilla{}} & \multicolumn{1}{c|}{60.61} & \multicolumn{1}{c|}{{\ul 78.96}} & \multicolumn{1}{c|}{52.20} & \multicolumn{1}{c|}{72.09} & \multicolumn{1}{c|}{7.76} & \multicolumn{1}{c|}{{\ul 62.94}} & \multicolumn{1}{c|}{48.26} & \multicolumn{1}{c|}{44.95} & \multicolumn{1}{c|}{{\ul 58.26}} & \multicolumn{1}{c|}{8.37} & \multicolumn{1}{c|}{{\ul 37.97}} & \multicolumn{1}{c|}{48.03} & \multicolumn{1}{c|}{\textbf{53.97}} & \multicolumn{1}{c|}{{\ul 52.49}} & \multicolumn{1}{c|}{29.12} & \multicolumn{1}{c|}{21.72} & \multicolumn{1}{c|}{{\ul 40.45}} & \multicolumn{1}{c|}{{\ul 66.49}} & \multicolumn{1}{c|}{54.81} & \multicolumn{1}{c|}{48.25} & \multicolumn{1}{c|}{42.46} & \multicolumn{1}{c|}{\textbf{40.66}} & \multicolumn{1}{c|}{41.39} & \multicolumn{1}{c|}{44.82} \\
\multicolumn{1}{|c|}{} & \multicolumn{1}{c|}{\Chimp} & \multicolumn{1}{c|}{{\ul 64.87}} & \multicolumn{1}{c|}{\textbf{81.76}} & \multicolumn{1}{c|}{\textbf{69.16}} & \multicolumn{1}{c|}{73.51} & \multicolumn{1}{c|}{8.29} & \multicolumn{1}{c|}{46.08} & \multicolumn{1}{c|}{21.46} & \multicolumn{1}{c|}{{\ul 58.53}} & \multicolumn{1}{c|}{52.25} & \multicolumn{1}{c|}{\textbf{26.50}} & \multicolumn{1}{c|}{34.56} & \multicolumn{1}{c|}{{\ul 52.79}} & \multicolumn{1}{c|}{47.65} & \multicolumn{1}{c|}{43.73} & \multicolumn{1}{c|}{\textbf{71.59}} & \multicolumn{1}{c|}{{\ul 33.08}} & \multicolumn{1}{c|}{24.80} & \multicolumn{1}{c|}{\textbf{70.72}} & \multicolumn{1}{c|}{54.92} & \multicolumn{1}{c|}{{\ul 59.81}} & \multicolumn{1}{c|}{{\ul 50.23}} & \multicolumn{1}{c|}{{\ul 37.58}} & \multicolumn{1}{c|}{{\ul 44.41}} & \multicolumn{1}{c|}{44.21} \\
\multicolumn{1}{|c|}{} & \multicolumn{1}{c|}{\Elf{}} & \multicolumn{1}{c|}{60.92} & \multicolumn{1}{c|}{69.30} & \multicolumn{1}{c|}{21.75} & \multicolumn{1}{c|}{67.49} & \multicolumn{1}{c|}{26.59} & \multicolumn{1}{c|}{28.57} & \multicolumn{1}{c|}{57.27} & \multicolumn{1}{c|}{23.17} & \multicolumn{1}{c|}{36.11} & \multicolumn{1}{c|}{6.14} & \multicolumn{1}{c|}{27.19} & \multicolumn{1}{c|}{51.41} & \multicolumn{1}{c|}{{\ul 47.86}} & \multicolumn{1}{c|}{51.22} & \multicolumn{1}{c|}{56.38} & \multicolumn{1}{c|}{24.52} & \multicolumn{1}{c|}{15.01} & \multicolumn{1}{c|}{37.78} & \multicolumn{1}{c|}{37.70} & \multicolumn{1}{c|}{7.35} & \multicolumn{1}{c|}{43.99} & \multicolumn{1}{c|}{25.98} & \multicolumn{1}{c|}{31.86} & \multicolumn{1}{c|}{30.97} \\
\multicolumn{1}{|c|}{} & \multicolumn{1}{c|}{\Elfp} & \multicolumn{1}{c|}{42.92} & \multicolumn{1}{c|}{44.18} & \multicolumn{1}{c|}{56.77} & \multicolumn{1}{c|}{68.25} & \multicolumn{1}{c|}{33.23} & \multicolumn{1}{c|}{53.84} & \multicolumn{1}{c|}{58.62} & \multicolumn{1}{c|}{51.62} & \multicolumn{1}{c|}{53.78} & \multicolumn{1}{c|}{5.36} & \multicolumn{1}{c|}{5.75} & \multicolumn{1}{c|}{36.37} & \multicolumn{1}{c|}{28.64} & \multicolumn{1}{c|}{48.47} & \multicolumn{1}{c|}{10.05} & \multicolumn{1}{c|}{11.26} & \multicolumn{1}{c|}{1.04} & \multicolumn{1}{c|}{49.29} & \multicolumn{1}{c|}{51.56} & \multicolumn{1}{c|}{48.81} & \multicolumn{1}{c|}{9.67} & \multicolumn{1}{c|}{26.43} & \multicolumn{1}{c|}{25.90} & \multicolumn{1}{c|}{33.11} \\
\multicolumn{1}{|c|}{} & \multicolumn{1}{c|}{\Camel{}} & \multicolumn{1}{c|}{\textbf{95.15}} & \multicolumn{1}{c|}{75.00} & \multicolumn{1}{c|}{38.54} & \multicolumn{1}{c|}{{\ul 95.04}} & \multicolumn{1}{c|}{{\ul 35.78}} & \multicolumn{1}{c|}{32.68} & \multicolumn{1}{c|}{{\ul 83.70}} & \multicolumn{1}{c|}{35.83} & \multicolumn{1}{c|}{/} & \multicolumn{1}{c|}{/} & \multicolumn{1}{c|}{/} & \multicolumn{1}{c|}{/} & \multicolumn{1}{c|}{/} & \multicolumn{1}{c|}{/} & \multicolumn{1}{c|}{/} & \multicolumn{1}{c|}{\textbf{67.42}} & \multicolumn{1}{c|}{12.51} & \multicolumn{1}{c|}{36.22} & \multicolumn{1}{c|}{{\ul 56.90}} & \multicolumn{1}{c|}{48.56} & \multicolumn{1}{c|}{/} & \multicolumn{1}{c|}{/} & \multicolumn{1}{c|}{/} & \multicolumn{1}{c|}{{\ul 48.28}} \\
\multicolumn{1}{|c|}{\multirow{-6}{*}{\rotatebox{90}{Decom. Speed}}} & \multicolumn{1}{c|}{\cellcolor[HTML]{FFEEC4}\method{}} & \multicolumn{1}{c|}{\cellcolor[HTML]{FFEEC4}48.28} & \multicolumn{1}{c|}{\cellcolor[HTML]{FFEEC4}47.15} & \multicolumn{1}{c|}{\cellcolor[HTML]{FFEEC4}{\ul 63.86}} & \multicolumn{1}{c|}{\cellcolor[HTML]{FFEEC4}\textbf{105.21}} & \multicolumn{1}{c|}{\cellcolor[HTML]{FFEEC4}{\textbf{56.43}}} & \multicolumn{1}{c|}{\cellcolor[HTML]{FFEEC4}{\textbf{84.58}}} & \multicolumn{1}{c|}{\cellcolor[HTML]{FFEEC4}\textbf{87.13}} & \multicolumn{1}{c|}{\cellcolor[HTML]{FFEEC4}\textbf{91.15}} & \multicolumn{1}{c|}{\cellcolor[HTML]{FFEEC4}\textbf{92.65}} & \multicolumn{1}{c|}{\cellcolor[HTML]{FFEEC4}{\ul 8.89}} & \multicolumn{1}{c|}{\cellcolor[HTML]{FFEEC4}{\textbf{40.85}}} & \multicolumn{1}{c|}{\cellcolor[HTML]{FFEEC4}{\textbf{68.78}}} & \multicolumn{1}{c|}{\cellcolor[HTML]{FFEEC4}26.22} & \multicolumn{1}{c|}{\cellcolor[HTML]{FFEEC4}\textbf{62.16}} & \multicolumn{1}{c|}{\cellcolor[HTML]{FFEEC4}{\ul 62.48}} & \multicolumn{1}{c|}{\cellcolor[HTML]{FFEEC4}27.25} & \multicolumn{1}{c|}{\cellcolor[HTML]{FFEEC4}\textbf{82.46}} & \multicolumn{1}{c|}{\cellcolor[HTML]{FFEEC4}42.06} & \multicolumn{1}{c|}{\cellcolor[HTML]{FFEEC4}\textbf{73.83}} & \multicolumn{1}{c|}{\cellcolor[HTML]{FFEEC4}\textbf{64.61}} & \multicolumn{1}{c|}{\cellcolor[HTML]{FFEEC4}\textbf{59.51}} & \multicolumn{1}{c|}{\cellcolor[HTML]{FFEEC4}24.59} & \multicolumn{1}{c|}{\cellcolor[HTML]{FFEEC4}\textbf{53.12}} & \multicolumn{1}{c|}{\cellcolor[HTML]{FFEEC4}\textbf{63.29}} \\ 
\hline
\end{tabular}
\end{adjustbox}
\end{table*}

\section{Experimental Study}
\label{sec:exp}

\subsection{Experimental Settings}
\label{ssec:settings}

\noindent\textbf{{Datasets}}.
\lcyr{M2.(3)}\lcy{
We employ one synthetic and 21 real-world datasets from previous studies~\cite{li_elf_2023,Elf_Plus_2023,camel}: 14 time-series and 8 non-time-series sets kept in original order (see Table~\ref{tab:ACB}).
The latter lack timestamps, and their irregularity challenges smoothness-based converters (type-\ding{172}); \textsf{AS} (synthetic, noisy) stresses \textit{redundancy-based} converters (type-\ding{173}); \textsf{PA}/\textsf{PO} supply high-precision geospatial coordinates that strain both types.
These datasets contain sequences of up to several hundred million points, with numerical precision ($\mathit{dp}$) from 3 to 17, value ranges vary from narrow (e.g., $[0.01, 1.45]$) to broad (e.g., $[13{,}068, 532{,}323]$), and domains spanning meteorology, finance, ecology, mobility, blockchain, and geospatial analysis
\footnote{\lcy{The resources, statistics, processing scripts for all datasets, and the \method{} source code are publicly available at \url{https://github.com/SuDIS-ZJU/DeXOR}.}}.\lcyr{M2.(3)}
} \\
\noindent\textbf{{Baselines and Metrics}}. 
\lcy{
We compare \method{} with several state-of-the-art \problem{} schemes, including \Gorilla{}~\cite{pelkonen_gorilla_2015}, \Chimp{}~\cite{liakos_chimp_2022}, \Camel{}~\cite{camel}, and the \Elf{} series (\Elf{}~\cite{li_elf_2023} and its variant \Elfp{}~\cite{Elf_Plus_2023}).}
\cyr{R4.W1 \\ \&D1 (1)}
\lcy{To broaden the comparison, Section~\ref{ssec:larger_buffer} also covers shemes that use sliding or truncated windows: \ALP{}~\cite{afroozeh_alp_2023}, \Chimpp{}~\cite{liakos_chimp_2022}, \SElfStar{}~\cite{li2025adaptive}, and \ElfStar{}~\cite{li2025adaptive}.}
%
%
%
To provide a more intuitive measure of compression ratios, we adopt the \textbf{\emph{Average Compression Bits} (ACB)} metric from \ALP{}~\cite{afroozeh_alp_2023}.
\cyblue{Compression ratio equals ACB/64 (the IEEE754 representation length).
We assess the efficiency in terms of \textbf{\emph{compression}} and \textbf{\emph{decompression speed}}, measured in megabytes per second (MB/s).} \\
\noindent\textbf{{Settings and Implementation}}. 
Experiments run on an Ubuntu server (Intel Xeon Gold 6348, 512 GB RAM).  
\lcyr{M1.(4)}\lcy{All \problem{} schemes are Java-based and \texttt{Docker}-containerized for reproducibility and resource isolation; the pipeline supports timestamp-driven streaming.  
System-level IoTDB~\cite{wang_apache_2020} integration is reported in Section~\ref{ssec:tsdb}.}

\subsection{Overall Performance}
\label{ssec:fundamental}

\begin{figure}[t]
    \centering
    \includegraphics[width=0.90\columnwidth]{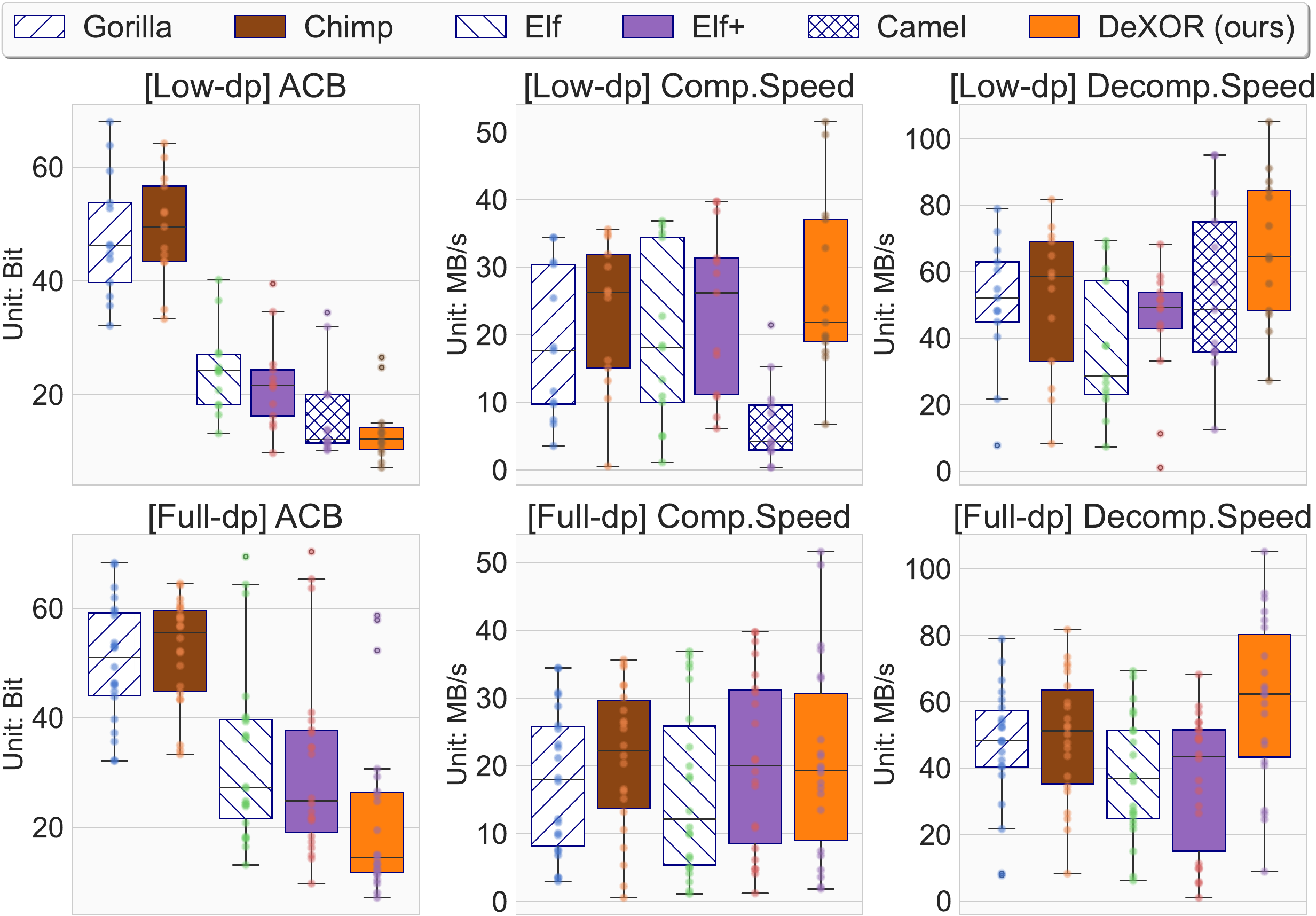}
    \caption{\lcy{Box-plot visualization of Table~\ref{tab:ACB} under low-$\mathit{dp}$ (top) and full-$\mathit{dp}$ (bottom); boxes show IQR as median, whiskers span 1.5$\times$IQR, and points mark individual dataset outcomes.}}
    \label{fig:full}
\end{figure}


Table~\ref{tab:ACB} \lcyr{R1.W5}\lcy{and Figure~\ref{fig:full}} present the results across the three metrics.

\noindent\textbf{Average Compression Bits}. 
\cyblue{
\method{} consistently delivers the best ACB overall, trailing only \Camel{} by $\leq2$ bits on three low-$\mathit{dp}$ sets (\textsf{DPT}, \textsf{SUK}, \textsf{EVC}), where \Camel{} employs lossy truncation targeting $q\in[-4,-1]$.
This specialization, however, incurs accuracy loss or overflow on high-$\mathit{dp}$ data (red) and outlier series (\textsf{FP}).
\Gorilla{} and \Chimp{} yield $\geq 30$ and generally above 50 bits via \XOR{} and variable-length encoding strategies; \Elf{} and \Elfp{} exhibit reduced ACBs ($\leq30$ and even $\leq10$ in some cases) on low-$\mathit{dp}$ but degrade on high-$\mathit{dp}$ (\textsf{AS}, \textsf{PA}, \textsf{PO}).  
Overall, only \method{} (alongside \Camel{}) sustains ACB$\leq20$ in most cases, compressing low-$\mathit{dp}$ inputs to $\leq25\%$ of their original size, and reaching as low as $12.5\%$ at best.
}



\noindent\textbf{Compression Speed}. 
\cyblue{\method{} stays within 5\% of the fastest baseline (\Chimp{}) and doubles the throughput of its closest ratio rival, \Camel{}.
On \textsf{CA} it lags because unstable tail coordinates trigger recomputation (Algorithm~\ref{alg:precondition}); yet in the low-$\mathit{dp}$ blue zone ($\mathit{dp} \leq 7$ and $q\in[-4,-1]$) it outruns \Elfp{} by 21\% and \Camel{} by 4.7×.
\XOR{} schemes (\Chimp{}, \Gorilla{}) minimize CPU; \Elf{} and \Elfp{} balance I/O; \method{}’s dynamic heuristics (see Algorithm~\ref{alg:precondition}) favour regular series, with only rare irregulars (\textsf{IR}, \textsf{CA}) incurring overhead.}


\noindent\textbf{Decompression Speed.}  
\cyblue{
\method{} leads on 13/22 datasets, outpacing \Camel{} by 31\% and \Gorilla{} by 41\% on average; it remains fastest even where its compression lags (\textsf{IR}, \textsf{CA}) and doubles state-of-the-art throughput on \textsf{SSD}.  
The only exception is a few low-$\mathit{dp}$ sets (e.g., \textsf{WS}) where \method{} runs at roughly half the speed of \Camel{}.
}


\cyblue{
This superior decompression performance stems from two factors.  
First, \method{} reads only the stored coordinate, so unstable patterns do not slow decompression.  
Second, reusing the previous LCP $\alpha_i$ cuts decompression work by ${\sim}50\%$ versus full recomputes (Section~\ref{ssec:decomp}).  
On low-$\mathit{dp}$ datasets (\textsf{PM}, \textsf{EVC}) I/O savings and reusable steps are limited, narrowing the advantage slightly, yet \method{} remains markedly faster on high-$\mathit{dp}$ data.
}


\subsection{\cy{Scalability Studies}}
\label{ssec:Scalability}

\begin{figure}[t]
    \centering
    \includegraphics[width=\columnwidth]{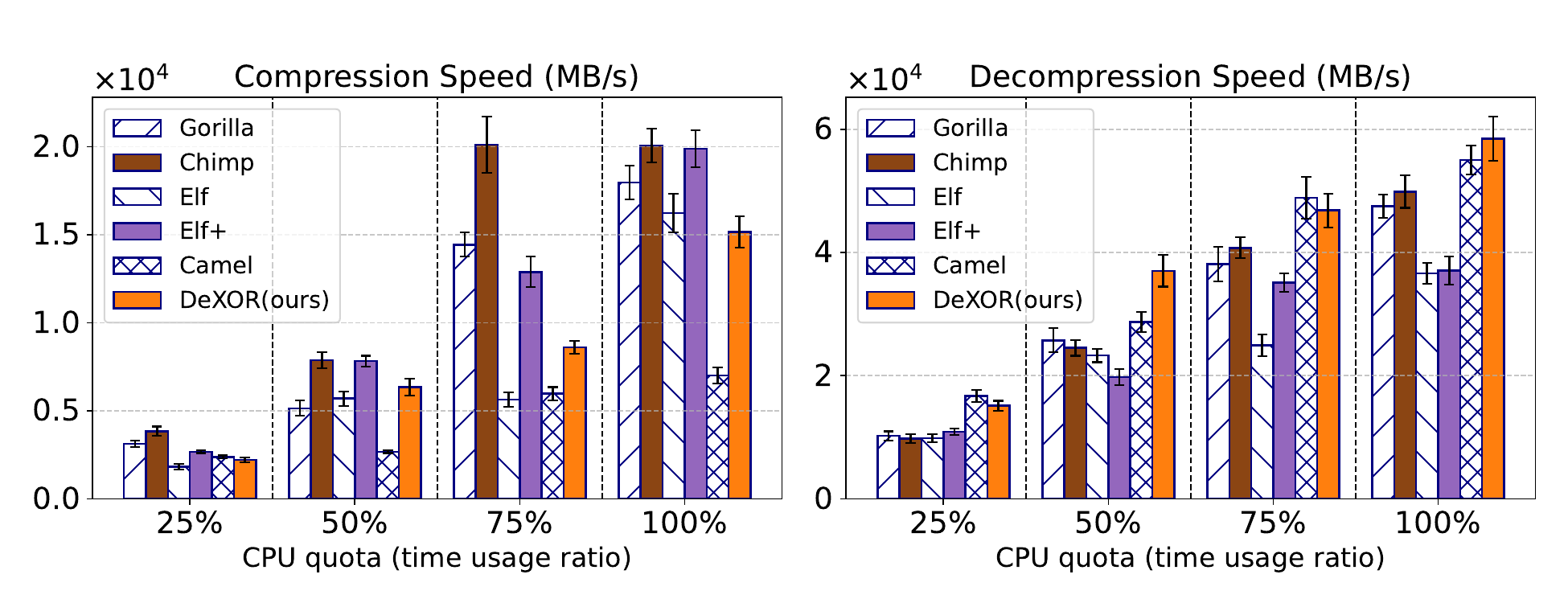}
    \caption{\cy{Scalability vs. CPU quotas. Bars show mean over 10 runs, with error bars indicating one standard deviation.}}
    \label{fig:period}
    \vspace{-2mm}
\end{figure}

\cyblue{
We assess the scalability of all schemes along two axes --- (1) \emph{CPU-quota constraints} and (2) \emph{data-sampling variability}:
}

\noindent\cy{\textbf{CPU-quota constraints} (Figure~\ref{fig:period})}.\cyr{R1.W4 \\ \&D4}
\cyblue{
We throttle CPU from 25\% to 100\% and record speeds.  
Most schemes scale linearly; \XOR-based schemes (\Gorilla{}, \Chimp{}) accelerate fastest at 50–75\%.  
\method{} improves steadily: second for decompression at 25\%, first for both compression and decompression at 50–100\%.
}

\noindent\cy{\textbf{Data-sampling variability} (Figure~\ref{fig:scalablity_datasets}).}\cyr{R1.W1 \\ \&D1 (1)}
We also assess robustness by compressing representative time-series datasets of distinct precisions, under two sampling strategies: (1) \emph{continuous sampling}, which preserves the temporal order, and (2) \emph{random sampling}, which disrupts contextual continuity.

Overall, \method{} consistently achieves the lowest ACB with minimal variance across both sampling strategies.
Random sampling significantly impacts the smoothness of streaming data. \XOR{}-based methods like \Gorilla{} are highly sensitive to smoothness, achieving the poorest ACB at 60\% random sampling in the \textsf{AP} dataset.
\Elf{} and \Elfp{} demonstrate better tolerance than simple \XOR{}-based methods.
Decimal-aware methods (\Camel{} and \method{}) perform well on \textsf{CT} and \textsf{DPT} datasets, but \Camel{} suffers a significant drop in ACB ($10 \sim 15$ bits lost) under random sampling on \textsf{CT} and \textsf{DPT}, while \method{} remains stable.
On \textsf{WS}, the \Elf{} series suffers from disrupted correlations. In contrast, \method{} maintains ACB stability by effectively utilizing data regularity (Section~\ref{sssec:postprocessing}).


\begin{figure*}[t]
    \centering
    \includegraphics[width=0.98\textwidth]{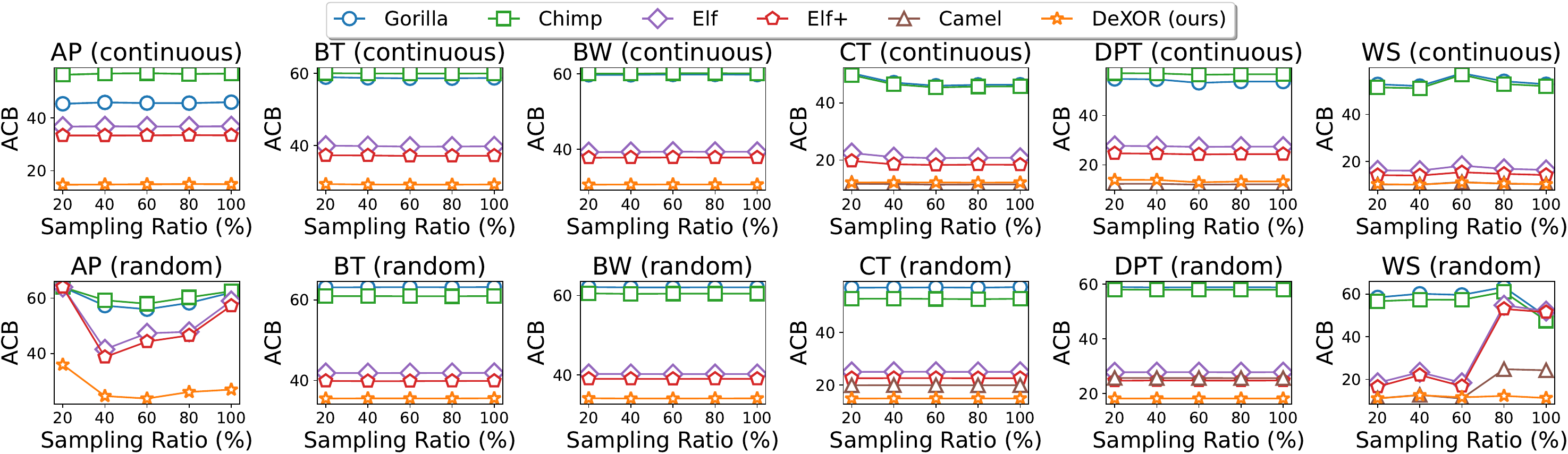}
    \caption{Scalability test on six representative time-series datasets (`continuous' means temporal order preserved while `random' means contextual continuity disrupted). \Camel{} does not work in high-$\mathit{dp}$ datasets \textsf{AP}, \textsf{BT}, and \textsf{BW}.}
    \label{fig:scalablity_datasets}
\end{figure*}

\subsection{\cy{Parameter Sensitivity Analysis}}
\label{ssec:Parameter} 

Section~\ref{ssec:elastic} introduces a hyperparameter $\rho$, which acts as a threshold for determining when to contract the bit length $\mathtt{EL}$ of the recorded exponential subtraction.
\cyr{R2.W3 \\ (1) 3a}\cy{
Optimizing $\rho$ can slightly improve compression ratios in high-$\mathit{dp}$ datasets, such as \textsf{AS}, \textsf{PA}, and \textsf{PO}.} 

As shown in Figure~\ref{fig:parameter_experiment}, the three dashed lines represent the scenario where contraction never occurs ($\rho \to +\infty$), resulting in significant redundancy and the worst ACB. Conversely, unconditional contraction ($\rho = 0$) performs better, indicating that uncontrolled expansion has a more detrimental impact than failures caused by premature contraction. This observation suggests favoring smaller $\rho$ values.
The effect of $\rho$ on ACB follows a consistent trend across datasets: ACB decreases rapidly to a minimal value and then grows very slowly. For the \textsf{PA} dataset, the ACB stabilizes once $\rho \geq 1$ and only begins to rise noticeably at $\rho \geq 9$. In contrast, for the \textsf{PO} dataset, $\rho = 0$ outperforms $\rho \geq 1$, suggesting frequent attenuation of exponential fluctuations, making $\rho = 0$ a locally optimal choice.

\cy{Based on these findings,\cyr{R2.W3 \\ (1) 3a} we set $\rho = 8$ as the default value for all datasets. 
This choice is likely optimal for \textsf{PA} and \textsf{PO}, while the difference from the local optimum for \textsf{AS} is less than 0.01 bits.
In most cases, this default setting suffices.}

\begin{figure}[t]
    \centering
    \includegraphics[width=0.95\columnwidth]{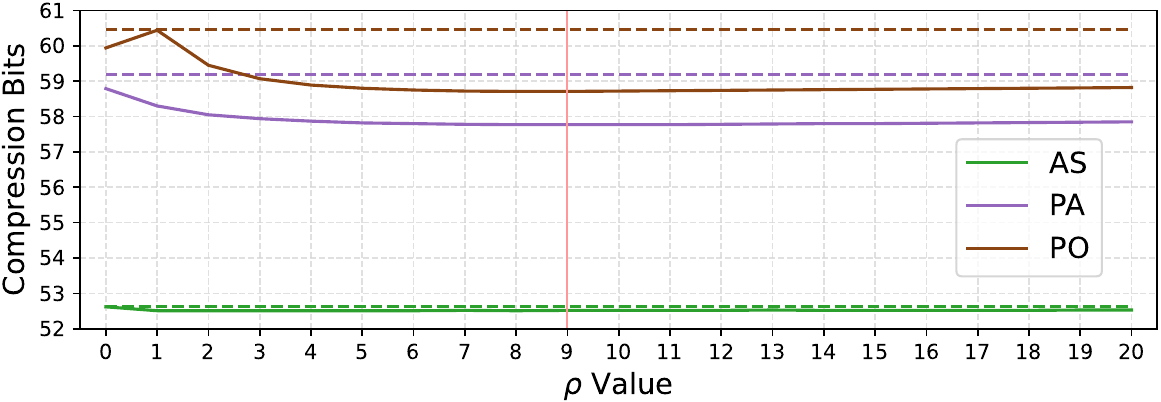}
    \caption{The impact of $\rho$ on ACB for \method{}. Dashed lines represent the cases when $\rho \to +\infty$.
}
    \label{fig:parameter_experiment}
    \vspace{-2mm}
\end{figure}

\subsection{Ablation Studies}
\label{ssec:Ablation} 

\begin{table}
    \centering
    \caption{Ablation study with best ACB results in bold. 
    $\Delta$ indicates ACB changes, with green for gains and red for losses. Gray sections indicate deactivation of the exception handler.}
    \label{tab:ablation}
    \setlength{\tabcolsep}{3pt} 
    \renewcommand{\arraystretch}{0.4} 
    \resizebox{0.95\columnwidth}{!}{%
    \begin{tabular}{ccr|rr|rr|rr}
    \toprule
    \multicolumn{2}{l}{\textbf{Dataset}} & \textbf{\method} & \multicolumn{2}{c|}{\textbf{w/o Excep.}} & \multicolumn{2}{c|}{\textbf{w/o \textsc{dec. xor}}} & \multicolumn{2}{c}{\textbf{w/o Both}} \\ \cmidrule(lr){4-5} \cmidrule(lr){6-7} \cmidrule(lr){8-9} 
     &  & \textbf{ACB $\downarrow$} & \textbf{ACB $\downarrow$} & \textbf{$\Delta$ (\%)} & \textbf{ACB $\downarrow$} & \textbf{$\Delta$ (\%)} & \textbf{ACB $\downarrow$} & \textbf{$\Delta$ (\%)} \\ \midrule
     \multirow{14}{*}{\rotatebox{90}{\textbf{Time-Series (TS)}}} & \textsf{WS} & \textbf{10.35} & \textbf{10.35} & \cellcolor[HTML]{DDDDDD}0.00  & 10.73 & {\color[HTML]{CB0000} -3.67} & 10.73 & {\color[HTML]{CB0000} -3.67} \\
     & \textsf{PM} & \textbf{7.12} & \textbf{7.12} & \cellcolor[HTML]{DDDDDD}0.00  & 8.72 & {\color[HTML]{CB0000} -22.47} & 8.72 & {\color[HTML]{CB0000} -22.47} \\
     & \textsf{CT} & \textbf{11.33} & \textbf{11.33} & \cellcolor[HTML]{DDDDDD}0.00  & 14.02 & {\color[HTML]{CB0000} -23.74} & 14.02 & {\color[HTML]{CB0000} -23.74} \\
     & \textsf{IR} & \textbf{8.01} & \textbf{8.01} & \cellcolor[HTML]{DDDDDD}0.00  & 13.47 & {\color[HTML]{CB0000} -68.16} & 13.47 & {\color[HTML]{CB0000} -68.16} \\
     & \textsf{DPT} & \textbf{13.22} & \textbf{13.22} & \cellcolor[HTML]{DDDDDD}0.00  & 18.18 & {\color[HTML]{CB0000} -37.52} & 18.18 & {\color[HTML]{CB0000} -37.52} \\
     & \textsf{SUSA} & \textbf{9.71} & \textbf{9.71} & \cellcolor[HTML]{DDDDDD}0.00  & 19.06 & {\color[HTML]{CB0000} -96.29} & 19.06 & {\color[HTML]{CB0000} -96.29} \\
     & \textsf{SUK} & \textbf{11.59} & \textbf{11.59} & \cellcolor[HTML]{DDDDDD}0.00  & 21.16 & {\color[HTML]{CB0000} -82.57} & 21.16 & {\color[HTML]{CB0000} -82.57} \\
     & \textsf{SDE} & \textbf{12.22} & \textbf{12.22} & \cellcolor[HTML]{DDDDDD}0.00  & 21.68 & {\color[HTML]{CB0000} -77.41} & 21.68 & {\color[HTML]{CB0000} -77.41} \\
     & \textsf{AP} & \textbf{14.87} & \textbf{14.87} & \cellcolor[HTML]{DDDDDD}0.00  & 27.68 & {\color[HTML]{CB0000} -86.15} & 27.68 & {\color[HTML]{CB0000} -86.15} \\
     & \textsf{BM} & \textbf{19.47} & \textbf{19.47} & \cellcolor[HTML]{DDDDDD}0.00  & 29.33 & {\color[HTML]{CB0000} -50.64} & 29.33 & {\color[HTML]{CB0000} -50.64} \\
     & \textsf{BW} & \textbf{30.67} & \textbf{30.67} & \cellcolor[HTML]{DDDDDD}0.00  & 32.21 & {\color[HTML]{CB0000} -5.02} & 32.21 & {\color[HTML]{CB0000} -5.02} \\
     & \textsf{BT}& \textbf{29.25} & \textbf{29.25} & \cellcolor[HTML]{DDDDDD}0.00  & 31.73 & {\color[HTML]{CB0000} -8.48} & 31.73 & {\color[HTML]{CB0000} -8.48} \\
     & \textsf{BP} & \textbf{25.89} & \textbf{25.89} & \cellcolor[HTML]{DDDDDD}0.00  & 34.35 & {\color[HTML]{CB0000} -32.68} & 34.35 & {\color[HTML]{CB0000} -32.68} \\
     & \textsf{AS} & \textbf{52.28} & 53.13 & {\color[HTML]{CB0000} -1.63} & 55.93 & {\color[HTML]{CB0000} -6.98} & 64.70 & {\color[HTML]{CB0000} -23.76} \\ \midrule
     \multirow{8}{*}{\rotatebox{90}{\textbf{Non-TS}}} & \textsf{FP} & \textbf{12.38} & \textbf{12.38} & \cellcolor[HTML]{DDDDDD}0.00  & 16.99 & {\color[HTML]{CB0000} -37.24} & 16.99 & {\color[HTML]{CB0000} -37.24} \\
     & \textsf{EVC} & \textbf{14.13} & \textbf{14.13} & \cellcolor[HTML]{DDDDDD}0.00  & 14.77 & {\color[HTML]{CB0000} -4.53} & 14.77 & {\color[HTML]{CB0000} -4.53} \\
     & \textsf{SSD} & \textbf{13.27} & \textbf{13.27} & \cellcolor[HTML]{DDDDDD}0.00  & 17.20 & {\color[HTML]{CB0000} -29.62} & 17.20 & {\color[HTML]{CB0000} -29.62} \\
     & \textsf{BL} & \textbf{15.00} & \textbf{15.00} & \cellcolor[HTML]{DDDDDD}0.00  & 18.62 & {\color[HTML]{CB0000} -24.13} & 18.62 & {\color[HTML]{CB0000} -24.13} \\
     & \textsf{CA} & \textbf{24.76} & \textbf{24.76} & \cellcolor[HTML]{DDDDDD}0.00  & 24.97 & {\color[HTML]{CB0000} -0.85} & 24.97 & {\color[HTML]{CB0000} -0.85} \\
     & \textsf{CO} & 26.53 & 26.53 & \cellcolor[HTML]{DDDDDD}0.00  & \textbf{26.19} & {\color[HTML]{009901} \textbf{1.28}} & \textbf{26.19} & {\color[HTML]{009901} \textbf{1.28}} \\
     & \textsf{PA} & 57.86 & 63.97 & {\color[HTML]{CB0000} -10.56} & \textbf{57.61} & {\color[HTML]{009901} \textbf{0.43}} & 64.70 & {\color[HTML]{CB0000} -11.82} \\
     & \textsf{PO} & 58.7 & 65.05 & {\color[HTML]{CB0000} -10.82} & \textbf{58.49} & {\color[HTML]{009901} \textbf{0.36}} & 65.24 & {\color[HTML]{CB0000} -11.14} \\ 
    \midrule
    \multicolumn{2}{c}{\textbf{Average}} & \textbf{21.76} & 22.36 & {\color[HTML]{CB0000} -2.78} & 26.05 & {\color[HTML]{CB0000} -19.74} & 27.08 & {\color[HTML]{CB0000} -24.46} \\ 
    \bottomrule
    \end{tabular}
    }

\end{table}

We conduct a module-wise ablation analysis to evaluate the contribution of each component \cyblue{while always retaining the \method{} compressor}. 
%
We examine: (1) \emph{w/o Excep.}, which stores extreme-case data directly in IEEE754 format without exception handling; and (2) \emph{w/o \textsc{dec. xor}}, which skips eliminating the LCP (Section~\ref{ssec:converter}) but still stores metadata and the scaled suffix, including $\delta = (o_i - q_i)$, $q_i$ and $\beta_i = v_i \times 10^{-q_i}$. 
Results are summarized in Table~\ref{tab:ablation}.


\noindent\textbf{Effect of the Exception Handler}.
This plug-in module is primarily triggered for high-$\mathit{dp}$ datasets, as indicated by the gray sections in Table~\ref{tab:ablation}. While its activation is limited to specific cases, it contributes greatly to ACB performance on datasets like {\textsf{AS}}, {\textsf{PA}}, and {\textsf{PO}}.
Without this module, \method{} experiences a slight drop in performance on the time-series dataset {\textsf{AS}} and suffers notable degradation on non-time-series datasets such as {\textsf{PA}} and {\textsf{PO}---sometimes performing worse than the original 64-bit IEEE754 representation.
%
%
\lcyr{M1.(5)}
\lcy{However, we confirm that \method{} remains competitive w.r.t.\ all baselines in terms of ACB, even without the exception handler.}
\cyr{R2.W3 \\ (2) 3b}\cy{In contrast, after ablating this module, any drawback on low-precision datasets must stem from tolerance failure (see Section~\ref{sssec:preconditions}), yet the observed value is zero; therefore, tolerance errors are negligible.}

\noindent\textbf{Effect of \DXOR{}}.
Removing \DXOR{} results in severe performance degradation, reducing \method{} to a scaling-to-integer scheme. 
As reported in Table~\ref{tab:ablation}, this leads to an average ACB loss of 19.74\% across all datasets.
However, \DXOR{} introduces overhead for specific datasets (e.g., {\textsf{PA}}, {\textsf{PO}}, {\textsf{CO}}) when misaligned scales result in front-padding with zeros, increasing required precision.
For example, processing $\uline{8.81479} \diamond 88.1479$ as the reference could result in the LCP coordinate $(o_i=2) > (p_i=0)$, rendering $\delta = (o_i-q_i = 7) > \mathit{dp} =(p_i-q_i +1=6)$\footnote{Consider a simplified case: $\uline{v_X}=8.8$ and $v_Y=88.1$ have different head coordinates $(v_X.p = 0) \neq (v_Y.p = 1)$, implying LCP $\alpha=0$ at $o=2 = \max(v_X.p, v_Y.p) + 1$. This pads the suffix to $088$ in \DXOR{}, which, while correct, is inefficient.}. 
This reflects a trade-off in our converter: while reusing the previous suffix avoids storing it explicitly, misalignment can occasionally demand higher precision.

\noindent\textbf{Effect of Both Modules}.
When both modules are removed, \method{} essentially retains only the functions of metadata reuse (Section~\ref{sssec:postprocessing}) and the compressor (Section~\ref{ssec:compressor}). 
Even so, Table~\ref{tab:ablation} shows that the remaining components achieve competitive compression ratios compared to \Elf{}, \Elfp{}, and \Camel{} (cf.\ Table~\ref{tab:ACB}).

Two notable observations emerge: 
(1) On the {\textsf{CO}} dataset, the stripped-down version of \method{} achieves slightly better performance than the full version (+1.28\%), due to misalignment overhead caused by rare data cases. 
A similar effect is observed on {\textsf{PA}} and {\textsf{PO}}, though the losses are offset by the removal of the exception handler.
(2) \DXOR{} reduces exceptions by eliminating the LCP, allowing some high-$\mathit{dp}$ data to fall outside the bounds of exceptional cases. 
For example, on the {\textsf{AS}} dataset, ablation deteriorates ACB by -23.76\%, far worse than the combined impact of individual ablations ((-1.63\%) + (-6.98\%)). 
This underscores a key advantage of \method{}: its ability to perform post-judgment of exceptions, enabling more accurate handling of high-$\mathit{dp}$ yet smooth data. 

\subsection{SLC with Buffering or Pre-processing}
\label{ssec:larger_buffer} 

\noindent\textbf{Larger Buffer}.
\cy{
\cyr{R4.W1 \\ \&D1 (1)}Our main setting (Section~\ref{ssec:problem}) restricts access to only the immediately preceding values. In contrast, some schemes leverage larger buffers --- e.g., sliding windows (\Chimpp{}~\cite{liakos_chimp_2022}, \SElfStar{}~\cite{li2025adaptive}) or mini‑batches (\ALP{}~\cite{afroozeh_alp_2023}, \Elf{}$\ast$~\cite{li2025adaptive}) --- to improve compression. Because these methods require more memory, we report them separately rather than mixing them with the primary baselines in Table~\ref{tab:ACB}.
}
%
Each scheme's sensitivity to buffer size and its target scenarios differ.
To ensure fairness, we use each's recommended buffer sizes: $N=128$ for \Chimpp{}, $N=1024$ for \ALP{}, and $N=1000$ for \SElfStar{} and \ElfStar{}. 
\lcyr{M2.(5)}
\lcy{As shown in Table~\ref{tab:larger_buffer_comparison}, despite competing against buffer‑heavy schemes, \method{} achieves the fastest decompression and ranks second in both compression ratio and speed, trailing \ElfStar{} by only 0.27 bits and 1.12~MB/s.}

\begin{table}
\centering
\caption{Comparison with larger buffer schemes.}
\label{tab:larger_buffer_comparison}
\setlength{\tabcolsep}{3pt} 
\renewcommand{\arraystretch}{0.4} 
\resizebox{0.95\columnwidth}{!}{%
\begin{tabular}{@{}l|ccccc@{}}
\toprule
Metrics (geomean) & \Chimpp{} & \SElfStar{} & \ALP{}  & \ElfStar{}  & \method{} (ours) \\ \midrule
AVG. Comp. Bits $\downarrow$ & 31.37  & 19.55 & 52.97 & \textbf{17.55} & \uline{17.82} \\ 
Comp. Speed (MB/s) $\uparrow$     & 12.65  & 14.54 & 2.07  & \textbf{16.06} & \uline{14.94} \\ 
Decomp. Speed (MB/s) $\uparrow$   & 40.27  & \uline{49.69} & 38.73 & 23.41 & \textbf{53.12} \\ \bottomrule
\end{tabular}
}
\end{table}

\begin{table}[t]
\caption{\lcy{ACB ($\downarrow$) under pre-processing strategies: original, zero-mean (\textsf{-Z}), and standardization (\textsf{-S}) on four datasets.}}
\label{tab:preprocessing}
\centering
\setlength{\tabcolsep}{6pt} 
\renewcommand{\arraystretch}{0.4} 
\resizebox{0.80\columnwidth}{!}{%
\begin{tabular}{@{}lccccc@{}}
\toprule
Dataset & \Gorilla{} & \Chimp{} & \Elf{} & \Elfp{} & \method{} (ours) \\ 
\midrule
\textsf{CT} & 46.34 & 45.75 & 20.78 & \uline{18.35} & \textbf{11.33} \\
\textsf{CT-Z} & \uline{47.95} & \textbf{47.66} & 52.51 & 53.51 & 57.10 \\
\textsf{CT-S} & \textbf{51.48} & \uline{52.04} & 54.97 & 55.95 & 55.94 \\
\midrule
\textsf{AS} & \uline{53.06} & 58.33 & 62.73 & 63.69 & \textbf{52.28} \\
\textsf{AS-Z} & \textbf{51.00} & 58.39 & 60.57 & 61.56 & \uline{53.44} \\
\textsf{AS-S} & \textbf{53.21} & 58.43 & 62.74 & 63.68 & \uline{54.34} \\
\midrule
\textsf{FP} & 32.26 & 34.05 & 17.82 & \uline{17.30} & \textbf{12.38} \\
\textsf{FP-Z} & \textbf{29.79} & \uline{34.69} & 37.53 & 38.53 & 56.05 \\
\textsf{FP-S} & \textbf{35.99} & \uline{40.90} & 43.60 & 44.56 & 42.20 \\
\midrule
\textsf{PA} & 61.94 & \uline{60.75} & 64.39 & 65.34 & \textbf{57.86} \\
\textsf{PA-Z} & 65.77 & \uline{63.69} & 66.48 & 67.44 & \textbf{59.11} \\
\textsf{PA-S} & 67.40 & \uline{63.89} & 68.27 & 69.24 & \textbf{59.10} \\
\bottomrule
\end{tabular}
}
\end{table}

\begin{table*}[t]
\centering
\renewcommand{\arraystretch}{0.5} 
\caption{\lcy{All results are produced within Apache IoTDB Tsfile engine, evaluated on three heterogeneous real-world datasets: \textsf{CT}~(time-series), \textsf{FP}~(non-time-series), and \textsf{PA}~(high-precision non-time-series).}}
\label{tab:iotdb}
    \setlength{\tabcolsep}{8pt} 
    \renewcommand{\arraystretch}{0.4} 
    \resizebox{\textwidth}{!}{%
\begin{tabular}{llcccccccccccc}
\toprule
\multirow{2}{*}{Secondary Compression} & \multirow{2}{*}{Algorithm} &
\multicolumn{4}{c}{Average Compression Bits ($\downarrow$)} & \multicolumn{4}{c}{Compression Throughput (MB/s, $\uparrow$)} & \multicolumn{4}{c}{Query Latency (ms/1k, $\downarrow$)} \\
\cmidrule(lr){3-6} \cmidrule(lr){7-10} \cmidrule(lr){11-14}
& & \textsf{CT} & \textsf{FP} & \textsf{PA} & Avg. & \textsf{CT} & \textsf{FP} & \textsf{PA} & Avg. & \textsf{CT} & \textsf{FP} & \textsf{PA} & Avg. \\
\midrule
\multirow{3}{*}{Uncompressed}
& \Sprintz{} & 51.94 & 52.01 & 61.85 & 55.27 & 1.344 & 4.160 & 3.136 & 2.880 & 0.67 & 0.188 & 0.132 & 0.330 \\
& \Gorilla{} & 58.17 & 40.28 & 67.57 & 55.34 & 6.016 & 5.888 & \textbf{4.096} & \textbf{5.205} & 0.174 & \textbf{0.073} & \textbf{0.090} & 0.112 \\
& \method{}   & \textbf{12.91} & \textbf{13.96} & \textbf{59.45} & \textbf{28.78} & \textbf{6.144} & \textbf{6.016} & 3.072 & 4.949 & \textbf{0.093} & 0.083 & 0.106 & \textbf{0.094} \\
\midrule
\multirow{3}{*}{Lz4}
& \Sprintz{} & 37.57 & 34.93 & 60.95 & 44.48 & 4.032 & 5.696 & 3.264 & 4.331 & 0.105 & 0.118 & 0.127 & 0.116 \\
& \Gorilla{} & 29.00 & 27.97 & 66.70 & 41.22 & 5.120 & 5.760 & \textbf{4.480} & \textbf{5.120} & 0.095 & \textbf{0.081} & 0.101 & 0.092 \\
& \method{}   & \textbf{11.56} & \textbf{12.83} & \textbf{58.54} & \textbf{27.64} & \textbf{5.952} & \textbf{5.824} & 3.008 & 4.928 & \textbf{0.079} & 0.086 & \textbf{0.090} & \textbf{0.085} \\
\midrule
\multirow{3}{*}{Snappy}
& \Sprintz{} & 36.70 & 34.17 & 60.72 & 43.86 & 3.840 & 5.632 & \textbf{4.032} & 4.501 & 0.119 & 0.099 & 0.091 & 0.103 \\
& \Gorilla{} & 30.13 & 28.89 & 66.44 & 41.82 & 5.120 & \textbf{6.208} & 3.968 & 5.099 & 0.140 & \textbf{0.077} & \textbf{0.080} & 0.098 \\
& \method{}   & \textbf{11.69} & \textbf{12.83} & \textbf{58.32} & \textbf{27.61} & \textbf{6.272} & 5.760 & 3.328 & \textbf{5.120} & \textbf{0.073} & 0.092 & 0.092 & \textbf{0.085} \\
\bottomrule
\end{tabular}
}
\end{table*}

\noindent\textbf{Pre-processing Methods}.
\lcy{
Zero-mean normalization and 0-1 standarization reduce local smoothness, increase entropy, and inject irreversible loss alongside extra runtime cost: e.g., converting $50.0\to-0.0431\cdots$ expands precision requirements that strain each SLC scheme.
\lcyr{M2.(6) \\ \& R1.W2}As shown in Table~\ref{tab:preprocessing}, these pre-processing methods hurt all precision-redundancy schemes, including \method{}, while simple \XOR{} schemes (\Gorilla{}, \Chimp{}) benefit.  
The impact on \method{} is less pronounced for high-precision, irregular datasets (\textsf{AS} and \textsf{PA}), where it maintains competitive results.
Therefore, we recommend avoiding transformations that inflate precision or introduce loss before using \method{}; if normalization is required, it should be applied post-\method{}'s lossless decompression for better performance.
}

\subsection{\cy{Real-World Time Series Database Testing}}
\label{ssec:tsdb}

\begin{table}[t]
\centering
\caption{\lcy{Compression and query performance of \Gorilla{} and~\method{} in IoTDB on reprsentative vector datasets.}} 
\label{tab:vector}
    \setlength{\tabcolsep}{3pt} 
    \renewcommand{\arraystretch}{0.4} 
    \resizebox{0.95\columnwidth}{!}{%
\begin{tabular}{lcccccc}
\toprule
\multirow{2}{*}{Metric} & \multirow{2}{*}{Algorithm} &
\multicolumn{3}{c}{Datasets} \\ \cmidrule(lr){3-5}
 & & \textsf{SIFT} & \textsf{WR} & \textsf{WW} \\
\midrule
\multirow{2}{*}{Average Compression Bits ($\downarrow$)} 
  & \Gorilla{} & 19.65 & 42.84 & 43.03 \\ 
  & \method{}   & \textbf{12.99} & \textbf{13.47} & \textbf{13.11} \\
\midrule
\multirow{2}{*}{Compression Throughput (MB/s, $\uparrow$)} 
  & \Gorilla{} & 13.510 & 4.933 & 6.958 \\ 
  & \method{}   & \textbf{18.493} & \textbf{11.167} & \textbf{14.724} \\
\midrule
\multirow{2}{*}{Query Latency (ms/1k, $\downarrow$)} 
  & \Gorilla{} & 0.017 & 0.014 & 0.005 \\ 
  & \method{}   & \textbf{0.0002} & \textbf{0.010} & \textbf{0.004} \\
\bottomrule
\end{tabular}
}
\vspace{-2mm}
\end{table}

\lcy{
Apache IoTDB's columnar engine TsFile~\cite{wang_apache_2020}\lcyr{M1.(4)} has a pluggable encoding layer that lets us add \method{} without orchestration changes.  
On three diverse datasets (\textsf{CT}, \textsf{FP}, \textsf{PA}), we compare \method{} against IoTDB's \Gorilla{} and \Sprintz{} schemes, each cascaded with downstream Lz4~\cite{lz4} or Snappy~\cite{Snappy} to test secondary compression stacking.  
Table~\ref{tab:iotdb} shows that \method{} consistently achieves better compression than \Gorilla{} and \Sprintz{} (up to 4.5$\times$), retains ingestion throughput, and answers 1,000 single-value queries faster; the secondary compressors Lz4 and Snappy yield only $<2$\% extra space reductions and negligible runtime overhead and are thus optional for \method{}.
}

\noindent\textbf{Support for Vector Data Compression}.
\lcyr{R1.W1}
\lcy{
While the main study covers 22 one-dimensional sequences, vectors are also used widely.
To determine whether \method{} can be reused as-is, we pick two representative sets: 
\textsf{SIFT}~\cite{TexmexSiftsmall} (10k vectors, 128 dimensions for computer vision) and \textsc{Wine‑Quality}~\cite{wine_quality_186} (4,898 vectors, 11 dimensions for red/white wines---\textsf{WR} and \textsf{WW}).
Each dimension is compressed independently.
A query corresponds to retrieving a complete vector record, and the latency is measured over 1,000 such queries.
As shown in Table~\ref{tab:vector}, without any vector-specific redesign, \method{} still achieves $\sim3.5\times$ better compression, $\sim2\times$ better throughput, and much lower latency versus dimension-wise \Gorilla{}, indicating the utility of \method{} for vector data.
}



\section{Related Work} 
\label{sec:relatedwork}


\lcy{
\textbf{General-purpose Compression}.  
\textsf{LZ4}~\cite{lz4}, \textsf{LZMA}~\cite{LZMA}, \textsf{Snappy}~\cite{Snappy}, \textsf{XZ}~\cite{xz}, and \textsf{Brotli}~\cite{Brotli_2019} exploit redundancy via Huffman and dictio-\\naries, achieving better compression on static data but incurring decompress-recompress cycles that hinder real-time use.
} 
%

\cyblue{
\textbf{Streaming Lossy Compression} embeds data-cleaning techniques~\cite{klus_direct_2021,fink_compression_2011,liu_deep_2024} and error-bounded mechanisms~\cite{yang_most_2023,liu2021Buff,chandak_lfzip_2020,li2025serf,di_fast_2016} to trade precision for compression ratio. 
Though we target lossless compression, our techniques readily extend to lossy regimes.
}


\lcy{
\textbf{Learned Compression}~\cite{yang_most_2023,di_fast_2016,liu_deep_2024,jiang_time_2023,guerra2025neats} predicts values and encodes high-precision stochastic residuals, so most are lossy (e.g., \textsf{SZ}~\cite{di_fast_2016}, \textsf{MOST}~\cite{yang_most_2023}).  
\lcyr{R4}\textsf{NeaTS}~\cite{guerra2025neats} stands out by providing a \emph{configurable} switch, enabling both lossy and lossless compression.  
Strongly model-dependent and hardware-demanding, these techniques are promising but need further exploration.
}

\lcy{
\textbf{Lossless Compression} has been developed for \textbf{streaming data}~\cite{afroozeh_alp_2023,chiarot_time_2023,xia_time_2024,Elf_Plus_2023,li_elf_2023,kuschewski_btrblocks_2023,pelkonen_gorilla_2015,liakos_chimp_2022,blalock_sprintz_2018,yu_two-level_2020,yan_model-free_2021,gomez-brandon_lossless_2021,hwang_lossless_2023,hartmann_mtsc_2018,huamin_chen_race_2004,alghamdi_chainlink_2020,jiang_time_2023,khelifati_corad_2019,li2025adaptive} and \textbf{floating-point data}~\cite{pelkonen_gorilla_2015,Elf_Plus_2023,li_elf_2023,kuschewski_btrblocks_2023,afroozeh_alp_2023,li2025adaptive}.
Compressing streaming data often utilizes contextual statistics.
Batch schemes like \ALP{}~\cite{afroozeh_alp_2023} and \Elf$\ast$~\cite{li2025adaptive} employ off-line techniques, including entropy and dictionary coding, trading higher latency for better compression.
For floating-point data, IEEE754 encoding issues are mitigated by techniques like \XOR{}~\cite{pelkonen_gorilla_2015} and scaling-to-integer~\cite{afroozeh_alp_2023}.  
\method{} extends these techniques to support integer compression and includes a configurable exception handler to adapt online to diverse datasets.
}

\lcy{
\textbf{Vector Data Compression}.
Low-dimensional vectors can be compressed dimension-wise with SIMD parallelism~\cite{afroozeh_alp_2023}; Compressing high-dimensional vectors (e.g., \textsf{DeltaPQ}~\cite{wang2020deltapq}) relies on quantization. 
Compressing sequential values in a vector is intuitive (see our implementation in Section~\ref{ssec:tsdb}), while exploiting inter-vector smoothness for lossless compression remains an open problem.}

\vfill
\section{Conclusion}
\label{sec:conclusion}

We presents \method{}, a streaming lossless compression scheme for floating‑point data that applies lightweight \XOR{}‑like conversions in decimal space. By exploiting common decimal prefixes and suffixes, it unifies smoothness‑ and redundancy‑based techniques while mitigating IEEE754 conversion errors, anomalies, and high‑precision cases. \method{} delivers state‑of‑the‑art performance and flexibility for decimal‑system numerical compression.

Clear directions exist for building a \method{} family for diverse scenarios:
\cyr{R1.W1 \\ \&D1 (2)}\cy{(1) Batch-mode \method{}: Employing advanced parallelism (e.g., SIMD vectorization) for high-performance, native vector data compression.}
\cy{\cyr{R4.W1 \\ \&D1 (2)}(2) Model-based lossless \method{}: Refining \DXOR{} for compressing prediction residuals from lightweight time-series forecasting models.}
(3) Lossy \method{}: Creating powerful, error-bounded variants adapted to specific data conditions.



\clearpage
\balance
\bibliographystyle{ACM-Reference-Format}
\bibliography{time.series.compression}


\begin{thebibliography}{46}


\ifx \showCODEN    \undefined \def \showCODEN     #1{\unskip}     \fi
\ifx \showDOI      \undefined \def \showDOI       #1{#1}\fi
\ifx \showISBNx    \undefined \def \showISBNx     #1{\unskip}     \fi
\ifx \showISBNxiii \undefined \def \showISBNxiii  #1{\unskip}     \fi
\ifx \showISSN     \undefined \def \showISSN      #1{\unskip}     \fi
\ifx \showLCCN     \undefined \def \showLCCN      #1{\unskip}     \fi
\ifx \shownote     \undefined \def \shownote      #1{#1}          \fi
\ifx \showarticletitle \undefined \def \showarticletitle #1{#1}   \fi
\ifx \showURL      \undefined \def \showURL       {\relax}        \fi
\providecommand\bibfield[2]{#2}
\providecommand\bibinfo[2]{#2}
\providecommand\natexlab[1]{#1}
\providecommand\showeprint[2][]{arXiv:#2}

\bibitem[\protect\citeauthoryear{??}{LZM}{2001}]%
        {LZMA}
 \bibinfo{year}{2001}\natexlab{}.
\newblock \bibinfo{title}{LZMA compression/decompression in golang}.
\newblock
  \bibinfo{howpublished}{\url{https://code.google.com/archive/p/lzma/}}.
\newblock


\bibitem[\protect\citeauthoryear{??}{Tex}{2010}]%
        {TexmexSiftsmall}
 \bibinfo{year}{2010}\natexlab{}.
\newblock \bibinfo{title}{{ANN\_SIFT10K (siftsmall)}}.
\newblock
\newblock
\urldef\tempurl%
\url{http://corpus-texmex.irisa.fr/}
\showURL{%
\tempurl}


\bibitem[\protect\citeauthoryear{??}{lz4}{2013}]%
        {lz4}
 \bibinfo{year}{2013}\natexlab{}.
\newblock \bibinfo{title}{LZ4: Extremely Fast Compression algorithm}.
\newblock \bibinfo{howpublished}{\url{https://github.com/lz4/lz4}}.
\newblock


\bibitem[\protect\citeauthoryear{??}{IEE}{2019}]%
        {IEEE754_2019}
 \bibinfo{year}{2019}\natexlab{}.
\newblock \showarticletitle{IEEE Standard for Floating-Point Arithmetic}.
\newblock \bibinfo{journal}{\emph{IEEE Std 754-2019 (Revision of IEEE
  754-2008)}} (\bibinfo{year}{2019}), \bibinfo{pages}{1--84}.
\newblock


\bibitem[\protect\citeauthoryear{??}{xz}{2022}]%
        {xz}
 \bibinfo{year}{2022}\natexlab{}.
\newblock \bibinfo{title}{XZ Utils}.
\newblock \bibinfo{howpublished}{\url{https://github.com/tukaani-project/xz}}.
\newblock


\bibitem[\protect\citeauthoryear{??}{Sna}{2023}]%
        {Snappy}
 \bibinfo{year}{2023}\natexlab{}.
\newblock \bibinfo{title}{Snappy: A fast compressor/decompressor}.
\newblock \bibinfo{howpublished}{\url{https://github.com/google/snappy}}.
\newblock


\bibitem[\protect\citeauthoryear{Abadi, Madden, and Ferreira}{Abadi
  et~al\mbox{.}}{2006}]%
        {abadi_columnstores}
\bibfield{author}{\bibinfo{person}{Daniel~J. Abadi}, \bibinfo{person}{Samuel
  Madden}, {and} \bibinfo{person}{Miguel Ferreira}.}
  \bibinfo{year}{2006}\natexlab{}.
\newblock \showarticletitle{Integrating Compression and Execution in
  Column-Oriented Database Systems}. In \bibinfo{booktitle}{\emph{SIGMOD}}.
  \bibinfo{pages}{671--682}.
\newblock


\bibitem[\protect\citeauthoryear{Afroozeh, Kuffo, and Boncz}{Afroozeh
  et~al\mbox{.}}{2023}]%
        {afroozeh_alp_2023}
\bibfield{author}{\bibinfo{person}{Azim Afroozeh}, \bibinfo{person}{Leonardo~X.
  Kuffo}, {and} \bibinfo{person}{Peter Boncz}.}
  \bibinfo{year}{2023}\natexlab{}.
\newblock \showarticletitle{{ALP}: {Adaptive} {Lossless} floating-{Point}
  {Compression}}.
\newblock \bibinfo{journal}{\emph{SIGMOD}} \bibinfo{volume}{1},
  \bibinfo{number}{4} (\bibinfo{year}{2023}), \bibinfo{pages}{1--26}.
\newblock
\showISSN{2836-6573}


\bibitem[\protect\citeauthoryear{Alakuijala, Farruggia, Ferragina, Kliuchnikov,
  Obryk, Szabadka, and Vandevenne}{Alakuijala et~al\mbox{.}}{2018}]%
        {Brotli_2019}
\bibfield{author}{\bibinfo{person}{Jyrki Alakuijala}, \bibinfo{person}{Andrea
  Farruggia}, \bibinfo{person}{Paolo Ferragina}, \bibinfo{person}{Eugene
  Kliuchnikov}, \bibinfo{person}{Robert Obryk}, \bibinfo{person}{Zoltan
  Szabadka}, {and} \bibinfo{person}{Lode Vandevenne}.}
  \bibinfo{year}{2018}\natexlab{}.
\newblock \showarticletitle{Brotli: A general-purpose data compressor}.
\newblock \bibinfo{journal}{\emph{TOIS}} \bibinfo{volume}{37},
  \bibinfo{number}{1} (\bibinfo{year}{2018}), \bibinfo{pages}{1--30}.
\newblock


\bibitem[\protect\citeauthoryear{Alghamdi, Zhang, Zhang, Rundensteiner, and
  Eltabakh}{Alghamdi et~al\mbox{.}}{2020}]%
        {alghamdi_chainlink_2020}
\bibfield{author}{\bibinfo{person}{Noura Alghamdi}, \bibinfo{person}{Liang
  Zhang}, \bibinfo{person}{Huayi Zhang}, \bibinfo{person}{Elke~A.
  Rundensteiner}, {and} \bibinfo{person}{Mohamed~Y. Eltabakh}.}
  \bibinfo{year}{2020}\natexlab{}.
\newblock \showarticletitle{{ChainLink}: {Indexing} {Big} {Time} {Series}
  {Data} {For} {Long} {Subsequence} {Matching}}. In
  \bibinfo{booktitle}{\emph{ICDE}}. \bibinfo{pages}{529--540}.
\newblock
\showISBNx{978-1-72812-903-7}


\bibitem[\protect\citeauthoryear{Blalock, Madden, and Guttag}{Blalock
  et~al\mbox{.}}{2018}]%
        {blalock_sprintz_2018}
\bibfield{author}{\bibinfo{person}{Davis Blalock}, \bibinfo{person}{Samuel
  Madden}, {and} \bibinfo{person}{John Guttag}.}
  \bibinfo{year}{2018}\natexlab{}.
\newblock \showarticletitle{Sprintz: Time Series Compression for the Internet
  of Things}.
\newblock \bibinfo{journal}{\emph{IMWUT}} \bibinfo{volume}{2},
  \bibinfo{number}{3} (\bibinfo{year}{2018}), \bibinfo{pages}{1--23}.
\newblock


\bibitem[\protect\citeauthoryear{Chandak, Tatwawadi, Wen, Wang, Aparicio~Ojea,
  and Weissman}{Chandak et~al\mbox{.}}{2020}]%
        {chandak_lfzip_2020}
\bibfield{author}{\bibinfo{person}{Shubham Chandak}, \bibinfo{person}{Kedar
  Tatwawadi}, \bibinfo{person}{Chengtao Wen}, \bibinfo{person}{Lingyun Wang},
  \bibinfo{person}{Juan Aparicio~Ojea}, {and} \bibinfo{person}{Tsachy
  Weissman}.} \bibinfo{year}{2020}\natexlab{}.
\newblock \showarticletitle{{LFZip}: {Lossy} {Compression} of {Multivariate}
  {Floating}-{Point} {Time} {Series} {Data} via {Improved} {Prediction}}. In
  \bibinfo{booktitle}{\emph{DCC}}. \bibinfo{pages}{342--351}.
\newblock


\bibitem[\protect\citeauthoryear{Chiarot and Silvestri}{Chiarot and
  Silvestri}{2023}]%
        {chiarot_time_2023}
\bibfield{author}{\bibinfo{person}{Giacomo Chiarot} {and}
  \bibinfo{person}{Claudio Silvestri}.} \bibinfo{year}{2023}\natexlab{}.
\newblock \showarticletitle{Time Series Compression Survey}.
\newblock \bibinfo{journal}{\emph{Comput. Surveys}} \bibinfo{volume}{55},
  \bibinfo{number}{10} (\bibinfo{year}{2023}), \bibinfo{pages}{1--32}.
\newblock
\showISSN{0360-0300, 1557-7341}


\bibitem[\protect\citeauthoryear{ClickHouse}{ClickHouse}{2025}]%
        {clickhouse}
\bibfield{author}{\bibinfo{person}{ClickHouse}.}
  \bibinfo{year}{2025}\natexlab{}.
\newblock \bibinfo{title}{Compression in ClickHouse}.
\newblock
  \bibinfo{howpublished}{\url{https://clickhouse.com/docs/data-compression/compression-in-clickhouse}}.
\newblock


\bibitem[\protect\citeauthoryear{Cortez, Cerdeira, Almeida, Matos, and
  Reis}{Cortez et~al\mbox{.}}{2009}]%
        {wine_quality_186}
\bibfield{author}{\bibinfo{person}{Paulo Cortez}, \bibinfo{person}{A.
  Cerdeira}, \bibinfo{person}{F. Almeida}, \bibinfo{person}{T. Matos}, {and}
  \bibinfo{person}{J. Reis}.} \bibinfo{year}{2009}\natexlab{}.
\newblock \bibinfo{title}{{Wine Quality}}.
\newblock
\newblock
\urldef\tempurl%
\url{https://archive.ics.uci.edu/dataset/186/wine+quality}
\showURL{%
\tempurl}


\bibitem[\protect\citeauthoryear{Di and Cappello}{Di and Cappello}{2016}]%
        {di_fast_2016}
\bibfield{author}{\bibinfo{person}{Sheng Di} {and} \bibinfo{person}{Franck
  Cappello}.} \bibinfo{year}{2016}\natexlab{}.
\newblock \showarticletitle{Fast {Error}-{Bounded} {Lossy} {HPC} {Data}
  {Compression} with {SZ}}. In \bibinfo{booktitle}{\emph{IPDPS}}.
  \bibinfo{pages}{730--739}.
\newblock
\showISBNx{978-1-5090-2140-6}


\bibitem[\protect\citeauthoryear{Fink and Gandhi}{Fink and Gandhi}{2011}]%
        {fink_compression_2011}
\bibfield{author}{\bibinfo{person}{Eugene Fink} {and}
  \bibinfo{person}{Harith~Suman Gandhi}.} \bibinfo{year}{2011}\natexlab{}.
\newblock \showarticletitle{Compression of time series by extracting major
  extrema}.
\newblock \bibinfo{journal}{\emph{JETAI}} \bibinfo{volume}{23},
  \bibinfo{number}{2} (\bibinfo{year}{2011}), \bibinfo{pages}{255--270}.
\newblock
\showISSN{0952-813X, 1362-3079}


\bibitem[\protect\citeauthoryear{Google}{Google}{2018}]%
        {google2018}
\bibfield{author}{\bibinfo{person}{Google}.} \bibinfo{year}{2018}\natexlab{}.
\newblock \bibinfo{title}{The Need for Mobile Speed: How Mobile Latency Impacts
  Publisher Revenue}.
\newblock \bibinfo{howpublished}{Google Research Whitepaper}.
\newblock


\bibitem[\protect\citeauthoryear{Guerra, Vinciguerra, Boffa, and
  Ferragina}{Guerra et~al\mbox{.}}{2025}]%
        {guerra2025neats}
\bibfield{author}{\bibinfo{person}{Andrea Guerra}, \bibinfo{person}{Giorgio
  Vinciguerra}, \bibinfo{person}{Antonio Boffa}, {and} \bibinfo{person}{Paolo
  Ferragina}.} \bibinfo{year}{2025}\natexlab{}.
\newblock \showarticletitle{Learned Compression of Nonlinear Time Series with
  Random Access}. In \bibinfo{booktitle}{\emph{ICDE}}.
  \bibinfo{pages}{1579--1592}.
\newblock


\bibitem[\protect\citeauthoryear{Gómez-Brandón, Paramá, Villalobos,
  Illarramendi, and Brisaboa}{Gómez-Brandón et~al\mbox{.}}{2021}]%
        {gomez-brandon_lossless_2021}
\bibfield{author}{\bibinfo{person}{Adrián Gómez-Brandón},
  \bibinfo{person}{José~R. Paramá}, \bibinfo{person}{Kevin Villalobos},
  \bibinfo{person}{Arantza Illarramendi}, {and} \bibinfo{person}{Nieves~R.
  Brisaboa}.} \bibinfo{year}{2021}\natexlab{}.
\newblock \showarticletitle{Lossless compression of industrial time series with
  direct access}.
\newblock \bibinfo{journal}{\emph{Computers in Industry}}
  \bibinfo{volume}{132} (\bibinfo{year}{2021}), \bibinfo{pages}{103503}.
\newblock
\showISSN{01663615}


\bibitem[\protect\citeauthoryear{{Huamin Chen}, {Jian Li}, and
  Mohapatra}{{Huamin Chen} et~al\mbox{.}}{2004}]%
        {huamin_chen_race_2004}
\bibfield{author}{\bibinfo{person}{{Huamin Chen}}, \bibinfo{person}{{Jian Li}},
  {and} \bibinfo{person}{P. Mohapatra}.} \bibinfo{year}{2004}\natexlab{}.
\newblock \showarticletitle{{RACE}: time series compression with rate
  adaptivity and error bound for sensor networks}. In
  \bibinfo{booktitle}{\emph{MASS}}. \bibinfo{pages}{124--133}.
\newblock
\showISBNx{978-0-7803-8815-4}


\bibitem[\protect\citeauthoryear{Hwang, Kim, Kim, and Kwak}{Hwang
  et~al\mbox{.}}{2023}]%
        {hwang_lossless_2023}
\bibfield{author}{\bibinfo{person}{Sang-Ho Hwang}, \bibinfo{person}{Kyung-Min
  Kim}, \bibinfo{person}{Sungho Kim}, {and} \bibinfo{person}{Jong~Wook Kwak}.}
  \bibinfo{year}{2023}\natexlab{}.
\newblock \showarticletitle{Lossless {Data} {Compression} for {Time}-{Series}
  {Sensor} {Data} {Based} on {Dynamic} {Bit} {Packing}}.
\newblock \bibinfo{journal}{\emph{Sensors}} \bibinfo{volume}{23},
  \bibinfo{number}{20} (\bibinfo{year}{2023}), \bibinfo{pages}{8575}.
\newblock
\showISSN{1424-8220}


\bibitem[\protect\citeauthoryear{Jensen, Pedersen, and Thomsen}{Jensen
  et~al\mbox{.}}{2017}]%
        {Jensen_Pedersen_Thomsen_2017}
\bibfield{author}{\bibinfo{person}{S{\o}ren~Kejser Jensen},
  \bibinfo{person}{Torben~Bach Pedersen}, {and} \bibinfo{person}{Christian
  Thomsen}.} \bibinfo{year}{2017}\natexlab{}.
\newblock \showarticletitle{Time series management systems: A survey}.
\newblock \bibinfo{journal}{\emph{TKDE}} \bibinfo{volume}{29},
  \bibinfo{number}{11} (\bibinfo{year}{2017}), \bibinfo{pages}{2581--2600}.
\newblock


\bibitem[\protect\citeauthoryear{Jiang, Xiang, Wang, and Zheng}{Jiang
  et~al\mbox{.}}{2023}]%
        {jiang_time_2023}
\bibfield{author}{\bibinfo{person}{Nan Jiang}, \bibinfo{person}{Qingping
  Xiang}, \bibinfo{person}{Hongzhi Wang}, {and} \bibinfo{person}{Bo Zheng}.}
  \bibinfo{year}{2023}\natexlab{}.
\newblock \showarticletitle{Time series compression based on reinforcement
  learning}.
\newblock \bibinfo{journal}{\emph{Information Sciences}}  \bibinfo{volume}{648}
  (\bibinfo{year}{2023}), \bibinfo{pages}{119490}.
\newblock
\showISSN{00200255}


\bibitem[\protect\citeauthoryear{Khelifati, Khayati, and
  Cudre-Mauroux}{Khelifati et~al\mbox{.}}{2019}]%
        {khelifati_corad_2019}
\bibfield{author}{\bibinfo{person}{Abdelouahab Khelifati},
  \bibinfo{person}{Mourad Khayati}, {and} \bibinfo{person}{Philippe
  Cudre-Mauroux}.} \bibinfo{year}{2019}\natexlab{}.
\newblock \showarticletitle{{CORAD}: {Correlation}-{Aware} {Compression} of
  {Massive} {Time} {Series} using {Sparse} {Dictionary} {Coding}}. In
  \bibinfo{booktitle}{\emph{Big Data}}. \bibinfo{pages}{2289--2298}.
\newblock
\showISBNx{978-1-72810-858-2}


\bibitem[\protect\citeauthoryear{Klus, Klus, Lohan, Granell, Talvitie, Valkama,
  and Nurmi}{Klus et~al\mbox{.}}{2021}]%
        {klus_direct_2021}
\bibfield{author}{\bibinfo{person}{Lucie Klus}, \bibinfo{person}{Roman Klus},
  \bibinfo{person}{Elena~Simona Lohan}, \bibinfo{person}{Carlos Granell},
  \bibinfo{person}{Jukka Talvitie}, \bibinfo{person}{Mikko Valkama}, {and}
  \bibinfo{person}{Jari Nurmi}.} \bibinfo{year}{2021}\natexlab{}.
\newblock \showarticletitle{Direct {Lightweight} {Temporal} {Compression} for
  {Wearable} {Sensor} {Data}}.
\newblock \bibinfo{journal}{\emph{Sensors Letters}} \bibinfo{volume}{5},
  \bibinfo{number}{2} (\bibinfo{year}{2021}), \bibinfo{pages}{1--4}.
\newblock
\showISSN{2475-1472}


\bibitem[\protect\citeauthoryear{Kuschewski, Sauerwein, Alhomssi, and
  Leis}{Kuschewski et~al\mbox{.}}{2023}]%
        {kuschewski_btrblocks_2023}
\bibfield{author}{\bibinfo{person}{Maximilian Kuschewski},
  \bibinfo{person}{David Sauerwein}, \bibinfo{person}{Adnan Alhomssi}, {and}
  \bibinfo{person}{Viktor Leis}.} \bibinfo{year}{2023}\natexlab{}.
\newblock \showarticletitle{{BtrBlocks}: {Efficient} {Columnar} {Compression}
  for {Data} {Lakes}}.
\newblock \bibinfo{journal}{\emph{SIGMOD}} \bibinfo{volume}{1},
  \bibinfo{number}{2} (\bibinfo{year}{2023}), \bibinfo{pages}{1--26}.
\newblock
\showISSN{2836-6573}


\bibitem[\protect\citeauthoryear{Li, Lu, Jensen, Tang, and Cheema}{Li
  et~al\mbox{.}}{2022}]%
        {li2022spatial}
\bibfield{author}{\bibinfo{person}{Huan Li}, \bibinfo{person}{Hua Lu},
  \bibinfo{person}{Christian~S Jensen}, \bibinfo{person}{Bo Tang}, {and}
  \bibinfo{person}{Muhammad~Aamir Cheema}.} \bibinfo{year}{2022}\natexlab{}.
\newblock \showarticletitle{Spatial data quality in the Internet of Things:
  Management, exploitation, and prospects}.
\newblock \bibinfo{journal}{\emph{CSUR}} \bibinfo{volume}{55},
  \bibinfo{number}{3} (\bibinfo{year}{2022}), \bibinfo{pages}{1--41}.
\newblock


\bibitem[\protect\citeauthoryear{Li, Chen, Lu, Xu, Yang, Chen, Bao, and
  Zheng}{Li et~al\mbox{.}}{2025a}]%
        {li2025serf}
\bibfield{author}{\bibinfo{person}{Ruiyuan Li}, \bibinfo{person}{Zechao Chen},
  \bibinfo{person}{Ruyun Lu}, \bibinfo{person}{Xiaolong Xu},
  \bibinfo{person}{Guangchao Yang}, \bibinfo{person}{Chao Chen},
  \bibinfo{person}{Jie Bao}, {and} \bibinfo{person}{Yu Zheng}.}
  \bibinfo{year}{2025}\natexlab{a}.
\newblock \showarticletitle{Serf: Streaming Error-Bounded Floating-Point
  Compression}.
\newblock \bibinfo{journal}{\emph{SIGMOD}} \bibinfo{volume}{3},
  \bibinfo{number}{3} (\bibinfo{year}{2025}), \bibinfo{pages}{1--27}.
\newblock


\bibitem[\protect\citeauthoryear{Li, Li, Wu, Chen, Guo, Zhang, and Zheng}{Li
  et~al\mbox{.}}{2023b}]%
        {Elf_Plus_2023}
\bibfield{author}{\bibinfo{person}{Ruiyuan Li}, \bibinfo{person}{Zheng Li},
  \bibinfo{person}{Yi Wu}, \bibinfo{person}{Chao Chen},
  \bibinfo{person}{Songtao Guo}, \bibinfo{person}{Ming Zhang}, {and}
  \bibinfo{person}{Yu Zheng}.} \bibinfo{year}{2023}\natexlab{b}.
\newblock \showarticletitle{Erasing-based lossless compression method for
  streaming floating-point time series}.
\newblock \bibinfo{journal}{\emph{ArXiv}}  \bibinfo{volume}{abs/2306.16053}
  (\bibinfo{year}{2023}).
\newblock


\bibitem[\protect\citeauthoryear{Li, Li, Wu, Chen, Liu, and Zheng}{Li
  et~al\mbox{.}}{2025b}]%
        {li2025adaptive}
\bibfield{author}{\bibinfo{person}{Ruiyuan Li}, \bibinfo{person}{Zheng Li},
  \bibinfo{person}{Yi Wu}, \bibinfo{person}{Chao Chen}, \bibinfo{person}{Tong
  Liu}, {and} \bibinfo{person}{Yu Zheng}.} \bibinfo{year}{2025}\natexlab{b}.
\newblock \showarticletitle{Adaptive Encoding Strategies for Lossless
  Floating-Point Compression}.
\newblock \bibinfo{journal}{\emph{IoTJ}} (\bibinfo{year}{2025}),
  \bibinfo{pages}{1--1}.
\newblock


\bibitem[\protect\citeauthoryear{Li, Li, Wu, Chen, and Zheng}{Li
  et~al\mbox{.}}{2023a}]%
        {li_elf_2023}
\bibfield{author}{\bibinfo{person}{Ruiyuan Li}, \bibinfo{person}{Zheng Li},
  \bibinfo{person}{Yi Wu}, \bibinfo{person}{Chao Chen}, {and}
  \bibinfo{person}{Yu Zheng}.} \bibinfo{year}{2023}\natexlab{a}.
\newblock \showarticletitle{Elf: Erasing-Based Lossless Floating-Point
  Compression}.
\newblock \bibinfo{journal}{\emph{PVLDB}} \bibinfo{volume}{16},
  \bibinfo{number}{7} (\bibinfo{year}{2023}), \bibinfo{pages}{1763--1776}.
\newblock


\bibitem[\protect\citeauthoryear{Liakos, Papakonstantinopoulou, and
  Kotidis}{Liakos et~al\mbox{.}}{2022}]%
        {liakos_chimp_2022}
\bibfield{author}{\bibinfo{person}{Panagiotis Liakos}, \bibinfo{person}{Katia
  Papakonstantinopoulou}, {and} \bibinfo{person}{Yannis Kotidis}.}
  \bibinfo{year}{2022}\natexlab{}.
\newblock \showarticletitle{Chimp: efficient lossless floating point
  compression for time series databases}.
\newblock \bibinfo{journal}{\emph{PVLDB}} \bibinfo{volume}{15},
  \bibinfo{number}{11} (\bibinfo{year}{2022}), \bibinfo{pages}{3058--3070}.
\newblock
\showISSN{2150-8097}


\bibitem[\protect\citeauthoryear{Liu, Jiang, Paparrizos, and Elmore}{Liu
  et~al\mbox{.}}{2021}]%
        {liu2021Buff}
\bibfield{author}{\bibinfo{person}{Chunwei Liu}, \bibinfo{person}{Hao Jiang},
  \bibinfo{person}{John Paparrizos}, {and} \bibinfo{person}{Aaron~J Elmore}.}
  \bibinfo{year}{2021}\natexlab{}.
\newblock \showarticletitle{Decomposed bounded floats for fast compression and
  queries}.
\newblock \bibinfo{journal}{\emph{PVLDB}} \bibinfo{volume}{14},
  \bibinfo{number}{11} (\bibinfo{year}{2021}), \bibinfo{pages}{2586--2598}.
\newblock


\bibitem[\protect\citeauthoryear{Liu, Djukic, Kulhandjian, and Kantarci}{Liu
  et~al\mbox{.}}{2024}]%
        {liu_deep_2024}
\bibfield{author}{\bibinfo{person}{Jinxin Liu}, \bibinfo{person}{Petar Djukic},
  \bibinfo{person}{Michel Kulhandjian}, {and} \bibinfo{person}{Burak
  Kantarci}.} \bibinfo{year}{2024}\natexlab{}.
\newblock \showarticletitle{Deep {Dict}: {Deep} {Learning}-based {Lossy} {Time}
  {Series} {Compressor} for {IoT} {Data}}. In \bibinfo{booktitle}{\emph{ICC}}.
  \bibinfo{pages}{1--6}.
\newblock


\bibitem[\protect\citeauthoryear{Maneas, Mahdaviani, Emami, and
  Schroeder}{Maneas et~al\mbox{.}}{2022}]%
        {maneas2022ssd}
\bibfield{author}{\bibinfo{person}{Stathis Maneas}, \bibinfo{person}{Kaveh
  Mahdaviani}, \bibinfo{person}{Tim Emami}, {and} \bibinfo{person}{Bianca
  Schroeder}.} \bibinfo{year}{2022}\natexlab{}.
\newblock \showarticletitle{Operational Characteristics of SSDs in Enterprise
  Storage Systems: A Large-Scale Field Study}. In
  \bibinfo{booktitle}{\emph{USENIX}}. \bibinfo{pages}{165--180}.
\newblock


\bibitem[\protect\citeauthoryear{Pan, Wang, Wu, and Wang}{Pan
  et~al\mbox{.}}{2018}]%
        {hartmann_mtsc_2018}
\bibfield{author}{\bibinfo{person}{Ningting Pan}, \bibinfo{person}{Peng Wang},
  \bibinfo{person}{Jiaye Wu}, {and} \bibinfo{person}{Wei Wang}.}
  \bibinfo{year}{2018}\natexlab{}.
\newblock \showarticletitle{MTSC: An Effective Multiple Time Series Compressing
  Approach}. In \bibinfo{booktitle}{\emph{DEXA}}. \bibinfo{pages}{267--282}.
\newblock


\bibitem[\protect\citeauthoryear{Pelkonen, Franklin, Teller, Cavallaro, Huang,
  Meza, and Veeraraghavan}{Pelkonen et~al\mbox{.}}{2015}]%
        {pelkonen_gorilla_2015}
\bibfield{author}{\bibinfo{person}{Tuomas Pelkonen}, \bibinfo{person}{Scott
  Franklin}, \bibinfo{person}{Justin Teller}, \bibinfo{person}{Paul Cavallaro},
  \bibinfo{person}{Qi Huang}, \bibinfo{person}{Justin Meza}, {and}
  \bibinfo{person}{Kaushik Veeraraghavan}.} \bibinfo{year}{2015}\natexlab{}.
\newblock \showarticletitle{Gorilla: a fast, scalable, in-memory time series
  database}.
\newblock \bibinfo{journal}{\emph{PVLDB}} \bibinfo{volume}{8},
  \bibinfo{number}{12} (\bibinfo{year}{2015}), \bibinfo{pages}{1816--1827}.
\newblock
\showISSN{2150-8097}


\bibitem[\protect\citeauthoryear{Wang, Huang, Qiao, Jiang, Rui, Zhang, Kang,
  Feinauer, McGrail, Wang, Luo, Yuan, Wang, and Sun}{Wang
  et~al\mbox{.}}{2020}]%
        {wang_apache_2020}
\bibfield{author}{\bibinfo{person}{Chen Wang}, \bibinfo{person}{Xiangdong
  Huang}, \bibinfo{person}{Jialin Qiao}, \bibinfo{person}{Tian Jiang},
  \bibinfo{person}{Lei Rui}, \bibinfo{person}{Jinrui Zhang},
  \bibinfo{person}{Rong Kang}, \bibinfo{person}{Julian Feinauer},
  \bibinfo{person}{Kevin~A. McGrail}, \bibinfo{person}{Peng Wang},
  \bibinfo{person}{Diaohan Luo}, \bibinfo{person}{Jun Yuan},
  \bibinfo{person}{Jianmin Wang}, {and} \bibinfo{person}{Jiaguang Sun}.}
  \bibinfo{year}{2020}\natexlab{}.
\newblock \showarticletitle{Apache {IoTDB}: time-series database for internet
  of things}.
\newblock \bibinfo{journal}{\emph{PVLDB}} \bibinfo{volume}{13},
  \bibinfo{number}{12} (\bibinfo{year}{2020}), \bibinfo{pages}{2901--2904}.
\newblock
\showISSN{2150-8097}


\bibitem[\protect\citeauthoryear{Wang and Deng}{Wang and Deng}{2020}]%
        {wang2020deltapq}
\bibfield{author}{\bibinfo{person}{Runhui Wang} {and} \bibinfo{person}{Dong
  Deng}.} \bibinfo{year}{2020}\natexlab{}.
\newblock \showarticletitle{DeltaPQ: lossless product quantization code
  compression for high dimensional similarity search}.
\newblock \bibinfo{journal}{\emph{PVLDB}} \bibinfo{volume}{13},
  \bibinfo{number}{13} (\bibinfo{year}{2020}), \bibinfo{pages}{3603--3616}.
\newblock


\bibitem[\protect\citeauthoryear{Xia, Xiao, Huang, Hu, Song, Huang, and
  Wang}{Xia et~al\mbox{.}}{2024}]%
        {xia_time_2024}
\bibfield{author}{\bibinfo{person}{Tianrui Xia}, \bibinfo{person}{Jinzhao
  Xiao}, \bibinfo{person}{Yuxiang Huang}, \bibinfo{person}{Changyu Hu},
  \bibinfo{person}{Shaoxu Song}, \bibinfo{person}{Xiangdong Huang}, {and}
  \bibinfo{person}{Jianmin Wang}.} \bibinfo{year}{2024}\natexlab{}.
\newblock \showarticletitle{Time series data encoding in {Apache} {IoTDB}:
  comparative analysis and recommendation}.
\newblock \bibinfo{journal}{\emph{VLDBJ}} \bibinfo{volume}{33},
  \bibinfo{number}{3} (\bibinfo{year}{2024}), \bibinfo{pages}{727--752}.
\newblock
\showISSN{1066-8888, 0949-877X}


\bibitem[\protect\citeauthoryear{Yan, Han, Xu, and Li}{Yan
  et~al\mbox{.}}{2021}]%
        {yan_model-free_2021}
\bibfield{author}{\bibinfo{person}{Lei Yan}, \bibinfo{person}{Jiayu Han},
  \bibinfo{person}{Runnan Xu}, {and} \bibinfo{person}{Zuyi Li}.}
  \bibinfo{year}{2021}\natexlab{}.
\newblock \showarticletitle{Model-Free Lossless Data Compression for Real-Time
  Low-Latency Transmission in Smart Grids}.
\newblock \bibinfo{journal}{\emph{TSG}} \bibinfo{volume}{12},
  \bibinfo{number}{3} (\bibinfo{year}{2021}), \bibinfo{pages}{2601--2610}.
\newblock
\showISSN{1949-3053, 1949-3061}


\bibitem[\protect\citeauthoryear{Yang, Tao, Dutkiewicz, Huang, Guo, and
  Cui}{Yang et~al\mbox{.}}{2013}]%
        {energy-efficient}
\bibfield{author}{\bibinfo{person}{Xianjun Yang}, \bibinfo{person}{Xiaofeng
  Tao}, \bibinfo{person}{Eryk Dutkiewicz}, \bibinfo{person}{Xiaojing Huang},
  \bibinfo{person}{Y.~Jay Guo}, {and} \bibinfo{person}{Qimei Cui}.}
  \bibinfo{year}{2013}\natexlab{}.
\newblock \showarticletitle{Energy-Efficient Distributed Data Storage for
  Wireless Sensor Networks Based on Compressed Sensing and Network Coding}.
\newblock \bibinfo{journal}{\emph{IEEE TWC}} \bibinfo{volume}{12},
  \bibinfo{number}{10} (\bibinfo{year}{2013}), \bibinfo{pages}{5087--5099}.
\newblock


\bibitem[\protect\citeauthoryear{Yang and Chen}{Yang and Chen}{2023}]%
        {yang_most_2023}
\bibfield{author}{\bibinfo{person}{Zehai Yang} {and} \bibinfo{person}{Shimin
  Chen}.} \bibinfo{year}{2023}\natexlab{}.
\newblock \showarticletitle{{MOST}: {Model}-{Based} {Compression} with
  {Outlier} {Storage} for {Time} {Series} {Data}}.
\newblock \bibinfo{journal}{\emph{SIGMOD}} \bibinfo{volume}{1},
  \bibinfo{number}{4} (\bibinfo{year}{2023}), \bibinfo{pages}{1--29}.
\newblock
\showISSN{2836-6573}


\bibitem[\protect\citeauthoryear{Yao, Chen, Fang, Gao, Jensen, and Li}{Yao
  et~al\mbox{.}}{2024}]%
        {camel}
\bibfield{author}{\bibinfo{person}{Yuanyuan Yao}, \bibinfo{person}{Lu Chen},
  \bibinfo{person}{Ziquan Fang}, \bibinfo{person}{Yunjun Gao},
  \bibinfo{person}{Christian~S. Jensen}, {and} \bibinfo{person}{Tianyi Li}.}
  \bibinfo{year}{2024}\natexlab{}.
\newblock \showarticletitle{Camel: Efficient Compression of Floating-Point Time
  Series}.
\newblock \bibinfo{journal}{\emph{SIGMOD}} \bibinfo{volume}{2},
  \bibinfo{number}{6} (\bibinfo{year}{2024}), \bibinfo{pages}{1--26}.
\newblock


\bibitem[\protect\citeauthoryear{Yu, Peng, Li, Wang, Shen, Mai, and Xie}{Yu
  et~al\mbox{.}}{2020}]%
        {yu_two-level_2020}
\bibfield{author}{\bibinfo{person}{Xinyang Yu}, \bibinfo{person}{Yanqing Peng},
  \bibinfo{person}{Feifei Li}, \bibinfo{person}{Sheng Wang},
  \bibinfo{person}{Xiaowei Shen}, \bibinfo{person}{Huijun Mai}, {and}
  \bibinfo{person}{Yue Xie}.} \bibinfo{year}{2020}\natexlab{}.
\newblock \showarticletitle{Two-Level Data Compression using Machine Learning
  in Time Series Database}. In \bibinfo{booktitle}{\emph{ICDE}}.
  \bibinfo{pages}{1333--1344}.
\newblock


\end{thebibliography}

\clearpage
\appendix
\section{Supporting Theory and Proofs}
\label{sec:theory}

\subsection{Loss of Smoothness in \Elf{}}
\label{ssec:elf_loss}

We provide a formula to quantify the \textbf{smoothness loss} after conversion:
\begin{equation}
    \hat{q} = \min(v_{X}.q, v_{Y}.q),
\end{equation}
\begin{equation}
S(v_{X}, v_{Y}) = \text{CBL}\left(\operatorname{abs}((v_{X} \times 10^{-\hat{q}}) - (v_{Y} \times 10^{-\hat{q}}))\right),
\end{equation}

The smoothness loss is defined as:  
\begin{equation}
    \text{Loss} = S(\hat{v}_{X}, \hat{v}_{Y}) - S(v_{X}, v_{Y}).
\end{equation}  

Certain algorithms, such as \Elf{}, introduce precision-related errors that disrupt smoothness. Specifically, for \Elf{}, the loss can be expressed as:  
\begin{equation}
\text{Loss}(\Elf{}) = S(v_{X} - \delta_{v_{X}}, v_{Y} - \delta_{v_{Y}}) - S(v_{X}, v_{Y}).
\end{equation} 

Here, $\delta_{v_{X}}$ denotes the discrepancy between the value after erasure and the original value, which may be represented as:
\begin{equation} 
\delta_{v_{X}} = sign \times 2^{(\mathit{exp}-\mathit{bias})}\times \mathit{fraction}_{(\text{lower}~g(v_{X})~\text{bits})},
\end{equation}
where the lower $g(v_{X})$ bits of the original value get erased. 

This concept, originally defined by the \Elf{} authors, can be described in the language of this paper as:
\begin{equation} 
g(v_{X}) = 52 - (\lceil(-q)log_2(10)\rceil + \mathit{exp} - \mathit{bias}).
\end{equation}

When \XOR{} operations are applied, precision mismatches occur. Specifically, when $g(v_{X}) \neq g(v_{Y})$, it follows that $\delta_{v_{X}} \neq \delta_{v_{Y}}$. As a result, the \Elf{} algorithm can disrupt up to $\operatorname{abs}\big(g(v_{X}) - g(v_{Y})\big)$ bits of already-erased CBL.  

Although \Elf{} attempts to combine redundancy elimination and smoothness exploitation, examples illustrate that performing \XOR{} after erasure leads to poor results. Reversing the order (i.e., performing erasure after \XOR{}) is equally problematic. Moreover, precision becomes uncontrollable after \XOR{} operations, making such optimizations challenging to implement effectively.

\subsection{Zero Loss of \DXOR{} Converter}  
\label{ssec:smoothness_loss}  

We formally prove that the \DXOR{} converter incurs no loss of smoothness.  

\begin{lemma}[Zero Loss of \DXOR{}]  
    The smoothness loss of \DXOR{} satisfies:  
    $\forall \uline{v_{X}}, v_{Y} \in \mathbb{R}, \text{Loss}(\uline{v_{X}} \diamond v_{Y}) = 0.$ 
\end{lemma}  

\begin{proof}  
According to the preconditions defined in Section~\ref{sssec:preconditions}, the \DXOR{} operation produces:  
\begin{equation}  
\hat{v}_{X} = \uline{v_{X}} \diamond v_{Y} = v_{X} - \alpha, \quad \hat{v}_{Y} = v_{Y} - \alpha,  
\end{equation}    
where $\alpha$ denotes the shared prefix between $v_{X}$ and $v_{Y}$.  

Let $\hat{q} = \min(v_{X}.q, v_{Y}.q)$. 
Next, we compute the smoothness metric $S(\hat{v}_{X}, \hat{v}_{Y})$,
\begin{align*}  
S(\hat{v}_{X}, \hat{v}_{Y})  
&= \text{CBL}\left(\operatorname{abs}((\hat{v}_{X} \times 10^{-\hat{q}}) - (\hat{v}_{Y} \times 10^{-\hat{q}}))\right) \\  
&= \text{CBL}\left(\operatorname{abs}((v_{X} \times 10^{-\hat{q}} + \gamma) - (v_{Y} \times 10^{-\hat{q}} + \gamma))\right),  
\end{align*}  
where $\gamma = -\alpha \times 10^{-\hat{q}}$.  

Since the $\gamma$ cancels out due to the shared prefix, we have:  
\begin{equation}   
S(\hat{v}_{X}, \hat{v}_{Y}) = \text{CBL}\left(\operatorname{abs}((v_{X} \times 10^{-\hat{q}}) - (v_{Y} \times 10^{-\hat{q}}))\right) = S(v_{X}, v_{Y}).  
\end{equation} 

Thus, the smoothness loss is:  
\begin{equation}    
\text{Loss}(v_{X} \diamond v_{Y}) = S(\hat{v}_{X}, \hat{v}_{Y}) - S(v_{X}, v_{Y}) = 0.  
\end{equation}

This completes the proof.  
\end{proof}  

\subsection{Proof of Fixed Bit Allocation for Unsigned Binary Suffix}
\label{ssec:proof_suffix}

In Section~\ref{ssec:compressor}, we introduce  Lemma~\ref{theo:upper_bound}, which asserts that for any $\beta_i \in \mathbb{Z}$, a fixed allocation of $\bar{\ell}_i$ bits achieves better compression than variable allocation of $\ell_i$.
Specifically, it holds that:
$$\mathbb{E}[(4 + \bar{\ell}_i)] < \mathbb{E}[(6 + \ell_i)].$$
The detailed proof is provided below:

\begin{proof}

We know the condition:
$$
\delta = o_i - q_i \in \mathbb{N}, \quad \operatorname{abs}(\beta_i) \in [10^{\delta-1}, 10^{\delta}),
$$
$$
\ell_i = \lceil \log_2(\operatorname{abs}(\beta_i) + 1) \rceil, \quad \bar{\ell}_i = \lceil \log_2(10^{\delta}) \rceil.
$$

We aim to prove:
$$\mathbb{E}[(\bar{\ell}_i+4)] < \mathbb{E}[(6+\ell_i)] \iff \mathbb{E}[(\bar{\ell}_i - \ell_i - 2)] < 0.$$

We divide the range of $\operatorname{abs}(\beta_i)$ by powers of 2. Let $j \in \mathbb{N}$ be the smallest integer such that $2^j > 10^{\delta-1}$, implying:
$$2^{j-1} \leq \mathbf{10}^{\delta-1} < 2^j < 2^{j+1} < 2^{j+2} < \mathbf{10}^\delta < 2^{j+4}.$$
We now analyze the relationship between $2^{j+3}$ and $10^\delta$, which leads to two cases:

\textbf{Case (1): $2^{j+3} > 10^\delta$}, with probability $\mathbb{P}_1$. In this case, the fixed allocation is $\bar{\ell}_i = \lceil \log_2(10^\delta) \rceil = j + 3$, while the variable allocation $\ell_i$ is given by:
$$\ell_i =  
\begin{cases} 
j & \text{if } \operatorname{abs}(\beta_i) \in [10^{\delta-1},2^j), \\
j+1 & \text{if } \operatorname{abs}(\beta_i) \in [2^j,2^{j+1}), \\
j+2 & \text{if } \operatorname{abs}(\beta_i) \in [2^{j+1},2^{j+2}), \\
j+3 & \text{if } \operatorname{abs}(\beta_i) \in [2^{j+2},10^{\delta}).
\end{cases} $$

To calculate the expectation for this case, we have:
$\mathbb{E}_1 = \mathbb{E}(\bar{\ell}_i+4 \mid 2^{j+3}>10^{\delta}) - \mathbb{E}(\ell_i+6 \mid 2^{j+3}>10^{\delta}) = j+1- \mathbb{E}(\ell_i \mid 2^{j+3}>10^{\delta}).$

Now, let's calculate $\mathbb{E}(\ell_i \mid 2^{j+3} > 10^\delta)$. This can be expressed as:
\begin{align*}
\mathbb{E}(\ell_i \mid 2^{j+3} > 10^{\delta}) 
&= \frac{(j+3)(10^\delta - 2^{j+2})}{10^\delta - 10^{\delta-1}} + \frac{(j+2)(2^{j+2} - 2^{j+1})}{10^\delta - 10^{\delta-1}} \\
&\quad + \frac{(j+1)(2^{j+1} - 2^{j})}{10^\delta - 10^{\delta-1}} + \frac{j(2^{j} - 10^{\delta-1})}{10^\delta - 10^{\delta-1}}.
\end{align*}

This expression can be further simplified as follows:
\begin{align*}
\mathbb{E}(\ell_i \mid 2^{j+3} > 10^{\delta}) 
&= j + \frac{3(10^\delta - 2^{j+2}) + 2(2^{j+2} - 2^{j+1}) + (2^{j+1} - 2^{j})}{10^\delta - 10^{\delta-1}} \\
&= j + \frac{3 \times 10^\delta - 2^{j+2} - 2^{j+1} - 2^{j}}{10^\delta - 10^{\delta-1}} \\
&= j + 3 - \frac{2^{j+2} + 2^{j+1} + 2^{j} - 3 \times 10^{\delta-1}}{10^\delta - 10^{\delta-1}}.
\end{align*}

Next, we approximate the terms as follows:
$$
\mathbb{E}(\ell_i \mid 2^{j+3} > 10^{\delta}) > j + 3 - \frac{2^{j+2} + 2^{j+1} + 2^{j} - 3 \times 2^{j-1}}{2^{j+2} - 2^{j}}.
$$

After simplifying the above expression, we obtain:
$$
\mathbb{E}(\ell_i \mid 2^{j+3} > 10^\delta) > j + \frac{7}{6}.
$$

Thus, the expectation simplifies to:
$$
\mathbb{E}_1 = j + 1 - \mathbb{E}(\ell_i \mid 2^{j+3} > 10^\delta) < -\frac{1}{6}.
$$


We now calculate the probability $\mathbb{P}_1$ for case (1):
$$\mathbb{P}_1\{2^{j+3} > 10^\delta\} = \mathbb{P}\{(j+3)\log_{10}2 > \delta\}.$$

According to the condition $2^j > 10^{\delta-1}$, we have:
$$j\log_{10}2+1 > \delta.$$

Similarly, from $2^{j-1} \leq 10^{\delta-1}$, this implies:
$$(j-1)\log_{10}2+1 \leq \delta.$$

Therefore, the value of $\delta$ lies within the range: 
$$\delta\in\big[(j-1)\log_{10}2+1,j\log_{10}2+1\big).$$

In summary, we can compute the possibility as follows:
\begin{align*}
& \mathbb{P}\{(j+3)\log_{10}2 > \delta\} 
 = \frac{(j+3)\log_{10}2 - (j-1)\log_{10}2-1}{\log_{10}2} \\
& \quad = \frac{4\log_{10}2-1}{\log_{10}2} 
\implies \mathbb{P}_1\{2^{j+3} > 10^\delta\} = 4 - \log_210 \approx 0.6781.
\end{align*}


\textbf{Case (2): $2^{j+3} \leq 10^\delta$}, with the probability given by:
$$\mathbb{P}_2 \{2^{j+3}\leq10^{\delta}\} = 1 -\mathbb{P}_1 \approx0.3219.$$ 

It is straightforward to see that for $ j \geq 0 $, $2^{j+3} \neq 10^\delta $, and thus we can replace the $2^{j+3} \leq 10^\delta $ with a strict inequality $2^{j+3} < 10^\delta $.

In this case, the fixed allocation is $\bar{\ell}_i = \lceil \log_2(10^\delta) \rceil = j + 4$, while the variable allocation $\ell_i$ is given by:
$$\ell_i =  
\begin{cases} 
j & \text{if } \operatorname{abs}(\beta_i) \in [10^{\delta-1},2^j), \\
j+1 & \text{if } \operatorname{abs}(\beta_i) \in [2^j,2^{j+1}), \\
j+2 & \text{if } \operatorname{abs}(\beta_i) \in [2^{j+1},2^{j+2}), \\
j+3 & \text{if } \operatorname{abs}(\beta_i) \in [2^{j+2},2^{j+3}), \\
j+4 & \text{if } \operatorname{abs}(\beta_i) \in [2^{j+3},10^{\delta}).
\end{cases} $$

We first compute the expectation as follows:
$\mathbb{E}_2 = \mathbb{E}(\bar{\ell}_i+4 \mid 2^{j+3}<10^{\delta}) - \mathbb{E}(\ell_i+6 \mid 2^{j+3}<10^{\delta}) = j+2- \mathbb{E}(\ell_i \mid 2^{j+3}<10^{\delta})$.

Next, we proceed to calculate the expectation $\mathbb{E}(\ell_i \mid 2^{j+3} < 10^{\delta})$ as follows:
\begin{align*}
\mathbb{E}(\ell_i \mid 2^{j+3} < 10^{\delta}) 
&= j + \frac{4 \times 10^\delta - 2^{j+3} - 2^{j+2} - 2^{j+1} - 2^{j}}{10^\delta - 10^{\delta-1}} \\
&= j + 4 - \frac{2^{j+3} + 2^{j+2} + 2^{j+1} + 2^{j} - 4 \times 10^{\delta-1}}{10^\delta - 10^{\delta-1}}.
\end{align*}

After approximating the terms, we obtain:
$$\mathbb{E}(\ell_i \mid 2^{j+3} < 10^{\delta}) > j + 4 - \frac{2^{j+3} + 2^{j+2} + 2^{j+1} + 2^{j} - 4 \times 2^{j-1}}{2^{j+3} - 2^{j}} = j + \frac{15}{7}.$$

Thus, the expectation simplifies to:
$$\mathbb{E}_2  = j+2 - \mathbb{E}(\ell_i|2^{j+3}<10^{\delta}) < -\frac{1}{7}.$$

Finally, we compute the overall expectation:
$$\mathbb{E}[(\bar{\ell}_i - \ell_i - 2)] = \mathbb{E}_1 \times \mathbb{P}_1 + \mathbb{E}_2 \times \mathbb{P}_2< -\frac{1}{6} \times \mathbb{P}_1 - \frac{1}{7} \times \mathbb{P}_2 \approx -0.159 < 0.$$
\end{proof}

\end{document}